\documentclass[9pt]{amsart}
\usepackage{graphicx}
\usepackage{verbatim}
\usepackage{bbm}
\usepackage{tikz}
\usepackage[latin1]{inputenc}

\usepackage{mathtools}

\usepackage{caption}
\usepackage{subcaption}

\usepackage[colorlinks,linkcolor = blue]{hyperref}

\newtheorem{prop}{Proposition}
\newtheorem{lemma}[prop]{Lemma}
\newtheorem{theorem}[prop]{Theorem}
\newtheorem{corollary}[prop]{Corollary}

\newtheorem{defi}{Definition}

\def\l@paragraph{\@tocline{4}{0pt}{1pc}{7pc}{}}

\def\cal{\mathcal}

\def\Z{{\mathbb Z}}
\def\N{{\mathbb N}}

\def\R{{\mathbb R}}
\def\E{{\mathbb E}}
\def\P{{\mathbb P}}

\def\var{\mathop{\operator@font var}\nolimits}
\def\cov{\mathop{\operator@font cov}\nolimits}
\def\Pois{\mathop{\operator@font Pois}\nolimits}

\def\steq#1{\stackrel{{\rm #1}}{=}}

\def\iM{1}\def\iP{2}

\newcommand{\diff}{\mathop{}\mathopen{}\mathrm{d}}
\newcommand\croc[1]{\left\langle #1\right\rangle}
\newcommand{\ind}[1]{\ensuremath{\mathbbm{1}_{\left\{#1\right\}}}}
\newcommand{\Var}{\mathop{}\mathopen{}\mathrm{Var}}
\newcommand{\Cov}{\mathop{}\mathopen{}\mathrm{Cov}}
\newcommand{\nocontentsline}[3]{}
\newcommand{\tocless}[3]{\bgroup\let\addcontentsline=\nocontentsline#1{#2\label{#3}}\egroup}

\graphicspath{{Fig/}}

\author{Philippe Robert}
\address{INRIA de Paris,
2 rue Simone Iff
CS 42112
75589 Paris Cedex 12
FRANCE}
\email{Philippe.Robert@inria.fr}
\urladdr{https://team.inria.fr/rap/robert}

\title[Mathematical Models of Gene Expression]{Mathematical Models of Gene Expression}
\date{\today}
\begin{document}

\begin{abstract}
In this paper we analyze the equilibrium properties of a large class of stochastic processes describing the fundamental biological  process within bacterial cells, {\em the production process of proteins}.  Stochastic models classically used in this context to describe the time evolution of the numbers of mRNAs and proteins are presented and discussed. An extension of these  models, which includes  elongation phases of mRNAs and proteins, is introduced.  A convergence result to equilibrium for the process associated to the number of proteins and mRNAs is proved and a representation of this  equilibrium  as a functional of a Poisson process in an extended state space is  obtained. Explicit  expressions  for the first two moments of the number of mRNAs and proteins at equilibrium are derived, generalizing some classical formulas.   Approximations used in the  biological literature for the equilibrium distribution of the number of proteins are discussed and investigated in the light of these results. Several convergence results for the distribution of the number of proteins at equilibrium are in particular obtained under different scaling assumptions. 
\end{abstract}

\maketitle

\bigskip
\hrule

\vspace{-3mm}

\tableofcontents

\vspace{-10mm}

\hrule

\bigskip

\newpage

\section{Introduction}\label{PIntroSec}
 The \emph{gene expression} is the process by which the genetic information is synthesized into a functional product, the proteins.  These macro{-}molecules are the crucial agents of functional properties of cells. They play an important role in most of the basic biological processes within the cell, either directly, or as a component of complex macro{-}molecules such as polymerases or ribosomes.  The information flow from DNA genes to proteins is a fundamental process, common to all living organisms.

We analyze this fundamental process in the context of prokaryotic cells, like bacterial  cells or archaeal cells.  The {\em cytoplasm} of these cells is not as structured as eukaryotic cells, like  mammalian cells for example, so that most of the macro{-}molecules of these cells can potentially collide with each other. This key biological process can be, roughly, described as resulting of multiple encounters/collisions of several types of macro{-}molecules of the cell: {\em polymerases} with DNA,  {\em ribosomes} with  {\em mRNAs}, or {\em proteins} with DNA,  \ldots An additional feature of this process is that it  is consuming an important fraction of energy resources of the cell, to build chains of amino{-}acids or  chains of nucleotides in particular. 

The fact that the cytoplasm of a bacterial cell is a disorganized medium  has important implications on the internal dynamics of these organisms. Numerous events  are triggered by random events associated  to thermal noise. When the external conditions are favorable, these cells can nevertheless multiply via division at a steady pace.  In each cell,  around 4000 different types of proteins are produced, with different concentrations: from 5 elements per cell for some rare proteins to 100,000 per cell for some ribosomal proteins. Despite of the noisy environment, the protein production process can handle such exponential growth with these constraints. The cost of the translation phase, the last step of the protein production process, is estimated to account for ~50\% of the energy consumption of a rapidly growing bacterial cell, see  Russell and Cook~\cite{Russell} for example.  

An important question in this context is of understanding and quantifying the role of the parameters of the cell in the production of proteins when large fluctuations may occur, as in the case of abundant proteins. Given the high level of energy consumed by protein production, these large stochastic fluctuations should be, at least, ``controlled'' to avoid a shortage of resources for example. Note nevertheless that, outside ``pure'' fluctuations, random events involving macro-molecules with low copy-numbers may play an important role in the time evolution of cells,  in the case of change of environment of the cell for example so that the cells may adapt. Some components of the ``noise'' in the gene expression may thus have a functional role.   See Eldar and Elowitz~\cite{eldar}, Lestas et al.~\cite{Lestas} and Norman et al.~\cite{Norman}.

From the point of view of the design, the cell can be seen as a system where the processing power is expressed in terms of numerous chemical reactions occurring in a noisy environment. An important difference with  classical, non-biological,  systems, like computer networks,  the possibility of {\em explicitly} storing information used to regulate the production process is not as clear as for these systems.  See Burrill and Silver~\cite{Burrill} on this topic.  This is a really challenging aspect of these biological systems.

\subsection{The Protein Production Process}
We give a quick, simplified, sketch of the steps involved in the production of proteins. See Watson et al.~\cite{Watson} for a much more detailed description of these complex processes. 
\begin{itemize}
\item[]
\item {\em  Gene Activation/Deactivation}. \\
  A gene is a contiguous section of the DNA  associated to a functional property of the cell. When a gene is {\em active}, a macro{-}molecule of the cell, an {\em RNA polymerase}, one of the macro{-}molecules moving in the cytoplasm,  can possibly  bind to the DNA at the promoter, i.e.\ the beginning of the section of DNA for the gene. The binding to the gene is subject to the various (random) fluctuations  of the medium. The gene may also be {\em inactive} and in this case no binding is possible. The reason is that some macro{-}molecules, like proteins or mRNAs, of the cell can ``block'', via chemical bonds, the gene so that  polymerases cannot access it. This is one of the regulation mechanisms of the cell. When random perturbations  break these bindings, the gene becomes again active.

\item[]
\item {\em Transcription}.  
  When an RNA polymerase is bound to an active gene, it starts to make a copy of this gene. The   product which is a sequence of nucleotides is a {\em messenger RNA}, or {\em mRNA}.  The time during which nucleotides are added sequentially  is the {\em elongation} phase of the production of the mRNA.  When the full sequence of nucleotides of the mRNA has been  successively assembled,  the mRNA is released in the cytoplasm. It has a finite lifetime, being degraded by other macro{-}molecules.
\item[]
\item {\em Translation}. The step is achieved through another large macro{-}molecule: the {\em ribosome}. A ribosome is also moving within the cytoplasm. When it encounters an mRNA, it can also bind to it via chemical reactions. In this case it builds a chain of amino{-}acids  using the mRNA as a template to produce a {\em  protein}.   This is the {\em elongation} phase of the production of the protein.  As for mRNAs,  proteins have also a finite lifetime.  See Figure~\ref{Fig1}.
\end{itemize}

\begin{figure}[ht]
        \begin{center}
\scalebox{0.8}{
                \begin{tikzpicture}[node distance=4.5cm,->,>=stealth,thick]
                  \node  (active)  {\includegraphics[width=2cm]{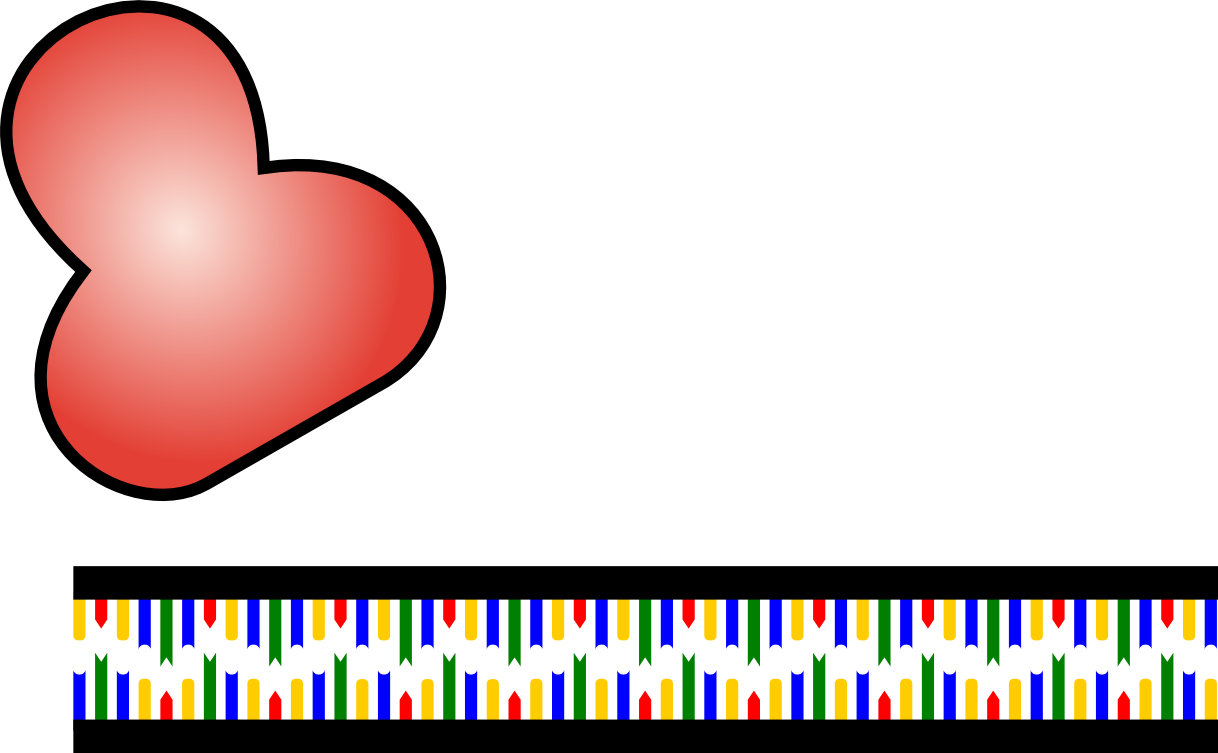}};
                  \node (name-active) [below of=active, node distance = 1cm]{Active Gene};
                \node (name-active2) [above of=active, right of=active, node distance = 18mm]{Polymerase};
                                \node (inactive) [below of=active, node distance = 3cm]   {\includegraphics[width=2cm]{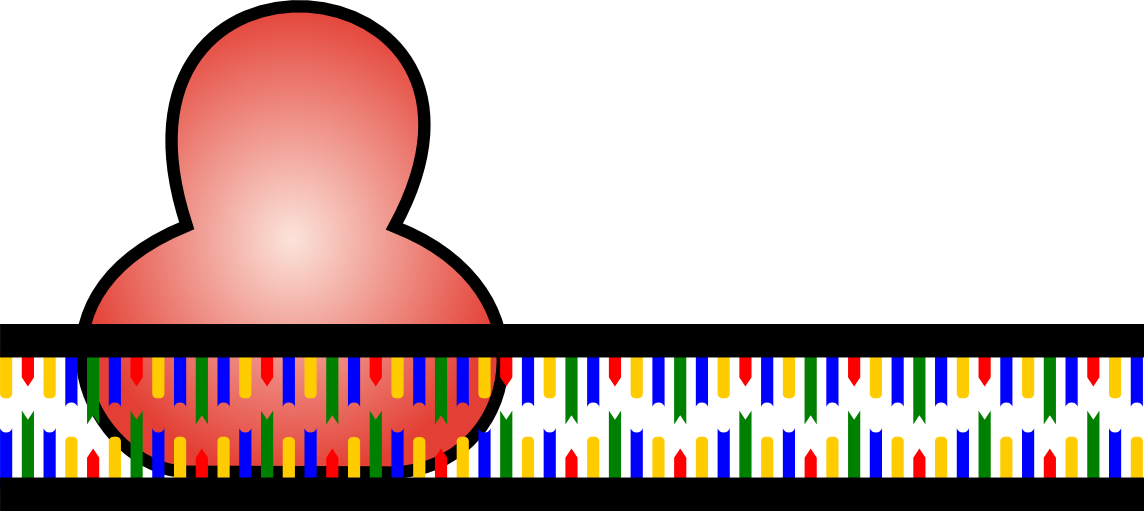}};
                  \node (name-inactive) [below of=inactive, node distance = 8mm]{Inactive Gene};
                \node (RNA) [right of=active] {\includegraphics[width=2cm]{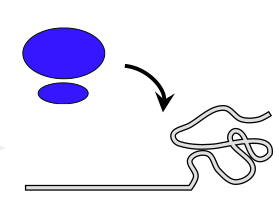}};
                \node (name-RNA) [above of=RNA, left of=RNA,  node distance = 8mm]{Ribosome};
                \node (name-RNA2) [below of=RNA,  right of=RNA, node distance = 9mm]{mRNA};
               \node (deadRNA) [below of=RNA, node distance = 2cm] {\includegraphics[width=1cm]{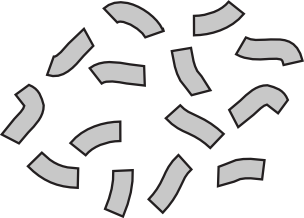}};
                \node (Prot) [right of=RNA] {\includegraphics[width=2cm,]{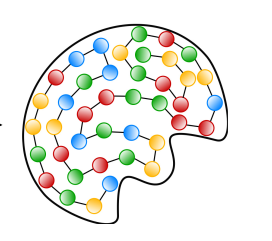}};
                \node (name-RNA) [below of=Prot, right of=Prot,  node distance = 7mm]{Protein};
                \node (deadProt) [below of=Prot, node distance = 2cm] {\includegraphics[width=1cm]{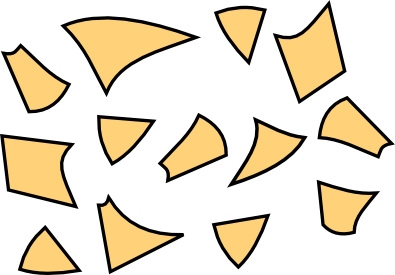}};
                \path   (active.south east) edge node [right] {} (inactive.north east)
                (inactive.north west) edge node [left] {} (active.south west);
                \path   (RNA) edge node [right] {} (deadRNA);
                \path   (Prot) edge node [right] {} (deadProt);
                \path (active) edge node [below] {} node[above=0.2cm]{\includegraphics[width=1.5cm]{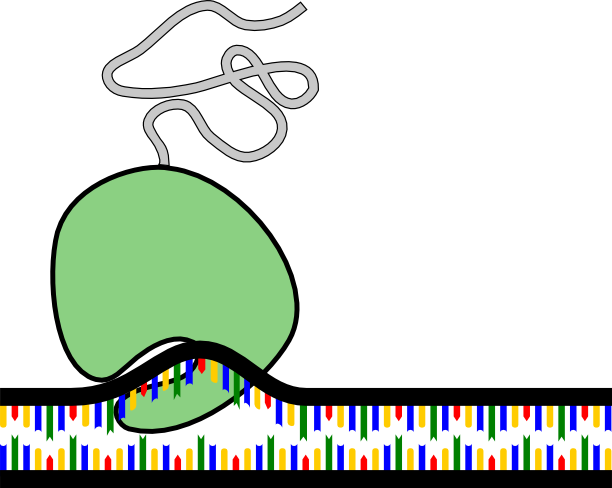}}  (RNA);
                \path (RNA) edge node [below] {} node[above=0.2cm]{\includegraphics[width=2cm]{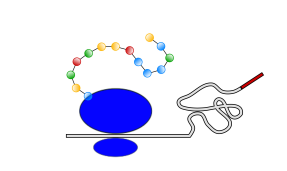}} (Prot);
\end{tikzpicture}}
        \end{center}
           \caption{A Simplified View of the Protein Production Process}\label{Fig1}
  \end{figure}
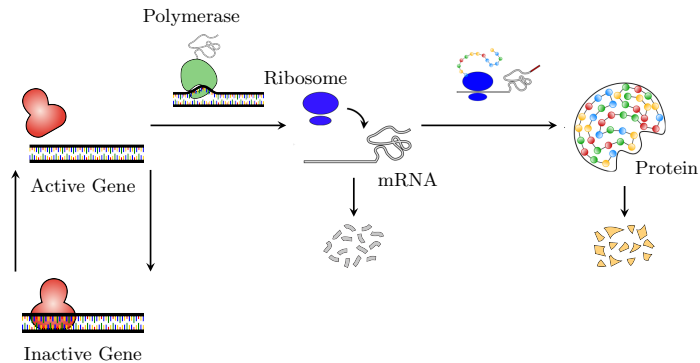

\subsection{Mathematical Models}\ 

The above description of the protein production process is clearly simplified. The  processes associated to actions of polymerases and ribosomes are composed of several steps. The way  ribosomes bind to mRNAs in particular, see Chapter~14 of Watson et al.~\cite{Watson} for example.  Once a polymerase binds to the gene, the messenger RNA chain is built through a series of specific stages, in which the polymerase recruits one of the four nucleotides in accordance to the DNA template.  Additionally, a dedicated proof reading mechanism takes place during this process.  There is a similar description for the translation step.  The binding events of polymerases to gene or of ribosomes to mRNAs are also due to a sequence of specific steps. 

Because of the disorganized medium of the bacteria, the protein production process is a highly stochastic process. The randomness is partially due to the thermal excitation of the environment. It drives the diffusion of the main components, mRNAs and ribosomes within the cytoplasm and  it also impacts the pairing of cellular components diffusing through the cytoplasm. It can cause the spontaneous rupture of such pairs,  before either transcription or translation can start.

The problem is of understanding the mechanisms used by the cell to produce a large number of proteins with very different concentrations in  such a random, ``noisy'', context. The main goals of a mathematical analysis in the biological literature are generally the following.
\begin{enumerate}
\item Estimate impact on variance of the number of proteins of assumptions on
  \begin{itemize}
  \item Activation/Deactivation rates of the gene;
  \item Transcription/Translation rates. Polymerases  may bind more easily to some genes. A similar phenomenon holds for ribosomes and mRNAs. This is generally mathematically  represented via a ``rate'' of binding: transcription rate for polymerases on gene and translation rate for ribosomes on mRNAs, $k_\iM$ and $k_\iP$ in the following;
  \item The distributions of lifetimes of mRNAs and proteins, the death rate of mRNAs, (resp., proteins),  is denoted by $\gamma_\iM$, (resp., $\gamma_\iP$). 
  \end{itemize}
  \begin{figure}[ht]
        \begin{center}
\scalebox{0.8}{
                \begin{tikzpicture}[node distance=2cm,->,>=stealth,thick]
                  \node  (active)  {DNA};
                \node (RNA) [right of=active] {mRNA};
                \node (Prot) [right of=RNA] {Protein};
                \node (deadProt) [below of=Prot, node distance = 1cm] {$\emptyset$};
                \node (deadRNA) [below of=RNA, node distance = 1cm] {$\emptyset$};
                \path   (active) edge node [above] {$k_\iM$} (RNA);
                \path   (RNA) edge node [above] {$k_\iP$} (Prot);
                \path   (RNA) edge node [right] {$\gamma_\iM$} (deadRNA);
                \path   (Prot) edge node [right] {$\gamma_\iP$} (deadProt);
\end{tikzpicture}}
        \end{center}
        \caption{Basic Chemical Reactions}
  \end{figure}
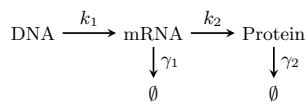
 \item Determine the parameters which achieve minimal variance of concentration of a type of protein with a fixed average concentration. When an analytical formula for the variance is available, one can determine, in principle, the parameters that minimize it, with the constraint of that its average is fixed. The general idea is of determining if the current parameters of the cell are adjusted, depending of the environment, to minimize this variance. 
\item[]\item Estimate  biological parameters for experiments.  Analytical formulas for the average and the variance of the number of proteins can also be used to estimate some biological parameters. See Taniguchi~\cite{Taniguchi}. 
\end{enumerate}
The protein production process  is a fundamental biological process, a process at the basis of ``life'', justifying the extensive efforts to understand and to quantify, via mathematical models, the role of its different  components: polymerases, ribosomes, mRNAs and proteins. It also serves as a generic example of a typical biological process: a component of type $A$ is produced with components of type $B$, but via the creation of another component of type $C$, which may also be used in the regulation of the process.

As a general remark, the mathematician should always keep in mind that she/he is studying an incredibly complex system, the molecular biology of gene expression. It has evolved over a time span of {\em four billion} years.  A mathematical model may take into account only partial aspects  of the sophisticated mechanisms into play. In this presentation we focus on the basic parameters of the transcription and translation phases. 

\bigskip
\noindent
{\bf Some history}.\\ 
The basic principles of molecular biology, starting with Crick and Watson~\cite{Crick}, Brenner et al.~\cite{Brenner} and Jacob and Monod~\cite{Jacob}, have been discovered in the 1950-1960's.  At that time, measuring  concentrations of  different types of macro{-}molecules within the cell was not really possible with  the experimental tools available. Efficient methods to estimate these concentrations, like the {\em fluorescence microscopy}, have been available much later  at the beginning of the 1990's.  This is probably the main motivation of mathematical models at the end the 1970's in the pioneering works Berg~\cite{Berg} and Rigney~\cite{Rigney} to study the fluctuations of the concentration of proteins within the cell.  Analytical expressions for the mean and the variance of the number of proteins in a cell have been derived for some simplified models.  These results were later extended  in the context of Markov chain theory. They are recalled in Section~\ref{MarkSec}. The formulas obtained can give some insight on the role  of  the parameters of the cell like the transcription rate of the gene or the average lifetime of some macro{-}molecules on these fluctuations. Interestingly, this is one of the few  examples in applied mathematics when mathematical models have been used {\em before} measurements could be done.   There is a huge literature on the stochastic analysis of this process using a large range of mathematical methods. For a detailed review of these works, see Paulsson~\cite{Paulsson}.  We focus here on exact formulas that can be obtained for specific stochastic models. 

\bigskip
\noindent
{\bf State Space of the Production Process}. \\
In most  mathematical analyses, the production process  of proteins of a fixed gene is investigated. There is an implicit assumption in this case that the cell allocates a fixed fraction of its resources to the production of the mRNAs, proteins of this gene, independently of the production process of other proteins. With this hypotheses the production process of this type of protein can be analyzed in isolation of the other processes within the cell. This is a convenient given the complexity of the interactions. In this case, the production process is usually  described as a three-dimensional  Markov process $(I(t),M(t),P(t))$ where, for $t{\ge}0$,
\begin{itemize}
\item $I(t){\in}\{0,1\}$ describes the state of the gene at time $t$, active or inactive;
\item $M(t)$ is the number of mRNAs at this instant;
\item $P(t)$ is the number of proteins.
\end{itemize}
For these models, the duration of elongation of mRNAs or of proteins is not taken into account. For example, the transition $M(t){\to}M(t){+}1$, the production of an mRNA, is, implicitly,  associated to the binding of a polymerase to the gene. This amounts to neglect the elongation time of an mRNA. A glance at the numbers of Section~\ref{NumbSec} reveals that, in average, the time to build an mRNA of 3000 nucleotides is  35 seconds, which is not small compared to its lifetime, of the order of $2$ minutes. For different types of mRNAs, the average length is 1000, see Figure~\ref{FigSize} of Section~\ref{NumbSec} in the Appendix.

Another assumption, the price of the Markov property, is that the duration of the transitions have an exponential distribution. This hypothesis for the duration of time to have a binding polymerase{-}gene or ribosome{-}mRNA can be justified with the current parameters of the cell. The assumption is more questionable for lifetimes of mRNAs and, in particular, of proteins.  The lifetimes of  proteins are comparable to the duration of time between divisions of the cell. The interpretation of death/degradation is more difficult in such a case. See the discussion on the assumptions on the distribution of lifetimes. 

The Markovian description of the protein production process gives the possibility of using classical results of Markov theory,  for the convergence to equilibrium as well as for an analytical characterization of equilibrium points. 

\bigskip
\noindent
{\bf Numbers of copies or Concentration~?} \\
The representation of the levels of proteins and mRNAs by {\em numbers}, i.e.\ integers,  is convenient for the mathematical analysis, in particular to investigate the stochastic fluctuations.  For some aspects, it is nevertheless more natural to express the various quantities in terms of {\em concentrations} rather than numbers.  For example, if ${\cal R}$, ${\cal M}$ and ${\cal P}$ denotes respectively the chemical species ribosomes, mRNAs and proteins, the production and destruction of a protein is expressed as , using notations from chemistry,
\[
\begin{cases}
  {\cal M}{+}{\cal R} \longrightarrow  {\cal P}{+}{\cal M}{+}{\cal R},\\
  \phantom{a{+}{\cal R}}{\cal P} \longrightarrow \emptyset.
\end{cases}
\]
Let  $[{\cal A}](t)$ denote the concentration of species ${\cal A}$,  ${\cal A}{\in}\{{\cal R},{\cal M},{\cal P}\}$,  at time $t{\ge}0$. This quantity is the ratio of the number of macro{-}molecules in the cell of  type ${\cal A}$ to the volume of the cell at time $t$. The {\em law of mass action} of chemical physics gives the associated {\em Michaelis-Menten} kinetics equations for the concentration of the various macro{-}molecules, they are analytically expressed by a deterministic differential equation,
\[
\frac{\diff}{\diff t} [{\cal P}](t)= k_{\rm on}[{\cal M}](t)[{\cal R}](t)-\nu [{\cal P}](t).
\]
See van Kampen~\cite{vanKampen} and Chapter~6 of Murray~\cite{Murray} for example. The parameter $k_{\rm on}$ is the rate at which a given ribosome binds to a given mRNA and $\nu$ is the exponential growth rate of the volume of the cell. The concentration of proteins decreases mostly because of volume growth. This is the {\em dilution effect}. In a Markovian context, see Section~\ref{MarkSec}, the number of copies of a given type of protein is used rather than its concentration. In this case  the parameter $\nu$ is interpreted as a death rate  of proteins, which is less natural in some way.  In this presentation, we will nevertheless use the discrete representation with numbers. It is more convenient to describe stochastic phenomena involving a finite, but not too large, number of macro{-}molecules, like for mRNAs.  See Anderson and Kurtz~\cite{Anderson} for an introduction on stochastic modeling of biological systems. 

\bigskip
\noindent
{\bf Organization of the Paper.}\\
Due to its historical importance, the analysis of the first moments of the equilibrium  of  Markovian models is presented in Section~\ref{MarkSec}. 
Section~\ref{c-model} introduces a fundamental marked Poisson point process used to describe the whole  protein production process. It is an extension based on the model of  Fromion et al.~\cite{Fromion}. A result of convergence to equilibrium is proved via a ``coupling from the past'' method. An important representation of the number of mRNAs and proteins in terms of the Poisson process is established. See Theorem~\ref{theoCVMP}.

It should be noted that we do not include negative/positive feedback mechanisms in the models of gene expression, i.e. when proteins/mRNAs, or some other macro{-}molecules, may lock/unlock the gene of some proteins under some circumstances. We are here  mainly focused on the ``historical'' model of gene expression. The main technical reason is that the marked Poisson point process approach presented in this paper does not seem to work for this class models. They are generally investigated via scaling methods. See Mackey~\cite{Mackey} for example. 

In Section~\ref{mRNAsec}, (resp., Section~\ref{TranslSec}),  general formulas for the mean and variance of the number of mRNAs  (resp., of proteins),  at equilibrium are established, generalizing the classical formulas of Markovian models.  An important part of the biological literature is devoted to the analysis of these first moments, this is the main motivation of these two sections.  The results obtained are an extension, with somewhat simpler technical arguments,  of results of  Fromion et al.~\cite{Fromion}.  When the gene is always active, the equilibrium of the detailed description of the state of mRNAs is characterized by Proposition~\ref{theoMemp}. A formula for the generating function of the equilibrium distribution of the number of proteins is established in Proposition~\ref{GenP}.   The joint distribution of the number of mRNAs and of proteins at equilibrium is also investigated in Section~\ref{TranslSec}.

In Section~\ref{Sec-AppP}  the various approximations  used in the biological literature are revisited in the light of the results proved in the previous sections. They are formulated as new scaling results when one of the parameters goes to infinity: switching rate of the state of the gene, average lifetime of proteins, translation rate, \ldots generally under the constraint that the average number of proteins is fixed. Several convergence  results for the equilibrium distribution of the number of proteins are obtained in this way. The methods rely on the results of Section~\ref{TranslSec},  probabilistic arguments and some technical estimates.  A {\em central limit theorem}, which seems to be new, is  established in Proposition~\ref{CTL}. A study on the impact of the elongation phase of proteins on the variability of the protein production process concludes this section. This topic is rarely addressed in  mathematical models of the biological literature. It is shown  that, under some statistical assumptions,  increasing the variability of the elongation phase in a stochastic model may have the surprising effect of reducing the variance of the equilibrium  distribution of the  number of proteins. A consequence of this observation is that the choice of exponential distribution for the distribution of the duration of the  elongation phase of proteins, as it would be natural in a Markovian context, may lead to underestimate the ``real''  variance. 

The central result of this approach is  Theorem~\ref{theoCVMP}, it gives an explicit representation of the number of proteins at equilibrium in terms of a functional of a marked  Poisson point process. The main advantages are that
\begin{enumerate}
\item it is not necessary to solve invariant measure  equations as in the Markovian approaches;
\item Formulas for Poisson processes, recalled in the appendix,  give additional insight on the equilibrium distribution of the number of proteins;
\item General distributions, instead of exponential distributions,  for the duration of elongation times and lifetimes of mRNAs and proteins can be included in the model. 
  \end{enumerate}

\bigskip
\noindent
{\bf Acknowledgments.}\\
The material presented in this review is based on a series of studies in collaboration with Vincent Fromion (INRA) who has, convincingly, introduced the author to these highly interesting processes. A PhD thesis~\cite{Leoncini} by Emanuele Leoncini has been also part of this collaboration. The author is grateful for comments on a preliminary version of this paper by Ga\"etan Vignoud, Emanuele Leoncini and Fr\'ed\'eric Fl\`eche.

\bigskip
\noindent
{\bf Mathematical Notations and Conventions.}\\
Throughout this presentation, for the sake of clarity, we will frequently do the following abuse of notations. If $X$ is some integrable random variable on $\R_+$
\begin{enumerate}
\item $X(\diff x)$ will denote its distribution on $\R_+$;
\item $(X_{n})$ an i.i.d. sequence of random variables with this distribution;
\item $F_X$ a random variable whose distribution  has density
  \[
x\mapsto  \frac{\P(X{\ge}x)}{\E(X)}, 
  \]
  with respect to Lebesgue's measure on $\R_+$.  We denote by $(F_{X,n})$ an i.i.d. sequence with this distribution. 
\end{enumerate}
Note that, with Fubini's Theorem, for $a{\ge}0$,
  \begin{equation}\label{eqFa}
    \E\left((X{-}a)^+\right)=\int_a^{+\infty} \P(X{\ge}u)\,\diff u=\E(X)\P(F_X{\ge}a),
  \end{equation}
  where $b^+{=}\max(b,0)$, $b{\in}\R$.  Additionally, a simple calculation of the Laplace transform shows that  the relation $F_X{\steq{dist}}X$ holds if and only if the law of $X$ is exponential. 

  \medskip
  
Throughout the paper, we will use the following notation:
\begin{itemize}
\item For mRNAs, $L_{\iM}(\diff w)$ is the distribution of their lifetimes and  $E_{\iM}(\diff v)$ the distribution of their elongation times.
\item For proteins, $L_{\iP}(\diff z)$ is the distribution of their lifetimes and  $E_{\iP}(\diff y)$ the distribution of their elongation times.
\end{itemize}
The index $\iM$ refers to mRNAs and $\iP$ to proteins. With our convention, items~(1) and~(2) above, $(L_{\iM})$ and $(E_{\iM})$, (resp.   $(L_{\iM,n})$ and $(E_{\iM,n})$), denote random variables, (resp.  i.i.d. sequences of random variables) associated  to the distribution $L_{\iM}(\diff w)$, (resp. $E_{\iM}(\diff v)$).  And similarly for proteins, for the random variables  $(L_{\iP})$ and $(E_{\iP})$ and the i.i.d.\ sequences  $(L_{\iP,n})$ and $(E_{\iP,n})$.
\section{The Classical Markovian Three-Step Model}\label{MarkSec}
In this section,  Markovian models of the protein production process are presented. Although these models have some limitations in terms of modeling this is an important class of models to study the stochastic fluctuations associated to gene expression. This is in fact the main stochastic model of gene expression used in the biological literature from the early works of Rigney~\cite{Rigney}, Berg~\cite{Berg} to more recent studies Thattai and van Oudenaarden~\cite{Thattai}, Shahrezaei and Swain~\cite{Swain2}. See Paulsson~\cite{Paulsson} for a survey. See also Chapter~6 of Bressloff~\cite{Bressloff} and Chapter~4 of Mackey et al.~\cite{Mackey}. These models are still popular in the biological literature. 

A more general modeling of the stochasticity of  gene expression is presented and discussed in Section~\ref{c-model}. It is analyzed in the rest of the paper.
In this section we will consider the time evolution of the number of mRNAs and of proteins associated to a  gene $G_0$ in a given cell (bacterium). The  statistical assumptions are the following:
\begin{enumerate}
\item If the gene is inactive, in state $0$ say, (resp.,  active, in state $1$), it becomes active, (resp., inactive) after an exponentially distributed  amount of time with parameter $k_+$  (resp., $k_-$).
\item If the gene is active, the mRNAs of $G_0$ are produced according to a Poisson process with parameter ${k_{\iM}}$ and their lifetimes are exponentially distributed with parameter $\gamma_{\iM}$.
\item during its lifetime an mRNA of $G_0$ produces proteins according to a Poisson process with parameter ${k_{\iP}}$ and the lifetime of a protein is exponentially distributed with parameter $\gamma_{\iP}$.
\end{enumerate}
For this model, the parameter $k_\iM$ can be seen as the rate at which a polymerase binds on an active gene, but also the rate at which an mRNA is produced. The steps of binding and elongation are thus reduced to a single step. The same remark holds for the production of proteins. The model of Section~\ref{c-model} distinguishes these two steps.
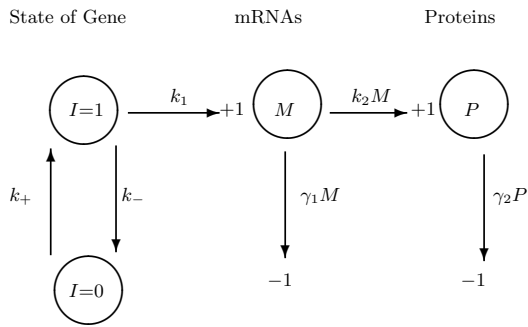
\begin{figure}
\setlength{\unitlength}{2044sp}

\begin{center}
\scalebox{0.8}{
\begin{picture}(7017,5292)(1828,-4594)
{\thicklines
\put(2395,-1328){\circle{1026}}
}%
{\thicklines\put(2465,-4058){\circle{1026}}
}%
\put(2165,-4158){$I{=}0$}
{\thicklines\put(5479,-1238){\circle{1026}}
}%
\put(4679,0){mRNAs}
\put(5279,-1408){$M$}
\put(1265,0){State of Gene}
\put(2165,-1408){$I{=}1$}
\put(3700,-1208){${k_{\iM}} $}
\put(6450,-1208){${k_{\iP}} M$}
\put(1265,-2700){$k_+$}
\put(2965,-2700){$k_- $}
\put(5679,-2700){$\gamma_{\iM} M$}
\put(8109,-4000){$-1$}
\put(8600,-2700){$\gamma_{\iP} P$}
{\thicklines\put(8309,-1238){\circle{1026}}
}%
\put(8159,-1408){$P$}
\put(7559,0){Proteins}
{\thicklines\put(3061,-1366){\vector( 1, 0){1390}}
}%
\put(4451,-1400){$+1$}
{\thicklines\put(6121,-1366){\vector( 1, 0){1175}}
}%
\put(7351,-1400){$+1$}
{\thicklines\put(1891,-3526){\vector( 0, 1){1620}}
}%
{\thicklines\put(2881,-1861){\vector( 0,-1){1620}}
}%
{\thicklines\put(5446,-1951){\vector( 0,-1){1620}}
}%
\put(5179,-4000){$-1$}
{\thicklines\put(8461,-1996){\vector( 0,-1){1620}}
}%
\end{picture}
}
\end{center}
\caption{Classical Three-Step Model for Protein Production.}\label{transfig1}
\end{figure}

For $t{\ge}0$, we denote $I(t){\in}\{0,1\}$ the state of  gene at time $t$ and  $M(t){\in}\N$, (resp.,  $P(t){\in}\N$), is the number of mRNAs, (resp., of proteins), at this instant. It is easily checked  that the process  $(Z(t)){\steq{def}}(I(t),M(t),P(t))$ is an irreducible Markov process in the state space ${\cal S}_0{\steq{def}}\{0,1\}{\times}\N^2$. The non{-}zero elements of its $Q$-matrix $Q{=}(q(z,z'),z,z'{\in}{\cal S}_0)$ outside the diagonal are given by, for $z{=}(i,x,y){\in}{\cal S}_0$, 
\[
\begin{cases}
q(z,z{+}e_{\iM}){=}k_+\ind{i{=}0}, &q(z,z{-}e_{\iM}){=}k_-\ind{i{=}1},\\  
q(z,z{+}e_2){=}{k_{\iM}}\ind{i{=}1}, &q(z,z{-}e_2){=}\gamma_{\iM} x,\\
q(z,z{+}e_3){=}{k_{\iP}} x , &q(z,z{-}e_3){=}\gamma_{\iP} y,
  \end{cases}
\]
where $e_i$, $i{=}1$, $2$, $3$, are the unit vectors of ${\cal S}_0$. See Figure~\ref{transfig1}. As a functional operator $Q$, it can be defined as 
\[
Q(f)(z)=\sum_{z'{\in}{\cal S}_0} q(z,z') f(z')= \sum_{z'{\in}{\cal S}_0\setminus\{z\} } q(z,z') (f(z'){-}f(z)),\qquad z{\in}{\cal S}_0,
\]
for some function $f$ on ${\cal S}_0$.
The starting point of analyses of the literature is always the classical system of Kolmogorov's equations associated to this Markov process
\[
\frac{\diff}{\diff t} \E\left[\rule{0mm}{3mm}f(Z(t))\right] =\E\left[\rule{0mm}{3mm}Q(f)(Z(t))\right],
\]
for some set of test functions $f$, like the indicator function of $x$, for $x{\in}{\cal S}_0$. They are generally referred to as {\em the master equation} in this literature. Its fixed point is the invariant distribution of the Markov process. 
\begin{prop}\label{IMPErgo}
  The Markov process $(I(t),M(t),P(t))$ has a unique invariant distribution $(\pi(z),z{\in}{\cal S}_0)$ with finite second moments
  \[
  \sum_{(i,x,y){\in}{\cal S}_0} \left(x^2{+}y^2\right)\pi(i,x,y)<{+}\infty.
  \]
\end{prop}
\begin{proof}
The proof is skipped. This result is in fact established in the next sections in a much more general context. See the proofs of Propositions~\ref{propM} and~\ref{FTMP}.
\end{proof}
This result on the first two moments is used in the following to derive an explicit expression of them. A much more general result holds in fact since the equilibrium of both variables $(M(t))$ and $(P(t))$ has a finite exponential moment, i.e. there exists $\eta{>}0$ such that
\[
  \sum_{(i,x,y){\in}{\cal S}_0} \left(e^{\eta x}{+}e^{\eta y}\right)\pi(i,x,y)<{+}\infty.
\]
See Gupta et al.~\cite{Gupta} for example. 

There does not seem that there exists a closed form expression for the invariant distribution  $(\pi(i,x,y))$ of the equilibrium equations. See Relation~\eqref{GenPM} below for the generating function of $P$ when the gene is always active. It turns out that, nevertheless, one can get  explicit expressions for the moments of this distribution. Related formulas of the kind have been the main and essentially, the only,  rigorous results of mathematical models of gene expression starting from the first studies in 1977. They are still used in the quantitative analyses of the biological literature. See the supplementary material of Taniguchi et al.~\cite{Taniguchi} for example. 
\begin{prop}\label{PaulProp}
  If $(I,M,P)$ is a random variable whose law is the invariant distribution of the Markov process $(I(t),M(t),P(t))$, then
  \[
  \E(I){=}\delta_+{\steq{def}}\frac{k_+}{k_+{+}k_-},\quad E(M){=}\delta_+\frac{{k_{\iM}}}{\gamma_{\iM}} \text{ and } \E(P){=}\delta_+\frac{{k_{\iM}}{k_{\iP}}}{\gamma_{\iM}\gamma_{\iP}},
\]
and, for the variances, if $\Lambda{=}k_+{+}k_-$,
\[
\frac{\Var(M)}{\E(M)}=1{+}(1{-}\delta_+)\frac{{k_{\iM}}}{\Lambda{+}\gamma_{\iM}},
\]
\begin{equation}\label{PaulVar}
  \frac{\Var(P)}{\E(P)}{=}1{+}\frac{k_{\iP}}{\gamma_\iM{+}\gamma_\iP}{+}(1{-}\delta_+)\frac{k_{\iM}k_{\iP}}{\gamma_{\iM}{+}\gamma_{\iP}}\frac{\Lambda{+}\gamma_\iM{+}\gamma_\iP}{(\Lambda{+}\gamma_\iM)(\Lambda{+}\gamma_\iP)}.
\end{equation}
\end{prop}
The ratio of the variance and of the mean given for $M$ and $P$ is called the  {\em Fano factor},   it can be seen as  a measure of the dispersion of the probability distribution. In the biological literature it is interpreted as a measure of deviation from the Poisson distribution. Recall that the Fano factor of a Poisson distribution is $1$.
\begin{proof}
 The proofs in the literature are generally based on the differential equations satisfied by the generating function of the distribution of the three-dimensional process. See Leoncini~\cite{Leoncini}.  A simple approach, avoiding manipulations of generating functions,  is presented, it is based on some natural flow equations at equilibrium. 

The equilibrium equations for the invariant distribution $\pi$ of the Markov process $(I(t),M(t),P(t))$ can be expressed as 
\[
\int_{{\cal S}_0}  Q(f)(i,x,y)\diff \pi(i,x,y)=\E[Q(f)(I,M,P)]=0,
\]
for any function on ${\cal S}_0$ such that $Q(f)$ is integrable with respect to $\pi$.  If, for $z{\in}{\cal S}_0$,  by taking $f$ the indicator function of $z$, this gives the usual balance equations.

We will use the class of functions $f(i,x,y){=}i^ax^by^c$, with $i{\in}\{0,1\}$, $a$, $b{\in}\{0,1,2\}$ to get an appropriate system of linear equations to get these moments. Proposition~\ref{IMPErgo} shows that $Q(f)$ is integrable for any function $f$ of this class. 

The equations
  \begin{align*}
    k_+\E(1{-}I){-}k_-E(I)=0, &&f(i,x,y)=i,\\
    k_\iM\E(I){-}\gamma_\iM\E(M)=0, && f(i,x,y)=x,\\
    k_\iP\E(M){-}\gamma_\iP\E(P)=0, && f(i,x,y)=y
  \end{align*}
  give immediately the formulas for the first moments.

A second system of linear  equations
    \begin{align*}
    k_\iM\E((2M{+}1)I) {+}\gamma_\iM\E(M(1{-}2M))=0,&& f(i,x,y)=x^2,\\
    k_\iP\E(M(2P{+}1)){+}\gamma_\iP\E(P(1{-}2P))=0,&& f(i,x,y)=y^2,\\
    k_\iM\E(IP){-}\gamma_\iM\E(MP){+}k_\iP\E(M^2){-}\gamma_\iP\E(MP)=0,&& f(i,x,y)=xy,\\
    k_+\E((1{-}I)M){-}k_-\E(IM){+}k_\iM\E(I){-}\gamma_\iM\E(IM)=0,&& f(i,x,y)=ix,\\
    k_+\E((1{-}I)P){-}k_-\E(IP){+}k_\iP\E(IM){-}\gamma_\iP\E(IP)=0,&& f(i,x,y)=iy,
    \end{align*}

\noindent
for the variables $E(M^2)$, $\E(P^2)$, $\E(MP)$, $\E(IM)$ and $\E(IP)$ gives the second moments of $M$ and $P$. The proposition is proved.   
\end{proof}
As it will be seen in a more general framework, when the gene is always active, i.e.\ $k_-{=}0$, the distribution of $M$ is Poisson. This explains in particular the identity $\Var(M){=}\E(M)$ of Proposition~\ref{PaulProp} in this case.  See Proposition~\ref{theoMemp} of Section~\ref{mRNAsec}. Note that it is not the case for $P$ since
\[
\frac{\Var(P)}{\E(P)}{>}1,
\]
by Proposition~\ref{PaulProp}. This is usually interpreted in the biological literature as the fact that the distribution of $P$ is more variable than a Poisson distribution. 

The above results, the first and second moments of the number of proteins at equilibrium,  are essentially the main mathematical results of the current biological literature in this domain. As an example of application of the methods presented in this paper, we state several results which give some additional insight on the properties of the equilibrium distribution of the number of proteins. They are simple applications of general results which are proved in Section~\ref{c-model} and~\ref{Sec-AppP}. 

An explicit representation of the generating function of $P$ is in fact available. 
\begin{prop}\label{CorGenP}
If  the gene is always active, the generating function of the equilibrium distribution of the number of proteins, for $z{\in}(0,1)$, is given by 
  \begin{multline}\label{GenPM}
  \E\left(z^{P}\right)
  =\exp\left({-}k_{\iM} \int_{\R_+}\left[1{-}e^{{-}\left( 1{-}z\right){k_{\iP}}\left(1{-}e^{{-}\gamma_\iP u}\right)/{\gamma_\iP}}\right] e^{-\gamma_\iM u}\diff u \right.\\\left. 
  {-}k_\iM\int_{\R_+^2}\left[1{-}e^{{-}\left( 1{-}z\right){k_{\iP}}e^{{-}\gamma_\iP u}\left(1{-}e^{{-}\gamma_\iP w}\right)/{\gamma_\iP}}\right] \gamma_{\iM}e^{-\gamma_\iM w} \diff u\diff w\right).
\end{multline}
\end{prop}
Bokes et al.~\cite{Bokes} gives, via an analytic approach with the master equation,  an alternative representation of the generating function in terms of an hypergeometric function.

The classical  relation~\eqref{PaulVar} for the variance shows that when the gene is always active, i.e.\ $\delta_+{=}1$,  and  the average lifetime $1/\gamma_\iP$ of proteins gets large, then
\[
\frac{\Var{P}}{\E(P)}{=}\E\left(\left(\frac{P-\E(P)}{\sqrt{\E(P)}} \right)^2\right)\sim 1{+}\frac{k_\iP}{\gamma_\iM}.
\]
The following proposition gives a more precise result, when lifetimes of proteins goes to infinity. It is a consequence of a general convergence theorem, Theorem~\ref{CTL} of Section~\ref{Sec-AppP}.
\begin{prop}
If  the gene is always active, for the convergence in distribution, the relation
    \[
  \lim_{\gamma_\iP{\to}0} \frac{P{-}\E(P)}{\sqrt{\E(P)}}={\cal N}\left(\sigma\right),
  \]
holds,  where ${\cal N}\left(\sigma\right)$ is a centered Gaussian random variable whose standard deviation is 
  \[
  \sigma =  \sqrt{1{+}{k_\iP}/{\gamma_\iM}}.
  \]
\end{prop}
This proposition gives a Gaussian approximation result for the distribution of $P$ as  $\E(P){+}\sqrt{\E(P)}{\cal N}(\sigma)$. The result will hold for protein types with a large number of copies in the cell. 

The following proposition establishes a convergence in distribution when the death rate of mRNAs is going to infinity with the constraint that the transcription rate $k_\iP$ is sufficiently large  so that each mRNA produces a fixed average number of proteins.
\begin{prop}\label{NBD}
If  the gene is always active,  then,  for $a{>}0$, the relation 
  \[
  \lim_{\substack{\gamma_\iM\to {+}\infty\\k_\iP/\gamma_\iM=a}} P=N(a),
  \]
holds for the convergence in distribution,  where $N(a)$ is random variable whose  distribution is given by, for $n{\in}\N$, 
\[
\P(N(a){=}n)=\frac{1}{(a{+}1)^{k_\iM/\gamma_\iP}} \frac{\Gamma(k_\iM/\gamma_\iP{+}n)}{n!\,\Gamma(k_\iM/\gamma_\iP)}\left(\frac{a}{1{+}a}\right)^n.
\]
where $\Gamma$ is the classical Gamma function.
\end{prop}
The distribution of $N(a)$ is the {\em negative binomial} distribution with parameters $k_\iM/\gamma_\iP$ and $a$. See Chapter~5 of Johnson et al.~\cite{Johnson}. This distribution is frequently mentioned in the approximations of the literature, Bokes et al.~\cite{Bokes}, Friedman et al.~\cite{Friedman} and Shahrezaei and Swain~\cite{Swain2} for example. It is also sometimes used to fit measurements of concentration of proteins. Its advantage compared to a simple Poisson distribution may be due to the fact that it is more flexible with its two parameters.
\begin{proof}
This is a straightforward application of Proposition~\ref{Long2Prop} of Section~\ref{Sec-AppP}.  The elongation time of proteins is null, i.e. $E_\iP{\equiv}0$, so that the generating function of the limit $N(a)$ is given
\[
\E\left(z^{N(a)}\right)=\exp\left({-}k_\iM\int_{0}^{+\infty}\frac{\left( 1{-}z\right)ae^{{-}\gamma_\iP u}}{1{+}\left( 1{-}z\right)a e^{{-}\gamma_\iP u}}\diff u\right)
=\left(\frac{1}{1{+}a{-}az}\right)^{\kappa},
\]
with  $\kappa{=}{k_\iM}/{\gamma_\iP}$.  Relations~\eqref{NBgen} and~\eqref{NBdist} of the appendix give the representation of this distribution.
\end{proof}
\section{A  General Stochastic Model}\label{c-model}
We introduce  in this section the  stochastic processes which will be studied for the analysis of the time evolution of the number of proteins.

The distributions of lifetimes of mRNAs and proteins are general, in particular they are not assumed to be exponentially distributed as in the Markovian model of Section~\ref{MarkSec}. The general distributions of the duration of elongations of mRNAs and proteins are also incorporated in the  stochastic model. Once the polymerase/ribosome is bound to the gene/mRNA,  the elementary components, nucleotides for mRNAs and amino-acids for proteins, are progressively added to build these components.  As it can be seen in Figure~\ref{FigSize} in the appendix, the number of these elements can vary from ${\sim}20$ elements up to several thousands depending on the gene considered. This, of course, incurs  delays which cannot be really neglected. This presentation is an extension of the framework of Fromion et al.~\cite{Fromion}. See also Leoncini~\cite{Leoncini}.

As it will be seen, a representation of the {\em equilibrium} of the protein production process in terms of  marked Poisson point processes is established.

Contrary to a Markovian analysis, the resolution of the linear system of balance equations for the equilibrium, is not necessary in this case.  More general assumptions on the distributions of some of the steps of the process can therefore be considered as mentioned before. It should be noted  however, that even in the Markovian context of the classical three-step model,  this representation simplifies much the analysis of the equilibrium.  For example, asymptotic results concerning the {\em equilibrium distribution} of the number of proteins are simpler to obtain than the approach using analytic tools with hypergeometric functions. See  Bokes et al.~\cite{Bokes} and Section~\ref{AppLongSec} for example. 

\subsection{Gene activation} \label{GenActSec}
The  state of the gene for the type of protein is either active or inactive. It is activated at rate $k_+$ and inactivated at rate $k_-$.  Time duration of these states are assumed to be exponentially distributed. 

The process of activation of the gene  is a Markov process $(I(t), t\in\R)$ with values in $\{0,1\}$ and whose $Q$-matrix $R_I$ is given by $r_I(0,1){=}k_+$ and $r_I(1,0){=}k_-$.  Without loss of generality, it will be assumed to be stationary. For this reason $(I(t))$ is defined on the whole real line,  in some sense, the activation/inactivation process has started at $t{=}{-}\infty$. As it will be seen, this is a convenient formulation to describe properly the equilibrium of the protein production process.

This process can be represented as a marked point process $((t_i,X_i), i{\in}\Z)$ where $(t_i)$ is the increasing sequence of instants of change of state of the gene, with the convention that $t_0{\leq} 0{<}t_1$. The marks $(X_i)$ in the set $\{0,1\}$ indicate the state of the gene at these instants. 
In particular, for $i{\in}\Z$,  conditionally  on the event $\{X_i{=}0\}$, (resp., on the event  $\{X_i{=}1\}$), the random variable $t_{i+1}{-}t_i$ is exponentially distributed with rate $k_+$, (resp., with rate $k_-$).  Because of our assumption on stationarity of the activation/deactivation process, the sequence of activation instants of the gene, $(t_i\ind{X_i{=}1})$,  is a stationary renewal  point process. See Section~11.5 of Robert~\cite{Robert}. 
We denote
\[
\delta_+{=}\frac{k_+}{\Lambda} \text{ and } \Lambda{\steq{def}}k_+{+}k_-,
\]
in particular $\P(I(t){=}1){=}\delta_+$, for all $t{\in}\R$.

\subsection{Transcription and Translation: a Fundamental Poisson  Process}
When a gene is active,  a {\em polymerase} can be bound to it according to a Poisson process with rate ${{k_{\iM}}}{>}0$. The {\em transcription phase} can start. An mRNA is then being built up by a process of aggregation of a specific sequence of nucleotides which are present in the cytoplasm. This is {\em the elongation phase} of an mRNA. Its duration has a distribution ${E_{\iM}}(\diff x)$ on $\R_+$. Similarly a {\em ribosome} is bound to a given mRNA according to a Poisson process with parameter ${k_{\iP}}$. The {\em translation phase} starts. A chain of {\em amino-acids} is created as the ribosome progresses on the mRNA. This is the {\em elongation phase} of a protein. Its duration has a distribution ${E_{\iP}}(\diff x)$ on $\R_+$. Both mRNAs and proteins have a finite lifetime. The distribution of the lifetime of an mRNA (resp., a protein) is  ${L_{\iM}}(\diff y)$ (resp.,  ${L_{\iP}}(\diff y)$) on $\R_+$. Throughout the paper, it is assumed that all these distribution have a finite first moment. 

\bigskip
\noindent
{\bf The Fundamental Poisson process.} These two steps are represented by  a marked  Poisson point process. See Section~\ref{PoisSec} of the Appendix.
\begin{equation}\label{PoisP}
  {\cal P} \steq{def} (u_n,E_{\iM,n},L_{\iM,n},{\cal N}_{\iP,n}, n{\in}\Z),
\end{equation}
on the state space
\[
{\cal S}\steq{def}\R{\times}{\R_+^2}{\times}{\cal M}_p(\R{\times}{\R_+^2}),
\]
where ${\cal M}_p(\R{\times}{\R_+^2})$ is the space of Radon point measures on $\R{\times}{\R_+^2}$, see Dawson~\cite{Dawson} for example, and 
\begin{enumerate}
\item $(u_n)$ is a Poisson process on $\R$ with parameter ${k_{\iM}}$, this is the sequence of possible instants when a polymerase can be bound to the gene. For $n{\in}\Z$, $u_n$ is indeed an instant of binding only if the gene is active at time $u_n$. The sequence $(u_n, n{\in}\Z)$ is assumed to be non-decreasing and indexed with the convention $u_0{\le}0{<}u_1$.
\item   $(E_{\iM,n}, n{\in}\Z)$ is the sequence of the duration of the elongation for these mRNAs, it is an i.i.d.  sequence whose common distribution is ${E_{\iM}}$. 
\item   $(L_{\iM,n})$ is the i.i.d. sequence of associated lifetimes of these mRNAs, its common law is ${L_{\iM}}$. 
\item $({\cal N}_{\iP,p})$ is an i.i.d. sequence with the same distribution as ${\cal N}_\iP$, a marked Poisson point process on $\R{\times}\R_+^2$ with intensity ${k_{\iP}}\diff x{\otimes}E_\iP{\otimes}{L_{\iP}}$.

  For $n{\in}\Z$, ${\cal N}_{\iP,n}{=}(x_{\iP,j}^n,E_{\iP,j}^n, L_{\iP,j}^n,j{\in}\Z)$ is the process associated with the protein production for the mRNA with index $n$. 
  \begin{enumerate}
  \item $(x_{\iP,j}^n,j{\in}\Z)$ is a Poisson process on $\R$ with parameter ${k_{\iP}}$, this is the sequence of possible instants when a ribosome can be bound to the mRNA with index  $n$. Only the instants occurring during the lifetime of this mRNA matter. 
    \item $(E_{\iP,j}^n)$ is the sequence of the duration of the elongation of proteins by this mRNA, it is an i.i.d.  sequence whose common distribution is $E_{\iP}$. 
    \item   $(L_{\iP,j}^n)$ is the i.i.d. sequence of associated lifetimes of associated proteins, its common law is ${L_{\iP}}$.
  \end{enumerate}
\end{enumerate}
The coordinates of  ${\cal P}(\diff u,\diff v, \diff w,\diff m)$ are associated to  the potential instants $u$ of the binding of a polymerase on the gene, the variable $v$  is the elongation time of the mRNA and the variable $w$ is its lifetime. The last component $m$ is the marked point process  associated to the protein production  process of this mRNA. The coordinates of $m(\diff x, \diff y, \diff z)$ are, $x$  is associated to binding instants  of ribosomes on this mRNA, $y$ is the elongation time of the corresponding protein  and $z$ its lifetime.

The intensity measure of  ${\cal P}$ is
\begin{equation}\label{nup}
\nu_{\cal P}{\steq{def}}{k_{\iM}}\diff u{\otimes}E_{\iM}{\otimes}{L_{\iM}}{\otimes}Q_\iP,
\end{equation}
where $Q_\iP$, $Q_{\iP}(\diff m){=}\P({\cal N}_{\iP}{\in}\diff m)$,  is the distribution on ${\cal S}$ of the random variable ${\cal N}_{\iP}$, a Poisson point process on $\R{\times}\R_+^2$ with intensity ${k_{\iP}}\diff x{\otimes}E_\iP{\otimes}{L_{\iP}}$.

Equivalently, if $F$ is a non{-}negative Borelian function on ${\cal S}$, then
\begin{align*}
\int_{\cal S}F(u,v,w,m)&\nu_{\cal P}(\diff u, \diff v, \diff w, \diff m)\\
&=\int_{\cal S}F(u,v,w,m) {k_{\iM}}\diff u E_{\iM}(\diff v){L_{\iM}}(\diff w) Q_\iP(\diff m)\\
&=\int_{\R}\E( F(u,E_1,L_1,{\cal N}_\iP)) {k_{\iM}}\diff u,
\end{align*}
where, with a slight abuse of notations, $E_1$, $L_1$ and ${\cal N}_{\iP}$ are independent random variables with respective distributions $E_{\iM}(\diff v)$, $L_{\iM}(\diff w)$ and $Q_\iP(\diff m)$.


\bigskip
\begin{figure}[t]
\setlength{\unitlength}{1744sp}
  \begin{picture}(10794,5040)(1516,-4594)
{\thicklines
\put(2395,-1328){\circle{1026}}
\put(2395,-4058){\circle{1026}}
\put(10171,-1411){\circle{1026}}
\put(11774,-1423){\circle{1026}}
\put(5446,-1366){\circle{1026}}
\put(7049,-1378){\circle{1026}}
}
\put(1891,-3526){\vector( 0, 1){1620}}
\put(7786,-1411){\vector( 1, 0){1710}}
\put(10756,-1411){\line( 1, 0){405}}
\put(11926,-2176){\vector( 0,-1){1620}}
\put(2881,-1861){\vector( 0,-1){1620}}
\put(3061,-1366){\vector( 1, 0){1710}}
\put(6031,-1366){\line( 1, 0){405}}
\put(7201,-2131){\vector( 0,-1){1620}}

\put(2106,-4111){${\scriptstyle I{=}0}$}%
\put(2106,-1421){${\scriptstyle I{=}1}$}%
\put(3646,-1151){$\scriptstyle {k_{\iM}}$}%
\put(8281,-1151){$\scriptstyle {k_{\iP}} M$}%
\put(1966,299){\tiny{State}}%
\put(1966,29){\tiny{of gene}}%
\put(6881,-1461){$\scriptstyle M$}%
\put(11656,-1501){$\scriptstyle P$}%
\put(5816,-2716){{\tiny Decay of}}%
\put(5816,-3016){{\tiny mRNAs}}%
\put(7516,-2716){$\scriptstyle {L_{\iM}}(\diff w)$}%
\put(10541,-2671){{\tiny Decay of}}%
\put(10541,-2971){{\tiny Proteins}}%
\put(12241,-2671){$\scriptstyle {L_{\iP}}(\diff z)$}%
\put(6066,-1681){$\scriptstyle {+}1$}%
\put(10846,-1771){$\scriptstyle {+}1$}%
\put(7111,-4176){$\scriptstyle {-}1$}%
\put(11791,-4131){$\scriptstyle {-}1$}%
\put(5721,-706){$\scriptstyle{{E_{\iM}}(\diff v)}$}%
\put(10406,-706){$\scriptstyle {E_{\iP}}(\diff y)$}%
\put(5471,299){\tiny{Elongation}}%
\put(5471, 29){\tiny{of mRNAs} }%
\put(3196,299){\tiny{Binding of}}%
\put(3196,29){\tiny{polymerases}}%
\put(3196,-251){\tiny{to gene}}%
\put(10286,299){\tiny{Elongation}}%
\put(10286, 29){\tiny{of proteins} }%
\put(7696,299){\tiny{Binding of}}%
\put(7696,29){\tiny{ribosomes}}%
\put(7696,-251){\tiny{to mRNAs}}%
\put(1431,-2716){$\scriptstyle k_+$}%
\put(3031,-2716){$\scriptstyle k_-$}%
\end{picture}%
\caption{A General Model for Protein Production.}\label{FigExt}
\end{figure}
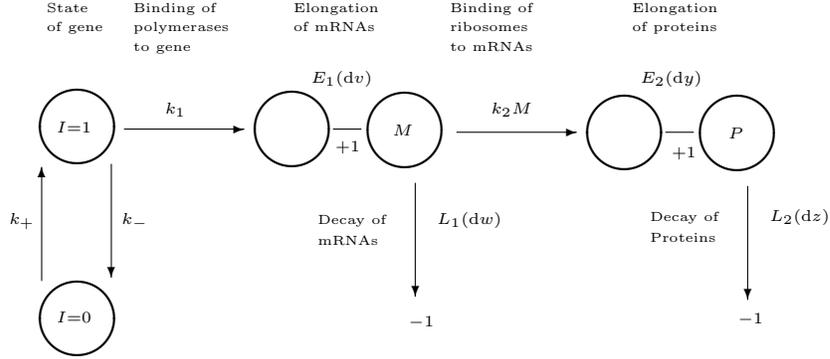

\medskip
\noindent
{\sc Examples.}
If the gene is always active and the system starts empty at time $0$, the formula
\[
\sum_{n{\ge}1}  \ind{u_n{+}E_{\iM,n}\leq t}=\int_{\R_+{\times}\R_+^2} \ind{u{+}v\leq t} {\cal P}(\diff u,\diff v, \diff w)
\]
gives the total number of mRNAs created up to time $t$.
The following notation has been used
\[
\int_{\cal S} f(u,v,w) {\cal P}(\diff u,\diff v,\diff w)=\int_{{\cal S}} f(u,v,w) {\cal P}(\diff u,\diff v, \diff w,\diff m),
\]
if $f$ is non-negative Borelian function on $\R{\times}\R_+^2$, i.e.\ ${\cal P}(\diff u,\diff v, \diff w)$ is the marginal distribution of ${\cal P}$ with respect to the three first coordinates, i.e.\ it stands for ${\cal P}(\diff u, \diff v, \diff w, {\cal M}_p(\R{\times}\R_+^2))$. 

Similarly, still with a permanently active gene, 
\[
\sum_{n{\ge}1}  \ind{u_n{+}E_{\iM,n}\leq t <  u_n{+}E_{\iM,n}{+}L_{\iM,n}}=\int_{\R_+\times\R_+^2} \ind{u{+}v\leq t< u{+}v{+}w} {\cal P}(\diff u,\diff v, \diff w)
\]
is the number of mRNAs present at time $t{\ge}0$. 

If we include the gene dynamics into the formula, the number of messengers present at time $t$ is
\[
\sum_{n{\ge}1}  I(u_n)\ind{u_n{+}E_{\iM,n}\leq t< u_n{+}E_{\iM,n}{+}L_{\iM,n}}=\int_{\R_+\times\R_+^2} I(u)\ind{u{+}v\leq t< u{+}v{+}w} {\cal P}(\diff u,\diff v, \diff w)
\]
If the mRNA with index $n$ is created at time $v{=}u_n{+}E_{\iP,n}$ and if its lifetime is $w$  then
\[
{\cal N}_{\iP,n}([v,v{+}w){\times}\R_+^2)= \int_{[v,v{+}w)\times\R_+^2} \,{\cal N}_{\iP,n}(\diff x,\diff y,\diff z) 
\]
is the total number of proteins created by such an mRNA during its lifetime.

\bigskip

\noindent 
{\bf Comments on the Model}\\
As a mathematical model, the above representation simplifies several aspects of the protein production process but the key steps of protein production are included. Furthermore, its main parameters have a clear biological interpretation.
\begin{enumerate}
\item[]
\item {\em Expression Rates}. The parameter ${k_{\iM}}$ is the {\em affinity} of polymerases, or the {\em transcription rate} for the gene considered. The larger this rate is, the more likely polymerases will bind to this gene rather than to some other genes with a lower affinity. The same remark applies for ${k_{\iP}}$, the affinity of ribosomes, the {\em translation rate}, for the corresponding mRNAs. 

It should be stressed that if mRNAs and  proteins are specific to a gene,  polymerases and  ribosomes are not. They can be used to create any type of mRNA and any type protein respectively. Polymerases and ribosomes can be seen as a kind of resources of the cell. The genes are, in some sense, competing to have access to polymerases.  This approach with the parameters $k_\iM$ and $k_\iP$ considers that a specific gene receives a fixed portion of resources of the cell in terms of polymerases and ribosomes. In particular, it does not include the competition for access to resources between the different genes. Alternative stochastic models have to be considered to analyze these aspects. 
\item[]
\item {\em Elongation Times}. An mRNA being a sequence of $N$ nucleotides, the elongation time ${E_{\iM}}$ of an mRNA is  represented as a sum of $N$ independent random variables, each of them corresponding to the duration of time required to get each of the nucleotides within the cytoplasm. A similar situation holds  for $E_\iP$ for the number of amino-acids of the protein. See Figure~\ref{FigSize} in the appendix. 
\item[]
\item\label{LFTPD} {\em Lifetimes}. The lifetimes of mRNAs are generally smaller than the lifetimes of proteins, mRNAs can be indeed degraded by other RNAs in the cytoplasm after a couple of minutes. This is one of the mechanisms used by the cell to regulate the protein production process. The situation is quite different for proteins, their lifetime may exceed the duration of time for cell division, of the order of 40mn for bacterium. At cell division, macro-molecules are assumed to be allocated at random into one of the two daughter cells.  In this case, if one follows a specific path of cells in the associated binary tree of successive cell divisions,  proteins ``vanish'' simply because they have been allocated in the daughter cell. This is referred to as the {\em dilution} phenomenon.
\item[]
  \item {\em Multiple Copies of the Gene}. In a favorable environment, cells grow and divide very quickly. In particular a copy of the DNA is always on the way in the cell, and, therefore, another copy of the gene is present at some stage. In principle this may double the transcription rate of the gene when the corresponding part of the DNA has been duplicated. The rate is  reduced at the  division time.  This aspect has been omitted in our model, mainly for the sake of simplicity. The main change would be that the process $(I(t))$ should have at least three values, depending on the number of copies of the gene which are active at a given time. See Paulsson~\cite{Paulsson} in a Markovian context. 
\end{enumerate}
The classical Markovian model of Section~\ref{MarkSec} corresponds to the case when the elongation times of mRNAs and proteins are null and that their respective lifetimes are exponentially distributed. 

\medskip
\noindent
{\bf How to Handle Functionals of ${\cal P}$.}\\
We will use repeatedly several possible representations in the calculations of expected values of functionals of ${\cal P}$. We give a quick sketch of the general approach. 
Let $f$ be a non{-}negative Borelian function on $\R{\times}\R_+^2{\times}\R{\times}\R_+^2$ and let, for $(u,v,w,m){\in}{\cal S}$,
\[
F(u,v,w,m)\steq{def}\int_{\R{\times}\R_+^2} f(u,v,w,x,y,z)m(\diff x, \diff y, \diff z).
\]
When the gene is always active, most of random variables of interest can be expressed under the form
\[
X{=}\croc{{\cal P},F}{\steq{def}}\int_{\cal S}F(u,v,w,m){\cal P}(\diff u, \diff v, \diff w, \diff m).
\]
As it will be seen, the variable activity of the gene adds some technical complications.
\begin{enumerate}
\item Calculation of the mean. Relation~\eqref{PoisMean} of the appendix gives that
\begin{align*}
  \E(X)&=\int_{\cal S}F(u,v,w,m)\nu_{\cal P}(\diff u, \diff v, \diff w, \diff m)\\
   &{=}\int_{\cal S}F(u,v,w,m){k_{\iM}}\diff u E_{\iM}(\diff v)L_{\iM}(\diff w)Q_\iP(\diff m){=}k_{\iM}\int_{\R} \E(F(u,E_1,L_1,{\cal N}_{\iP}))\diff u,
\end{align*}
with the same slight abuse of notations as before. For $(u,v,w){\in}\R_+^3$, since the Poisson process ${\cal N}_{\iP}$ has intensity measure ${k_{\iP}}\diff x{\otimes}E_\iP{\otimes}{L_{\iP}}$ on $\R{\times}\R_+^2$,
\begin{multline*}
  \E(F(u,v,w,{\cal N}_{\iP}))=\E\left(\int_{\R{\times}\R_+^2} f(u,v,w,x,y,z){\cal N}_{\iP}(\diff x, \diff y, \diff z)\right)\\
  =\int_{\R{\times}\R_+^2} f(u,v,w,x,y,z){k_{\iP}}\diff xE_\iP(\diff y)L_{\iP}(\diff z)
    =k_{\iP}\int_{\R} \E(f(u,v,w,x,E_\iP,L_\iP))\diff x,
\end{multline*}
hence
\[
  \E(X)=k_{\iM}k_{\iP}\int_{\R^2} \E(f(u,E_1,L_1,x,E_2,L_2))\diff u \diff x.
  \]
For  calculation of variances, the approach is similar, via Relation~\eqref{PoisVar} of the appendix and an additional trick. See the proof of Relation~\eqref{VarM} for example. 
\item[]
\item Exponential Moments of $X$. For more precise results on the distribution of $X$, one has to express $\E(\exp({-} \xi X))$, for some $\xi{\ge}0$, the {\em Laplace transform} of $X$ at $\xi$. Proposition~\ref{Pois-lapois}  of the appendix gives the relation
\begin{align*}
  \E\left(e^{{-}\xi \croc{{\cal P},F}}\right)&=\exp\left({-}\int_{\cal S} \left(1{-}e^{{-}\xi F(u,v,w,m)}\right)\nu_{\cal P}(\diff u, \diff v, \diff w, \diff m)\right)\\
&=\exp\left({-}\xi \int_{\cal S} \left(1{-}e^{{-}\xi F(u,v,w,m)}\right){k_{\iM}}\diff u E_{\iM}(\diff v)L_{\iM}(\diff w)Q_\iP(\diff m)\right)\\
&=\exp\left({-}k_\iM\int_{\R{\times}\R_+^2} \left(1{-}\E\left(e^{{-}\xi F(u,v,w,{\cal N}_{\iP})}\right)\right) \diff uE_{\iM}(\diff v)L_{\iM}(\diff w)\right).
\end{align*}
By using again Proposition~\ref{Pois-lapois} for the Poisson process ${\cal N}_\iP$, with the same arguments as before, we get, for $(u,v,w){\in}\R{\times}\R_+^2$,
\begin{align*}
\E\left(e^{{-}\xi F(u,v,w,{\cal N}_{\iP})}\right)&= \exp\left({-}\int_{\R}\left(1{-}\E\left(e^{{-}\xi f(u,v,w,x,y,z)}\right)\right)\,{k_{\iP}}\diff x E_\iP(\diff y)L_\iP(\diff z)\right)\\
&=\exp\left({-}{k_{\iP}}\int_{\R}\left(1{-}\E\left(e^{{-}\xi f(u,v,w,x,E_2,L_2)}\right)\right)\,\diff x\right).
\end{align*}
This expression for  $\E(\exp({-}X))$ though explicit is not simple to handle. It is nevertheless usable to get several limit results in Section~\ref{Sec-AppP}. The trick of using independent random variables $E_a$, $L_a$, $a{\in}\{\iM,\iP\}$ or their distributions $E_a(\diff s)$, $L_a(\diff t)$ is used throughout this paper. Its main advantages are of simplifying the calculations {\em and } the expressions of the formulas obtained.
\end{enumerate}

We state a simple result on some invariance properties of the Poisson process ${\cal P}$. For  $m{\in}{\cal M}_p(\R{\times}\R_+^2)$ and $f$ is a Borelian function on $\R{\times}\R_+^2$, the measure $\check{m}$ is defined by 
\[
\croc{\check{m},f}{=}\int_{\R{\times}\R_+^2} f(x,y,z) \check{m}(\diff x, \diff y,\diff z)=\int_{\R{\times}\R_+^2} f({-}x,y,z) {m}(\diff x, \diff y,\diff z).
\]
\begin{prop}\label{PoisSym}
If ${\cal P}$ is the marked Poisson point process on ${\cal S}$ defined by Relation~\eqref{PoisP} and $F$ is some non{-}negative Borelian function on $\{0,1\}{\times}{\cal S}$ then
  \[
  \int_{\cal S}  F(I(u), u,v,w,m){\cal P}(\diff u, \diff v, \diff w, \diff m)\steq{dist.}  \int_{\cal S} F(I({-}u),{-}u,v,w,\check{m}){\cal P}(\diff u, \diff v, \diff w, \diff m)
  \]
\end{prop}
\begin{proof}
  This is a simple consequence of
  \begin{itemize}
  \item the reversibility of $(I(t))$. It is easily checked that $(I({-}t))$ is a Markov process with the same $Q$-matrix, and with same distribution at $t{=}0$ as $(I(t))$. See Kelly~\cite{Kelly} for example. 
  \item If $(t_n)$ is a Poisson process on $\R$ with rate $\lambda{>}0$, then $({-}t_n)$ is  a Poisson process with the same rate. 
\item The independence of $(I(t))$ and ${\cal P}$.
  \end{itemize}
  Consequently,  we have the identity
  \[
  (I(u_n), u_n,L_{\iM,n},{\cal N}_{\iP,n})\steq{dist} (I({-}u_n), {-}u_n,L_{\iM,n},\check{\cal N}_{\iP,n}),
  \]
and, therefore, our proposition. 
\end{proof}
\subsection{Convergence to Equilibrium}\label{ConvSec}
As before we denote by $M(t)$, (resp., $P(t)$) the number of mRNAs, (resp., of proteins)  at time $t{\ge}0$. It is assumed that the initial state is as follows: there are
\begin{itemize}
\item $M(0)$ mRNAS  with respective remaining lifetimes $L^0_{\iM,k}$, $k{=}1$,\ldots,$M(0)$;
\item  $P(0)$ proteins  with respective remaining lifetimes $L^0_{\iP,k}$, $k{=}1$,\ldots,$P(0)$. The marked Poisson process associated to the creation of proteins by the $k$th mRNA present at time $0$ is denoted by ${\cal N}_{{\iP},k}^{0}$.
\end{itemize}
For simplicity it is assumed that, initially, there are no mRNAs or proteins in their elongation phase. The proof of convergence in distribution does not change if the initial state includes components in their elongation phase.

The following theorem is the key result concerning the equilibrium distribution of the number of mRNAs and proteins.
\begin{enumerate}
\item It shows the convergence to equilibrium without having a Markovian framework which is, usually, the classical approach to prove such a convergence. As a consequence,  the distribution of elongation times, lifetimes of mRNAs and proteins can be general. 
\item The equilibrium of the number of mRNAs and of proteins are expressed in terms of  functionals of the marked Poisson process ${\cal P}$. The closed form expressions of averages and variances are obtained with this representation.

\end{enumerate}
\begin{theorem}\label{propEq}[Number of mRNAs and Proteins at Equilibrium]\label{theoCVMP}
As $t$ goes to infinity, the random variable $(M(t),P(t))$ converges in distribution to $(M,P)$ defined by 
\begin{equation}\label{EqEqM}
  M=
  \left(\int_{\cal S} I(u)
  \ind{v{\le}u{<}v{+}w}{\cal P}(\diff u,\diff v, \diff w)\right)
\end{equation}
and 
\begin{equation}\label{EqEqP}
  P=  \left(\displaystyle\int_{{\cal S}} I(u)
 \left(\int_{\R{\times}\R_+^2}\ind{\substack{x{+}v{\le}u{<}x{+}v{+}w\\y{\le} x{<}y{+}z}}m(\diff x,\diff y,\diff z)\right)
{\cal P}(\diff u,\diff v, \diff w,\diff m)\right),
\end{equation}
where ${\cal P}$ is the marked Poisson point process on ${\cal S}$ defined by Relation~\eqref{PoisP}. 
\end{theorem}
The proof of the convergence relies on {\em coupling from the past arguments}. The idea consists  in starting the process at time ${-}t$ and to study its state at time $0$, which will have the same distribution as the original process at time $t$. If the process has convenient properties, of monotonicity for example, it may happen that the state at time $0$ of the shifted process converges {\em almost surely} as $t$ goes to infinity. This gives then the convergence in distribution of $(I(t),M(t),P(t))$ when $t$ goes to infinity. This method applies in our case. One of the earliest works using this method seems to be Loynes~\cite{Loynes} (1962). Its use has been popularized later by Propp and Wilson~\cite{Propp} to study the Ising Model.  See Levin at al.~\cite{Levin} for a survey.
\begin{proof}
We first express the variables $M(t)$ and $P(t)$  in terms of the point process ${\cal P}$. 

The variable $M(t)$ is the sum of the number of  initial  mRNAs still alive at time $t$ and of the number of  mRNAs born before time $t$ and still alive at time $t$. This gives the following formula,
\[
 M(t) = \sum_{k=1}^{M(0)} \ind{L_{\iP,k}^0{>}t} +\int_{\cal S} I(u)\ind{0{<}u{+}v{\le}t{<}u{+}v{+}w}{\cal P}(\diff u,\diff v, \diff w).
 \]
 Similarly, $P(t)$ is the sum of three quantities corresponding to the number of 
 \begin{enumerate}
 \item the number of proteins present at time $0$ and still alive at time $t$,
   \[
   \sum_{k=1}^{P(0)} \ind{L_{\iP,k}^0{>}t};
   \]
 \item the number of proteins created before time $t$ by one of the mRNAs present at time $0$  and still alive at time $t$,
   \[
   \sum_{k=1}^{M(0)} \int_{\R{\times}\R_+^2} I(u)\ind{0{<}u{<}L^0_{\iM,k},u{+}v\leq t {<} u{+}v{+}w}{\cal N}_{\iP,k}^{0}(\diff u, \diff v,\diff w);
   \]
 \item proteins created before time $t$ by new mRNAs and still alive at time $t$,
\[
   P_\iM(t){\steq{def}} \int_{{\cal S}} I(u)\ind{u{\in}(0,t)}
   \left(\int_{\R{\times}\R_+^2} \ind{\substack{u{+}v{\le}x{<}u{+}v{+}w\\x{+}y{\le} t{<}x{+}y{+}z}}m(\diff x,\diff y,\diff z)\right){\cal P}(\diff u,\diff v, \diff w,\diff m).
\]
 \end{enumerate}
For simplicity of presentation, we will prove the convergence in distribution of $(P(t))$. The convergence in distribution  of $(M(t),P(t))$ is  similar. 
Clearly, the two first terms converge almost surely to $0$ when $t$ goes to infinity. We have thus only to take care of $P_\iM(t)$.

The stationarity of the process ${\cal P}$ with respect to translation by ${-}t$, Proposition~\ref{PoisSym},  the identity $(I(s{-}t)){\steq{dist}}(I(s))$, and the fact  that $(I(s))$ is independent of ${\cal P}$ give that $P_\iM(t)$ has the same distribution as $P_2(t)$ with
\[
  P_2(t){\steq{def}}\int_{\cal S} I(u)\ind{u{\in}(-t,0)}
  \left(\int_{\R{\times}\R_+^2} \ind{\substack{u{+}v{\le}x{<}u{+}v{+}w\\x{+}y{\le} 0{<}x{+}y{+}z}}m(\diff x,\diff y,\diff z)\right){\cal P}(\diff u,\diff v, \diff w,\diff m).
\]
   The quantity $P_2(t)$ can be seen the number of proteins at time $0$ when the process starts empty at time ${-}t$, i.e. without mRNAs or proteins.  This is a non-decreasing function of $t$ converging almost surely to $P_2(\infty)$ defined by 
   \[
\int_{\cal S} I(u)\ind{u{<}0}
  \left(\int_{\R{\times}\R_+^2} \ind{\substack{u{+}v{\le}x{<}u{+}v{+}w\\x{+}y{\le} 0{<}x{+}y{+}z}}m(\diff x,\diff y,\diff z)\right){\cal P}(\diff u,\diff v, \diff w,\diff m).
\]
Proposition~\ref{PoisSym} gives that the quantity $P_2(\infty)$ has the same distribution as 
\[
\int I(u)\ind{u{>}0}
  \left(\int \ind{\substack{x{+}v{\le}u{<}x{+}v{+}w\\y{\le} x{<}y{+}z}}m(\diff x,\diff y,\diff z)\right){\cal P}(\diff u,\diff v, \diff w,\diff m).
\]
it is easy to see that this last term is the second coordinate of~\eqref{EqEqP}. The theorem is proved.
\end{proof}
\section{Transcription} \label{mRNAsec}
Relation~\eqref{EqEqM} gives that the distribution of the number of mRNAs at equilibrium is given by the law of the random variable $M$ defined by 
\begin{equation}\label{EqEqM2}
M=\int_{\R{\times}\R_+^2} 
I(u) G_{\iM}(u,v,w) {\cal P}(\diff u,\diff v, \diff w),
\end{equation}
with
\begin{equation}\label{EqMG1}
G_{\iM}(u,v,w){\steq{def}}\ind{v{\le}u{<}v{+}w}.
\end{equation}

\medskip
The process of activation/deactivation of the gene complicates significantly the derivation of the mathematical expressions for the mean and variance of $M$. Formulas are much more simple when the gene is always active. We will need some technical results on this process. We denote by ${\cal F}^I$ the $\sigma$-field generated by the stochastic process $(I(t))$,
\[
{\cal F}^I=\sigma\left\langle I(t), t{\in}\R\right\rangle.
\]
A representation of the distribution  of $M$  in terms of a marked Poisson point process  conditioned on the $\sigma$-field ${\cal F}^I$ is the main tool used. For the calculation of the variance an additional work has then to be done to ``remove'' this conditioning. It turns out that the correlation structure of the process $(I(t))$  plays a role at this stage.  
\begin{lemma}[Correlation Function of $(I(t))$]\label{CorI}
For $t{\ge}0$, $\E(I(t)){=}\delta_+$ and 
  \[
  \P(I(t){=}1|I(0){=}1)=\delta_++(1{-}\delta_+)e^{{-}\Lambda t},
  \]
  where $\Lambda{=}k_+{+}k_-$ and $\delta_+{=}k_+/\Lambda$.
\end{lemma}
\begin{proof}
  Let $p(t){=}\P(I(t){=}1|I(0){=}1)$, then Kolmogorov's forward equations give the ODE
  \[
  p'(t)=k_+(1{-}p(t)){-}k_-p(t)
  \]
  which is easily solved. 
\end{proof}
Define ${\cal P}^I$ the marked point process by, for a non{-}negative Borelian function on  the space ${\cal S}{\steq{def}}\R{\times}\R_+^2{\times}{\cal M}_p(\R{\times}\R_+^2)$,
\begin{equation}\label{PI}
\croc{{\cal P}^I,f}= \int_{\cal S}  I(u)f(u,v,m){\cal P}(\diff u,\diff v, \diff m).
\end{equation}
The following proposition gives the  intuitive, but important, result that, conditionally on ${\cal F}^I$, the point process ${\cal P}^I$ is a marked Poisson point process with intensity measure
\begin{equation}\label{nupI}
\nu_{\cal P}^I(\diff u, \diff v, \diff m) =  I(u) {k_{\iM}}\diff u\,{L_{\iM}}(\diff v) Q_{\iP}(\diff m). 
\end{equation}
See Relation~\eqref{nup} for the unconditional case. 
\begin{prop}\label{CondI}
For any non{-}negative Borelian function $f$ on the state space ${\cal S}$, the relation
\begin{multline*}
  \left.\E\left(\exp\left({-}\croc{{\cal P}^I,f}\right)\right| {\cal F}^I\right)
  =\exp\left({-}\int_{\cal S}  I(u)\left(1{-}e^{{-}f(u,v,w,m)}\right)\nu_P(\diff u,\diff v,\diff w,\diff m)\right)\\
  =\exp\left({-}{k_{\iM}}\int_{\cal S}  I(u)\left(1{-}e^{-f(u,v,w,m)}\right) \,\diff u E_\iM(\diff v) {L_{\iM}}(\diff w)Q_{\iP}(\diff m)\right)
\end{multline*}
holds almost surely.
\end{prop}
\begin{proof}
  For the sake of rigor,  despite of its intuitive content, a  proof is given. The (formal) difficulty is of expressing rigorously the conditioning with respect to ${\cal F}^I$.   We use the notations of Section~\ref{GenActSec} where  the process $(I(t))$ is defined by the doubly infinite sequence $(t_i,X_i)$. Denote, for $N{\ge}1$,
  \[
  {\cal F}^I_N\steq{def} \sigma\croc{(t_i,X_i), {-}N{\le}i{\le}N},
  \]
  note that, since the distribution of the  finite marginals of the sequence $((t_i,X_i), i{\in}\Z)$ determine its distribution, we have
  \[
  \lim_{n\to{+}\infty}  \uparrow {\cal F}^I_N=  {\cal F}^I.
  \]
  Let $f$ a non{-}negative Borelian function on ${\cal S}$ with compact support in the following sense: there exists some $K{>0}$ such that  $f(z){=}0$ if $z{=}(u,v,w,m){\in}{\cal S}$ when $u{\not\in}[-K,K]$.   On the event $\{t_{{-}N}{\le}{-}K, t_{N}{\ge}K\}$,
  \[
 \croc{{\cal P}^I,f}=\sum_{i=-N}^N \ind{X_i{=}1}\int_{\cal S} \ind{u{\in}[t_i,t_{i+1})} f(u,v,w,m)\,{\cal P}(\diff u, \diff v, \diff w, \diff m).
 \]
 By independence of $(t_i,X_i)$ and ${\cal P}$ and by using Proposition~\ref{Pois-lapois},  it is easily checked that, almost surely on this event,
\begin{multline*}
\left.\E\left(\exp\left({-}\croc{{\cal P}^I,f}\right)\right| {\cal F}^I_N\right)\\
=\exp\left({-}\sum_{i{=}{-}N}^N \ind{X_i{=}1}\int_{\cal S}  \ind{u{\in}[t_i,t_{i+1})}\left(1{-}e^{{-}f(u,v,w,m)}\right)\nu_P(\diff u,\diff v,\diff w,\diff m)\right).
\end{multline*}
Indeed, one has to multiply both sides of this relation by some non{-}negative Borelian function of $((t_i,X_i),{-}N{\le}i{\le}N)$, take the expected value of these expressions. By using the independence of $((t_i,X_i),{-}N{\le}i{\le}N)$ and ${\cal P}$, and the expression of the Laplace transform of a Poisson process, see Proposition~\ref{Pois-lapois}, it is easy to check the equality of these two terms.  Consequently, for $N$  sufficiently large, almost surely, 
\[
\E\left(\exp\left({-}\croc{{\cal P}^I,f}\right)| {\cal F}^I_N\right)=
\exp\left({-}\int_{\cal S}  I(u)\left(1{-}e^{{-}f(u,v,w,m)}\right)\nu_P(\diff u,\diff v,\diff w,\diff m)\right).
\]
A classical result from martingale theory gives the almost sure convergence
\[
\lim_{N{\to}{+}\infty} \left.\E\left(\exp\left({-}\croc{{\cal P}^I,f}\right)\right| {\cal F}^I_N\right)=\left.\E\left(\exp\left({-}\croc{{\cal P}^I,f}\right)\right| {\cal F}^I\right).
\]
See  Williams~\cite{Williams}. Hence the proposition holds for these class of functions $f$ with compact support  in the sense defined above.  We conclude with  the fact that any positive Borelian function can be expressed as a (monotone) limit of such functions. 
\end{proof}

\subsection{Mean and Variance of the Number of mRNAs at Equilibrium}

\begin{prop}[The first moments of the number of mRNAs]\label{propM}
If $M$ is the number of mRNAs at equilibrium, then 
\begin{equation}\label{EM}
\E(M)= \frac{k_+}{\Lambda}{k_{\iM}}\E({L_{\iM}}),
\end{equation}
and
\begin{equation}\label{VarM}
\frac{\Var(M)}{\E(M)}=1{+}\E(M)\frac{k_-}{k_+}\E\left(e^{-\Lambda |F_{{L_{\iM}},1}{-}F_{{L_{\iM}},2}{+}E_{\iM,1}{-} E_{\iM,2}|}\right).
\end{equation}
where $\Lambda{=}k_+{+}k_-$, and $E_{\iM,1}$ and $E_{\iM,2}$ are independent random variables with distribution ${E_{\iM}}$, and  $F_{{L_{\iM}},1}$ and $F_{{L_{\iM}},2}$ are independent random variables with density $\P({L_{\iM}}{\ge}u)/\E({L_{\iM}})$ on $\R_+$.
\end{prop}
Recall that the distribution of the elongation time of an mRNA is ${E_{\iM}}$ and the duration of its lifetime is ${L_{\iM}}$.
Note that when the lifetimes $L_\iM$  of an mRNA is exponentially distributed with  parameters $\gamma_\iM$, then
$F_{L_\iM}{\steq{dist}}L_\iM$, and if the elongation time of an mRNA is null, i.e.\ $E_\iM{\equiv}0$, the above relation gives 
\[
\frac{\Var(M)}{\E(M)}{-}1=\E(M)\frac{k_-}{k_+}\E\left(e^{-\Lambda |L_{\iM,1}{-}L_{\iM,2}|}\right)=\frac{k_-}{\Lambda}\frac{k_\iM}{\gamma_\iM}\E\left(e^{-\Lambda L_{\iM}}\right)
=\frac{k_-}{\Lambda}\frac{k_\iM}{\Lambda{+}\gamma_\iM},
\]
since $|L_{\iM,1}{-} L_{\iM,2}|{\steq{dist}}L_\iM$, due to the exponential assumption. This is the relation for the variance of $M$ in Proposition~\ref{PaulProp}.
\begin{proof}
Proposition~\ref{CondI} gives that, almost surely,
\begin{align}
\E\left(M{\mid}{\cal F}^I\right) &=
\int_{\cal S} I(u) \ind{v{\le}u{<}v{+}w}{\nu}_{\cal P}^I(\diff u,\diff v, \diff w,\diff m)\label{uaux1}\\
&=\int_{\R{\times}\R_+^2} I(u)\ind{v{\le}u{<}v{+}w} {k_{\iM}}\diff u {E_{\iM}}(\diff v){L_{\iM}}(\diff w)\notag\\
&={k_{\iM}}\int_{\R} I(u)\P({E_{\iM}}{\le}u{<}{E_{\iM}}{+}{L_{\iM}})\diff u,\notag
\end{align}
by integrating this identity, we get from Lemma~\ref{CorI},
\begin{equation}\label{uaux2}
\E(M)=\delta_+\int_{\R{\times}\R_+^2} \ind{v{\le}u{<}v{+}w} {k_{\iM}}\diff u {E_{\iM}}(\diff v){L_{\iM}}(\diff w)=\delta_+{k_{\iM}}\E({L_{\iM}}).
\end{equation}
From Representation~\eqref{EqEqM2} of $M$ and Relation~\eqref{PoisVar} of Corollary~\ref{PoisMom}, we obtain
\[
\E\left(M^2{\mid} {\cal F}^I\right)=\left(\E\left(M{\mid}{\cal F}^I\right)\right)^2
 +\int_{\R{\times}\R_+^2} I(u)\ind{v{\le}u{<}v{+}w} {k_{\iM}}\diff u {E_{\iM}}(\diff v){L_{\iM}}(\diff w),
 \]
and, after integration,
 \begin{equation}\label{uaux3}
\E\left(M^2\right)=\E\left(\E\left(M{\mid}{\cal F}^I\right)^2\right) +\E(M).
 \end{equation}
 Define
 \[
 \Delta^I(M)\steq{def}  \E\left(\E(M{\mid}{\cal F}^I)^2\right)-\E(M)^2,
 \]
 so that $\Var(M){=}\Delta^I(M){+}\E(M)$. Relations~\eqref{uaux1} and~\eqref{uaux2} give
\[
\Delta^I(M)=  \E\left(\left(\E(M{\mid}{\cal F}^I) -\E(M)\right)^2\right) =
   \E\left[\left({k_{\iM}}\int_{\R} (I(u){-}\delta_+)\P({E_{\iM}}{\le}u{<}{E_{\iM}}{+}{L_{\iM}})\diff u \right)^2\right].
\]
We can rewrite the square of the integral as a double integral in the following way,
\begin{multline*}
  \left(\int_{\R} (I(u){-}\delta_+)\P({E_{\iM}}{\le}u{<}{E_{\iM}}{+}{L_{\iM}})\diff u \right)^2\\=
  \int_{\R^2} (I(u){-}\delta_+)(I(u'){-}\delta_+)\P({E_{\iM,1}}{\le}u{<}{E_{\iM,1}}{+}{L_{\iM,1}})\P({E_{\iM,2}}{\le}u'{<}{E_{\iM,2}}{+}{L_{\iM,2}})\diff u\diff u',
\end{multline*}
where the variables $E_{\iM,i}$, $i{=}1$, $2$, and $L_{\iM,i}$, $i{=}1$, $2$, are independent. The quantity $\Delta^I(M)$ can be expressed as 
\begin{multline*}
\Delta^I(M)={k_{\iM}}^2\int_{\R^2} \E\left[\rule{0mm}{4mm}(I(u){-}\delta_+)(I(u'){-}\delta_+)\right]\\{\times}\P(E_{\iM,1}{\le}u{<}E_{\iM,1}{+}{L_{\iM}}, E_{\iM,2}{\le}u'{<}E_{\iM,2}{+}L_{\iM,2})\diff u \diff u'.
\end{multline*}
From Lemma~\ref{CorI} we get, for $u$, $u'{\in}\R_+$,
\begin{equation}\label{CorII}
\E[(I(u){-}\delta_+)(I(u'){-}\delta_+)]=\E(I(u)I(u')){-}\delta_+^2=\delta_+(1{-}\delta_+)e^{-\Lambda |u{-}u'|},
\end{equation}
hence the quantity $\Delta^I(M)/({k_{\iM}}^2 \delta_+(1{-}\delta_+))$ is given by 
\begin{align*}
  \int_{\R^2} e^{-\Lambda |u{-}u'|}&\P(E_{\iM,1}{<}u{<}E_{\iM,1}{+}L_{\iM,1}, E_{\iM,2}{\le}u'{<}E_{\iM,2}{+}L_{\iM,2})\diff u \diff u'\\
  &=\E\left(\int_{\R_+^2} e^{-\Lambda |u{+}E_{\iM,1}{-}(u'{+} E_{\iM,2})|}\P(u{\le}L_{\iM,1},u'{<}L_{\iM,2})\diff u \diff u'\right)\\
  &=\E({L_{\iM}}^2)\E\left(e^{-\Lambda |F_{{L_{\iM}},1}{+}E_{\iM,1}{-}(F_{{L_{\iM}},2}{+} E_{\iM,2})|}\right).
\end{align*}
If we gather these results into Relation~\eqref{uaux3}, we obtain the identity
\[
\frac{\Var(M)}{\E(M)}=1{+}\E(M)\frac{1{-}\delta_+}{\delta_+}\E\left(e^{-\Lambda |F_{{L_{\iM}},1}{+}E_{\iM,1}{-}(F_{{L_{\iM}},2}{+} E_{\iM,2})|}\right),
\]
which is the desired formula. 
\end{proof}
We now turn to a more detailed analysis of the invariant distribution of $(M(t))$ in the case when the gene is always active. 
\subsection{The Equilibrium Distribution of the State of mRNAs}
It is now assumed that the process $(I(t))$ is constant and equal to $1$. In particular $\delta_+{=}1$. We introduce a random measure that describes precisely the state of the mRNAs,
\[
\Lambda_{\iM}(t){\steq{def}}\sum_{k=1}^{M(t)} \delta_{R_i(t)}\text{ and }
\]
where $R_i(t)$ is the residual lifetime of the $i$th mRNA present at time $t$. The measure $\Lambda_{\iM}(t)$ is the {\em empirical distribution} associated to the residual lifetimes. 
\begin{prop}[Convergence to Equilibrium of the number of  mRNAs]\label{mRNAprop}
When the gene is always active, the process $(\Lambda_{\iM}(t))$ converges in distribution to a random measure $(\Lambda_{\iM}^*)$ defined by, for $f$ a continuous function with compact support on $\R_+$,
  \[
  \croc{\Lambda_{\iM}^*,f}=\int_{\cal S} f(v{+}w{-}u)\ind{v{\le}u{<}v{+}w} {\cal P}(\diff u,\diff v, \diff w).
 \]
\end{prop}
\begin{proof}
  The proof is similar to the proof of Theorem~\ref{theoCVMP}.  Theorem~3.26 of Dawson~\cite{Dawson} is used, it is enough to show that for the convergence in distribution
  \[
  \lim_{t{\to}+\infty}   \croc{\Lambda_{\iM}(t),f}=  \croc{\Lambda_{\iM}^*,f},
  \]
for every non{-}negative continuous function $f$ with compact support on $\R_+$.
 By definition of the vector $(R_i(t))$,
\begin{align*}
  \croc{\Lambda_{\iM}(t),f}&=\sum_{n} f(u_n{+}E_{\iM,n}{+}L_{\iM,n}{-}t) \ind{0{<}u_n, u_n{+}E_{\iM,n}{\le}t{<} u_n{+}E_{\iM,n}{+}L_{\iM,n}}\\
&=\int_{\cal S} f(u{+}v{+}w{-}t)\ind{0{<}u, u{+}v{\le}t{<}u{+}v{+}w} {\cal P}(\diff u,\diff v, \diff w)\\
  &{\steq{dist}}\int_{\cal S} f(u{+}v{+}w)\ind{{-}t{<}u,u{+}v{\le}0{<}u{+}v{+}w} {\cal P}(\diff u,\diff v, \diff w),
\end{align*}
by invariance of the Poisson point process with respect to translation by ${-}t$, see Proposition~\ref{PoisSym}. We conclude that $(\Lambda_{\iM}(t))$ converges in distribution to $\Lambda$ defined by
\begin{multline*}
  \croc{\Lambda,f}{=}\int_{\cal S} f(u{+}v{+}w)\ind{u{+}v{\le}0{<}u{+}v{+}w} {\cal P}(\diff u,\diff v, \diff w)\\
 {\steq{dist}}\int_{\cal S} f(v{+}w{-}u)\ind{{-}u{+}v{\le}0{<}{-}u{+}v{+}w} {\cal P}(\diff u,\diff v, \diff w)=\croc{\Lambda_{\iM}^*,f},
\end{multline*}
by invariance with respect to the mapping $(u,v,w){\mapsto}(-u,v,w)$, again by using  Proposition~\ref{PoisSym}. The proposition is proved. 
\end{proof}
This proposition states that at equilibrium the number of mRNAs is Poisson with parameter ${k_{\iM}}\E({L_{\iM}})$ and the residual lifetimes of the mRNAs are independent and distributed as $F_{{\sigma}_{\iM}}$. 
\begin{prop}[Equilibrium of the state of mRNAs]\label{theoMemp}
When the gene is always active,
  \begin{enumerate}
    \item the distribution of the variable $M$, the number of mRNAs at equilibrium, is a Poisson distribution with parameter $\E(M)$ given by 
  \[
  \E(M)={k_{\iM}}\E({L_{\iM}}).
  \]
\item The residual lifetimes of the mRNAs are i.i.d. with common distribution
  \[
  F_{{L_{\iM}}}(\diff x)=\frac{\P({L_{\iM}}{>}x)}{\E({L_{\iM}})}\diff x.
  \]
  \end{enumerate}
\end{prop}
\begin{proof}
We  calculate the Laplace transform of $\Lambda_{\iM}^*$, by using Proposition~\ref{Pois-lapois}  of the appendix and the fact that $\nu_{\cal P}$ defined by Relation~\eqref{nup} is the intensity measure of the Poisson process ${\cal P}$, we have
\begin{align*}
  \E\left(e^{{-}\croc{\Lambda_{\iM}^*,f}}\right)&= \exp\left[-\int_{{\cal S}} \left(1-e^{-f(v{+}w{-}u)}\right)\,\nu_{\cal P} (\diff u, \diff v,\diff w)\right]\\
  &= \exp\left[{-}{k_{\iM}}\int_{\cal S} \left(1{-}e^{-f(v{+}w{-}u)}\right)\ind{v{\le}u{<}v{+}w}\,\diff u {E_{\iM}}(\diff v){L_{\iM}}(\diff w)\right]\\
  &= \exp\left[{-}{k_{\iM}}\E({L_{\iM}})\left(1{-}\E\left(e^{-f(F_{{L_{\iM}}})}\right)\right)\right]
 \end{align*}
By taking $f\equiv -\log(z)$ for some fixed $z{\in}(0,1)$, we obtain that
\[
\E\left(z^M\right)=\exp\left[{-}{k_{\iM}}\E({L_{\iM}})\left(1{-}z\right)\right],
\]
the distribution of $M$ is Poisson with parameter $\E(M){=}{k_{\iM}}\E({L_{\iM}})$. Additionally, since
\[
\E\left(e^{{-}\croc{\Lambda_{\iM}^*,f}}\right)= \sum_{n=0}^{+\infty} \frac{\E(M)^n}{n!} e^{{-}\E(M)}\left(\E\left(e^{-f(F_{{L_{\iM}}})}\right)\right)^n.
\]
we have that $\Lambda_{\iM}^*$ has the same Laplace transform, and therefore the same distribution,  as the random measure 
\[
\sum_{i=1}^{M} \delta_{F_{{L_{\iM}},i}},
\]
where $(F_{{L_{\iM}},i})$ is an i.i.d. sequence of random variables with distribution $F_{{L_{\iM}}}$ independent of $M$. The proposition is proved. 
\end{proof}
\section{Translation} \label{TranslSec}
Theorem~\ref{propEq} shows that the distribution of the number of proteins at equilibrium has the same distribution as the random variable
\begin{equation}\label{EqP2}
  P=\int_{{\cal S}} I(u)G_{\iP}(u,v,w,m) {\cal P}(\diff u,\diff v, \diff w, \diff m),
\end{equation}
recall that ${\cal S}{=}\R{\times}\R_+^2{\times}{\cal M}_p(\R{\times}\R_+^2)$, and 
\begin{equation}\label{EqPG2}
G_{\iP}(u,v,w,m)\steq{def} \int_{\R{\times}\R_+^2}\ind{\substack{x{+}v{\le}u{<}x{+}v{+}w\\y{\le} x{<}y{+}z}}m(\diff x,\diff y,\diff z)
\end{equation}
where ${\cal P}$ is the marked Poisson point process on ${\cal S}$ defined by Relation~\eqref{PoisP}.

\subsection{Mean and Variance of the Number of Proteins at Equilibrium}\ \\
We have an explicit representation of the variance of the number of proteins at equilibrium given by the following proposition. This is an extension of the results of Fromion et al.~\cite{Fromion}, see also Leoncini~\cite{Leoncini}.
\begin{prop}[The two first moments of the number of proteins]\label{FTMP}
 If $P$ is the number of proteins at equilibrium, then 
  \begin{equation}\label{meanP}
\E(P)=\frac{k_+}{\Lambda} {k_{\iM}}\E({L_{\iM}}){k_{\iP}}\E({L_{\iP}}),  
  \end{equation}
  and
  \begin{multline}\label{varP}
    \frac{\Var(P)}{\E(P)}=1{+}  {k_{\iP}}\E({L_{\iP}}) \P\left(\rule{0mm}{4mm}F_{{L_{\iM}}}{\ge}|E_{\iP,1}{-}E_{\iP,2}{+}F_{{L_{\iP}},1}{-}F_{{L_{\iP}},2}|\right)\\
    + \frac{k_-}{k_+}\E(P)\E\left(e^{-\Lambda |F_{{{L_{\iM}}},1}{-}F_{{{L_{\iM}}},2}{+}F_{{{L_{\iP}}},1}{-}F_{{{L_{\iP}}},2}{+}E_{\iM,1}{-}E_{\iM,2}{+}E_{\iP,1}{-}E_{\iP,2}|}\right),
  \end{multline}
  where $\Lambda{=}k_+{+}k_-$ and,  for $a{=}\iM$ and $\iP$,  the random variables $E_{a,1}$ and $E_{a,2}$ are independent  with  distribution $E_{a}$,  and  $F_{L_a,1}$ and $F_{L_a,2}$ are independent random variables with density $\P(L_a{\ge}u)/\E(L_a)$ on $\R_+$.
\end{prop}
The proof is given in  Section~\ref{ProofFTMP} of Appendix. It follows the same arguments as in the proof of Proposition~\ref{propM} but with a more technical framework due to the mark $m$, in the  space ${\cal M}_p(\R{\times}\R_+^2)$, of the Poisson process ${\cal P}$.

If the average of $P$ does not depend on the distributions of elongation times of mRNAs and  of proteins, the second moment does depend on these distributions, and also on the distributions of lifetimes. Note however that if $\Lambda$ is large, i.e. at least one of the states of the gene is changing rapidly, Relation~\eqref{varP} shows that  the dependence of the distribution of the elongation time of an mRNA on the variance is small in this case.

When the elongation times are null, i.e.\ $E_\iM{\equiv}0$ and $E_\iP{\equiv}0$ and the lifetimes $L_\iM$ and $L_\iP$ are exponentially distributed with respective parameters $\gamma_\iM$ and $\gamma_\iP$, in this case, for $a{\in}\{\iM,\iP\}$, the variable $F_{L_a}$ is also exponential since $F_{L_a}{\steq{dist}}L_a$, a simple calculation gives the  classical formula~\eqref{PaulVar} of the Markovian model. 

Outside the  Fano parameter, the biological literature defines the {\em noise } associated to the production of proteins as the variance of $P/\E(P)$, the quantity $\Var(P)/\E(P)^2$. The above formulas and Relation~\eqref{EM} give that it can be represented as
  \begin{multline*}
    \frac{\Var(P)}{\E(P)^2}=\frac{1}{\E(P)} {+} \frac{1}{\E(M)}\P\left(\rule{0mm}{4mm}F_{{L_{\iM}}}{\ge}|E_{\iP,1}{-}E_{\iP,2}{+}F_{{L_{\iP}},1}{-}F_{{L_{\iP}},2}|\right)\\
    + \frac{k_-}{k_+}\E\left(e^{-\Lambda |F_{{{L_{\iM}}},1}{-}F_{{{L_{\iM}}},2}{+}F_{{{L_{\iP}}},1}{-}F_{{{L_{\iP}}},2}{+}E_{\iM,1}{-}E_{\iM,2}{+}E_{\iP,1}{-}E_{\iP,2}|}\right).
  \end{multline*}

  \subsection{The Equilibrium Distribution of the Number of Proteins}\ \\ 

  \noindent
When the gene is always active,  an explicit expression of the generating function of $P$ can be obtained. As it will be seen, its form depends on the whole distribution of the lifetimes and elongation times. Nevertheless, by using appropriate scalings, it can be used to get limit results for its distribution and therefore some insight. This is the purpose of  Section~\ref{Sec-AppP}.

Note that for classical Markovian models, without elongation times in particular,   Relation~(26) of Bokes et al.~\cite{Bokes} gives,  via an analytic approach,  an explicit expression of the joint  generating function of $M$ in terms of an hypergeometric function.

\begin{prop}[Generating Function of $P$]\label{GenP}
If the gene is always active, i.e.\  the process $(I(t))$ is constant and equal to $1$, and  $P$ is the random variable defined by Relation~\eqref{EqEqP}, the distribution of $P$ is given by, for $z{\in}[0,1]$,
\[
  \E\left(z^{P}\right)
=\exp\left({-}\int_{\R_+^2}\left[\rule{0mm}{4mm}1{-}e^{{-}\left( 1{-}z\right)k_{\iP}\E(L_{\iP})\P\left(F_{L_{\iP}}{+}E_{\iP}{\in}(u{-}w,u)\right)}\right] k_{\iM}\diff u L_{\iM}(\diff w)\right),
\]
where $F_{L_{\iP}}$ and $L_{\iP}$  are two independent random variables and the distribution of $F_{L_{\iP}}$ has density $\P(L_{\iP}{\ge}u)/\E(L_{\iP})$ on $\R_+$.
\end{prop}
\begin{proof}
Relations~\eqref{EqP2} and~\eqref{EqPG2} defining $P$ and~\eqref{Pois-lapois}  of Appendix for the Laplace transform of Poisson point processes give
  \begin{align*}
    \E\left(\exp\left(-\xi P\right)\right)
    &=\E\left(\exp\left(-\xi \int_{\cal S}G_{\iP}(u,v,w,m){\cal P}(\diff u, \diff v,\diff w,\diff m)\right)\right)\\
  & =\exp\left({-}\int_{\cal S} \left[1{-}e^{-\xi G_{\iP}(u,v,w,m) }\right] \nu_{\cal P}(\diff u,\diff v,\diff w,\diff m)\right)\\
&=   \int_{\cal S} \left (1{-}\E\left( e^{-\xi G_{\iP}(u,v,w,{\cal N}_{\iP}) }\right)\right) k_{\iM}\diff u E_{\iM}(\diff v)L_{\iM}(\diff w),
  \end{align*}
  where, as before,  ${\cal N}_{\iP}$ is a Poisson process whose distribution is $Q_\iP$.  For $(u,v,w){\in}\R{\times}\R_+^2$, by using again Relation~\eqref{Pois-lapois} for the Laplace transform of the Poisson process ${\cal N}_\iP$,
\begin{align}
  \E\left(e^{-\xi G_{\iP}(u,v,w,{\cal N}_{\iP}) }\right)\label{Dimaux}\\
&\hspace{-1cm}=\E\left(\exp\left({-}\xi\int_{\R{\times}\R_+^2}\ind{\substack{x{+}v{\le}u{<}x{+}v{+}w \\      y{\le} x{<}y{+}z}}{\cal N}_{\iP}(\diff x, \diff y,\diff z)\right)\right)\notag \\ 
  &\hspace{-1cm} =\exp\left({-}\left( 1{-}e^{-\xi}\right)\int_{\R{\times}\R_+^2}\ind{\substack{x{+}v{\le}u{<}x{+}v{+}w\\y{\le} x{<}y{+}z}}k_{\iP}\diff xE_{\iP}(\diff y)L_{\iP}(\diff z)\right)\notag \\
  &\hspace{-1cm}=\exp\left({-}\left( 1{-}e^{-\xi}\right)\int_{\R_+}\P\left(\substack{x{+}E_{\iP}{+}v{<}u{<}x{+}E_{\iP}{+}v{+}w\\0{\le} x{<}L_{\iP}}\right)k_{\iP}\diff x\right)\notag \\
& \hspace{-1cm} =\exp\left({-}\left( 1{-}e^{-\xi}\right)k_{\iP}\E(L_{\iP})\P\left(F_{L_{\iP}}{+}E_{\iP}{+}v{\le}u{<}F_{L_{\iP}}{+}E_{\iP}{+}v{+}w\right)\right),\notag 
\end{align}
with the same arguments as in the proof of Relation~\eqref{VarM}. Hence,
  \begin{multline*}
    \int_{\cal S} \left[1{-}e^{-\xi G_{\iP}(u,v,w,m) }\right] \nu_{\cal P}(\diff u,\diff v,\diff w,\diff m)\\
    =\int_{\cal S}\left[1{-}\exp\left({-}\left( 1{-}e^{-\xi}\right)k_{\iP}\E(L_{\iP})\P\left(u{-}v{-}w{\le}F_{L_{\iP}}{+}E_{\iP}{<}u{-}v\right)\right)\right] \\ \hfill  k_{\iM}\diff u E_{\iM}(\diff v) L_{\iM}(\diff w)\\
       =\int_{\cal S}\left[1{-}\exp\left({-}\left( 1{-}e^{-\xi}\right)k_{\iP}\E(L_{\iP})\P\left(u{-}w{\le}F_{L_{\iP}}{+}E_{\iP}{<}u\right)\right)\right] k_{\iM}\diff u  L_{\iM}(\diff w),
  \end{multline*}
which gives the desired formula for the generating function. 
\end{proof}
The explicit expression of the generating function of $P$  for Markovian models   is detailed  in Proposition~\ref{CorGenP} of Section~\ref{MarkSec}.  

\subsection{The Joint Distribution of the Numbers of mRNAs and Proteins}\ \\
This topic is investigated in few  references of the literature such as  Bokes et al.~\cite{Bokes} and Taniguchi et al.~\cite{Taniguchi}. Its motivation lies in the fact that it could be useful if some information on the number of mRNAs could be extracted from  knowledge on the number of proteins. 

We begin to study the covariance function of $M$ and $P$. Relations~\eqref{EqEqM2} and~\eqref{EqP2} give the representations
\begin{equation}\label{eqMPdef}
  \begin{cases}
M\displaystyle =\int_{\cal S} I(u) G_{\iM}(u,v,w)\, {\cal P}(\diff u,\diff v, \diff w),\\\ \\
P\displaystyle =\int_{{\cal S}} I(u)G_{\iP}(u,v,w,m)\, {\cal P}(\diff u,\diff v, \diff w, \diff m),
  \end{cases}
\end{equation}
  where the functions $G_\iM$ and $G_\iP$ are defined by Relations~\eqref{EqMG1} and~\eqref{EqPG2}.
By using Proposition~\ref{CondI} as before on the conditional distribution of $(M,P)$ with respect to ${\cal F}^I$,  the $\sigma$-field associated to the state of the gene. See Section~\ref{mRNAsec}. Relation~\eqref{PoisVar} for the covariance of functionals of Poisson processes gives the formula
\begin{multline*}
C^I\steq{def}  \E\left(MP|{\cal F}^I\right){-}\E\left(M|{\cal F}^I\right)\E\left(P|{\cal F}^I\right)\\
  =\int_{\cal S} I(u)G_{\iM}(u,v,w,m)G_{\iP}(u,v,w,m)\nu_{\cal P}(\diff u,\diff v, \diff w,\diff m).
\end{multline*}
By integrating this identity and Relations~\eqref{EqMG1} and ~\eqref{EqPG2}  defining $G_\iM$ and $G_\iP$, we obtain that $\E(C^I)$ is equal to 
\begin{align*}
\int_{\cal S} \E(I(u))\ind{v{\le}u{<}v{+}w}&\int_{\R{\times}\R_+^2}\ind{\substack{x{+}v{\le}u{<}x{+}v{+}w \\y{\le} x{<}y{+}z}}m(\diff x,\diff y,\diff z)\nu_{\cal P}(\diff u,\diff v, \diff w,\diff m)\\
=\delta_+\int_{\cal S} \ind{v{\le}u{<}v{+}w}&\E\left(\int_{\R{\times}\R_+^2}\ind{\substack{x{+}v{\le}u{<}x{+}v{+}w \\y{\le} x{<}y{+}z}}{\cal N}_\iP(\diff x,\diff y,\diff z)\right)k_\iM \diff u E_\iM(\diff v)L_\iM(\diff w)\\
&= \delta_+{k_{\iM}}{k_{\iP}}\E\left(\int\ind{0{\le}u{<}{L_{\iM}}} \int_{\R^2}\ind{\substack{x{\le}u{<}x{+}{L_{\iM}}\\{E_{\iP}}{\le} x{<}{E_{\iP}}{+}{L_{\iP}}}}\diff x\diff u\right)\\
&=\delta_+{k_{\iM}}{k_{\iP}}\E(L_\iP)\E\left(\left(L_{\iM}{-}F_{L_{\iP}}-E_{\iP}\right)^+\right),
\end{align*}
hence, with Relation~\eqref{eqFa},
\[
\E(C^I)= \E(P)\P\left(F_{L_{\iM}}{\ge}F_{L_{\iP}}{+}E_{\iP}\right).
\]
Concerning the product of conditional expectations in the definition of $C^I$, the following identities have established in the proof of Propositions~\ref{propM} and~\ref{FTMP},
\begin{align*}
\E\left(M{\mid} {\cal F}^I\right) &={k_{\iM}}\int_{\R} I(u)\P({E_{\iM}}{<}u{<}{E_{\iM}}{+}{L_{\iM}})\diff u,\\
\E\left(P{\mid}{\cal F}^I\right) &= {k_{\iM}}{k_{\iP}}\int_{\R}I(u)\E\left(\int_{\R}\ind{\substack{x{+}{E_{\iM}}{\le}u{<}x{+}{E_{\iM}}{+}{L_{\iM}}\\{E_{\iP}}{\le} x{<}{E_{\iP}}{+}{L_{\iP}}}}\diff x\right)\diff u.
\end{align*}
Using the same argument as before to express a product of integrals as a double integral, this gives the identity
\begin{multline*}
\E\left(\E\left(M|{\cal F}^I\right)\E\left(P|{\cal F}^I\right)\right)/(k_{\iM}^2k_{\iP})\\=
        \int_{\R^3} \E\left(I(u)I(u')\right)\E\left(\ind{E_{\iM,1}{\le}u{<}E_{\iM,1}{+}L_{\iM,1}}
         \ind{\substack{x{+}{E_{\iM,2}}{\le}u'{<}x{+}{E_{\iM,2}}{+}{L_{\iM,2}}\\{E_{\iP}}{\le} x{<}{E_{\iP}}{+}{L_{\iP}}}}\right)\diff x\diff u'\diff u.
\end{multline*}
By using Corollary~\ref{CorI} again, $\E[I(u)I(u')]{=}\delta_+(\delta_+{+}(1{-}\delta_+)\exp(\Lambda |u{-}u'|)$, for $u$, $u'{\ge}0$, the last expression is the sum of two terms,
\[
         \delta_+^2\int_{\R^3}\E\left(\ind{0{\le}u{<}L_{\iM,1}}
         \ind{\substack{0{\le}u'{<}{L_{\iM,2}}\\0{\le} x{<}{L_{\iP}}}}\right)\diff x\diff u'\diff u
=         \delta_+^2\E(L_\iM)^2\E(L_\iP),
\]
and
\[
         \delta_+(1{-}\delta_+)\int_{\R^3} e^{{-}\Lambda|u{-}u'|} \E\left(\ind{E_{\iM,1}{\le}u{<}E_{\iM,1}{+}L_{\iM,1}}
         \ind{\substack{x{+}{E_{\iM,2}}{\le}u'{<}x{+}{E_{\iM,2}}{+}{L_{\iM,2}}\\{E_{\iP}}{\le} x{<}{E_{\iP}}{+}{L_{\iP}}}}\right)\diff x\diff u'\diff u,
\]
which is, after the same manipulations as in the proof of Proposition~\ref{FTMP}, 
\[
         \delta_+(1{-}\delta_+)\E(L_\iM)^2\E(L_\iP)\E\left(e^{{-}\Lambda| F_{L_{\iM},2}{-}F_{L_{\iM},1}{+}F_{L_{\iP}}{+}E_{\iP}{+}{E_{\iM,2}{-}E_{\iM,1}})|}\right).
         \]
We have therefore proved the following proposition. 
\begin{prop}[Covariance of $M$ and $P$]
If $M$ and $P$ are the random variables defined by Relation~\eqref{EqEqM} and~\eqref{EqEqP}, then the covariance of $M$ and $P$ is given by 
\begin{multline}\label{CovMP}
\frac{\Cov(M,P)}{\E(P)} = \P\left(F_{L_{\iM}}{\ge}F_{L_{\iP}}{+}E_{\iP}\right){+}\E(M)
\\+\frac{k_-}{k_+}\E(M)\E\left(e^{{-}\Lambda| F_{L_{\iM},2}{-}F_{L_{\iM},1}{+}{E_{\iM,2}{-}E_{\iM,1}}{+}F_{L_{\iP}}{+}E_{\iP}|}\right),
\end{multline}
  for  $E_{\iM,1}$ and $E_{\iM,2}$, are independent random variables with distribution $E_\iM$,  and   $F_{L_\iM,1}$ and $F_{L_\iM,2}$ are independent random variables with density $\P(L_\iM{\ge}u)/\E(L_\iM)$ on $\R_+$, $F_{L_{\iP}}$ is  a random variable with density $\P(L_\iP{\ge}u)/\E(L_\iP)$ on $\R_+$.
\end{prop}

The lifetime $L_\iM$ of the mRNAs is usually much smaller than the lifetime $L_\iP$ of proteins, this implies that the same property also holds for the variables $F_{L_{\iM}}$ and  $F_{L_{\iP}}$. If, additionally, the average number of mRNAs $\E(M)$ is small,  the right-hand-side of Relation~\eqref{CovMP} is small, so that $\Cov(M,P)$ should be close to $0$. This property has already been noticed in some measurements, see Taniguchi et al.~\cite{Taniguchi} for example.  We conclude this section with an explicit representation of the generating function of the random variable $(M,P)$ when the gene is always active. 
\begin{prop}[Joint Distribution of $M$ and $P$]
When the gene is always active, \\  for $z_\iM$ and $z_\iP{\in}(0,1)$,
\begin{multline}
\E\left(z_\iM^Mz_\iP^P\right)\\ = \exp\left[\rule{0mm}{6mm} {-}k_\iM\E(L_\iM) \int_{\R_+} \left(1{-} z_\iM\exp\left[{-}\left( 1{-}z_\iP\right)k_{\iP}\E(L_{\iP})\P\left(F_{L_{\iP}}{+}E_{\iP}{\le}u\right)\rule{0mm}{5mm}\right]\right)F_{L_\iM}(\diff u) \right.\\\left. 
+k_\iM \int_{\R_+^2} \left(1{-} \exp\left[{-}\left( 1{-}z_\iP\right)k_{\iP}\E(L_{\iP})\P\left(F_{L_{\iP}}{+}E_{\iP}{\in}(u, u{+}w) \right)\rule{0mm}{5mm}\right]\right) \diff u  L_\iM(\diff w)\rule{0mm}{6mm}\right],
\end{multline}
where $F_{L_{\iM}}$, $F_{L_{\iP}}$ and $E_{\iP}$  are independent random variables and, for $a{\in}\{\iM,\iP\}$, the random variable $F_{L_{a}}$ has the density $\P(L_{a}{\ge}u)/\E(L_{a})$ on $\R_+$.
\end{prop}
\begin{proof}
For $a_\iM$, $a_\iP{>}0$,  Relation~\eqref{eqMPdef} gives that
  \[
  a_\iM M{+}a_\iP P= \int_{\cal S} F(u,v,w,m) \,{\cal P}(\diff u, \diff v,\diff w,\diff m)
  \]
where
\begin{align*}
  F(u,v,w,m)&\steq{def} a_\iM G_{\iM}(u,v,w)+ a_\iP G_{\iP}(u,v,w,m)\\&=a_\iM\ind{v{\le}u{<}v{+}w}
  {+}a_\iP \int_{\R{\times}\R_+^2}\ind{\substack{x{+}v{\le}u{<}x{+}v{+}w\\y{\le} x{<}y{+}z}}m(\diff x,\diff y,\diff z).
\end{align*}
We have, by Relation~\eqref{Pois-lapois}  of the appendix,
\[
\E\left(e^{-a_\iM M-a_\iP P}\right){=} \exp\left({-}\int_{\cal S} \left(1{-} e^{{-}F(u,v,w,m)}\right) \,\nu_{\cal P}(\diff u, \diff v,\diff w,\diff m)\right),
\]
hence,
\begin{multline*}
{-}\log\left(\E\left(e^{-a_\iM M-a_\iP P}\right)\right)\\
{=}\int_{\R{\times}\R_+^2} \left(1{-}e^{{-}a_\iM G_\iM(u,v,w)} \E\left( e^{-a_\iP G_\iP(u,v,w,{\cal N}_\iP)}\right)\right)k_\iM \diff u E_\iM(\diff v) L_\iM(\diff w),
\end{multline*}
and, if we use Relation~\eqref{Dimaux},
\[
\E(\left(e^{-a_\iP G_{\iP}(u,v,w,{\cal N}_2)) }\right)
{=}\exp\left({-}\left( 1{-}e^{-a_\iP}\right)k_{\iP}\E(L_{\iP})\P\left(F_{L_{\iP}}{+}E_{\iP}{\in}(u{-}v{-}w, u{-}v) \right)\right),
\]
hence
\begin{multline*}
{-}\log\left(\E\left(e^{-a_\iM M-a_\iP P}\right)\right)\\
{=}\int_{\R{\times}\R_+^2} \left(1{-} \exp\left[{-}a_\iM\ind{v{\le}u{<}v{+}w}{-}\left( 1{-}e^{-a_\iP}\right)k_{\iP}\E(L_{\iP})\P\left(F_{L_{\iP}}{+}E_{\iP}{\in}(u{-}v{-}w, u{-}v) \right)\right]\right)\\k_\iM \diff u E_\iM(\diff v) L_\iM(\diff w).
\end{multline*}
We conclude the proof with standard calculations.
\end{proof}

\section{Convergence Results}\label{Sec-AppP}
The purpose of this section is of revisiting, via rigorous convergence theorems,  several classical research topics of the literature on the stochasticity of the gene expression. Most of studies in this domain take place in the Markovian setting of Section~\ref{MarkSec}  without the elongation of mRNAs and proteins. Kolmogorov's equations associated to the Markov process, the ``master equation'' as it is usually presented, are the starting point of these studies.

As the equilibrium equations cannot be solved explicitly, several scenarios  are investigated, via scalings,  to get insight on the equilibrium distribution of the number of proteins:
Fast Switching rates between active and inactive states, mRNAs with short lifetimes, proteins with long lifetimes,  gene with large expression rates, \ldots 
See Friedman et al.~\cite{Friedman}, Raj et al.~\cite{Raj}, Shahrezaei and Swain~\cite{Swain2}, Swain et al.~\cite{Swain} and Thattai and van Oudenaarden~\cite{Thattai}.  See Bokes et al.~\cite{Bokes} for a quite extensive analysis of the Markovian model.  We will look at three different scaling situations in the light of the results derived in the previous sections. 
\begin{enumerate}
\item  Long{-}Lived Proteins.\\
  The distribution of lifetimes of proteins is an exponential distribution with parameter $\gamma_\iP$, and its mean is converging to infinity. A central limit theorem is proved for the equilibrium of the  number $P$ of proteins.  It is shown that, in distribution,  $P{\sim} \E(P){+}\sqrt{\E(P)}{\cal N}$, where ${\cal N}$ is a centered Gaussian random variable. 
\item[]
\item Short{-}Lived mRNAs and  Long-Lived Proteins.\\
The distribution of  lifetimes of  mRNAs  and proteins are  exponential with respective  parameters $\gamma_\iM$ and $\gamma_\iP$, the mean lifetime  $1/\gamma_\iM$ of mRNAs is converging to $0$ and the death rate $\gamma_\iP$ of proteins is fixed so that  $k_\iM k_\iP/(\gamma_\iM \gamma_\iP)$, the mean  number of proteins at equilibrium, is fixed.
\item[]
\item  Short{-}Lived mRNAs and High Expression Rate of Proteins. \\
  The distribution of  lifetimes of  mRNAs  and proteins are  exponential with respective  parameters $\gamma_\iM$ and $\gamma_\iP$.  The mean lifetime  of mRNAs is converging to $0$ and the translation rate  $k_\iP$ of proteins is fixed so that  the quantity $k_\iP/\gamma_\iM$, the average number of proteins produced by an mRNA, is fixed. 
\end{enumerate}
For several of these regimes, in a Markovian setting, convergence results of this section can be seen in the context of a stochastic averaging principle: the system is driven by a fast process, associated with the mRNAs (for example), and by a slow process describing the time evolution of proteins. In the limit, the time evolution on finite time intervals of proteins can be described as if the fast process is at equilibrium at any time. A typical result states that this picture also holds for the equilibrium of the slow process,  but this is a more difficult result to prove in general. See Kurtz~\cite{Kurtz} for example.  A similar observation could also be done for the central limit theorem~\ref{CTL}, see Chapter~4 of Anderson and Kurtz~\cite{Anderson}.  A direct approach is used here to prove  convergence results by taking advantage of the explicit representation of equilibrium of Theorem~\ref{theoCVMP}.

In most cases, the gene will be assumed to be always active. The reason is that, if some results could be obtained on the distribution of $P$, we have not been able to get usable expressions to get some insight in the scaling regimes analyzed here. The next result shows that, in practice, a model with a permanently active gene is reasonable, provided some parameters are adjusted. This is a folk result of the biological literature: if the state of the gene switches more and more rapidly, then, in the limit, the model is equivalent to a model with a permanently active gene but with a reduced transcription rate. 

\subsection{Fast Switching Rates of Gene between Active and Inactive States}\ 

\noindent
Proposition~\ref{mRNAprop} shows that when the gene is always active then the equilibrium distribution of the number of mRNAs is Poisson. This is minimal from the point of view of the Fano factor in this class of models. If the gene switches quickly between the two states active/inactive, one can expect an averaging effect. The following proposition establishes this intuitive result.

\begin{prop}
If $P_N$ is the random variable defined by Relation~\eqref{EqEqP} with the activation/deactivation rates are respectively given by $k_{+}N$ and $k_{-}N$ for some scaling parameter $N{\ge}1$, then, for the convergence in distribution,
  \[
  \lim_{N\to+\infty} P_N=P,
  \]
  where the distribution of $P$ is the equilibrium distribution of the number of proteins for a model of gene expression for which the gene is always active and with the same parameters except for $k_\iM$ which is replaced by $k_\iM k_+/(k_+{+}k_-)$. 
\end{prop}
\begin{proof}
The stationary activation/deactivation process $(I_N(t))$ can clearly  be expressed as $(I(Nt))$ where $(I(t))$ is a stationary process with switching rates $k_+$ and $k_-$

Let $f$ be an integrable function on $\R_+$ and 
\[
X_N\steq{def} \int_{\R_+} I(Nu) f(u)\,\diff u,
\]
we now prove  the convergence in distribution
\[
  \lim_{N\to+\infty} X_N=\delta_+ \int f(u)\,\diff u,
\]
with $\delta_+{\steq{def}}k_+/(k_+{+}k_-)$.  Note that the convergence is certainly true for the first moment of $X_N$, since $\E(I(t)){=}\delta_+$ for all $t{\ge}0$.

By writing the square of $X_N$ as a double integral, with Proposition~\ref{PaulProp} we obtain the relation
\begin{multline*}
\E(X_N^2)=\E\left(\int_{\R_+^2} I(Nu) I(Nu') f(u) f(u')\,\diff u \diff u'\right)\\=\left(\delta_+\int f(u)\,\diff u\right)^2+\delta_+(1{-}\delta_+) 
\int_{\R_+^2} e^{-N(k_+{+}k_-) |u{-}u'|} f(u)f(u')\,\diff u \diff u'.
\end{multline*}
Therefore,
\[
\lim_{N\to+\infty} \E\left((X_N{-}\E(X_N))^2\right)=\lim_{N\to+\infty} \E\left(X_N^2\right){-}\left(\E(X_N)\right)^2=0,
\]
in particular, $(X_N)$ converges in distribution to the limit of the sequence $(\E(X_N))$.

Denote by ${\cal F}^{I_N}$ the filtration associated to the process $(I(Nt))$. Relations~\eqref{EqP2} and~\eqref{EqPG2} for $P_N$,  and~\eqref{Pois-lapois}  of the appendix for the Laplace transform of Poisson point processes give
\begin{align*}
\E\left(\exp\left(-\xi P_N\right)\left|{\cal F}^{I_N}\right)\right.
    &=\E\left(\exp\left(-\xi \int_{\cal S}I(Nu)G_{\iP}(u,v,w,m){\cal P}(\diff u, \diff v,\diff w,\diff m)\right)\right)\\
& =\exp\left({-}\int_{\cal S} I(Nu)\left[1{-}e^{-\xi G_{\iP}(u,v,w,m) }\right] \nu_{\cal P}(\diff u,\diff v,\diff w,\diff m)\right)\\
  & =\exp\left({-}\int_{\R} I(Nu)\left[1{-}\E\left(e^{-\xi G_{\iP}(u,E_\iM,L_\iM,{\cal N}_{\iP})}\right)\right] k_{\iM}\diff u\right),
\end{align*}
where ${\cal N}_{\iP}$ is a Poisson process whose distribution is $Q_\iP$ defined by Relation~\eqref{nup}.

By using the convergence result, we obtain that
\[
\lim_{N\to+\infty} \E\left(\exp\left(-\xi P_N\right)\right)=
\exp\left({-}\int_{\R} \left[1{-}\E\left(e^{-\xi G_{\iP}(u,E_\iM,L_\iM,{\cal N}_{\iP})}\right)\right] \delta_+k_{\iM}\diff u\right),
\]
which is the Laplace transform of the variable $P$ associated to a Poisson process with the same characteristics as ${\cal P}$ except that the transcription rate is $\delta_+ k_\iM$ and that the gene is always active.
\end{proof}

\subsection{Asymptotic Behavior of the Equilibrium Distribution}\label{AppLongSec}\

In this section, the invariant distribution of the number of proteins when the gene is always active is analyzed under some scaling conditions. Several of them rely on the fact that the average lifetime of a protein is much  larger than the lifetime of an mRNA, which is of the order of $2$mn for an mRNA, where, for a protein, it is at least $30$mn. See the numbers of  Section~\ref{NumbSec} of the appendix. This approach is used in the literature to get a further insight on the distribution of $P$, i.e.\ with more information than the first two moments that have an explicit expression. It is usually done via approximations on the equation satisfied by the generating function of $P$. See  Bokes et al.~\cite{Bokes} or Shahrezaei and Swain~\cite{Swain2} where this is done via an expansion of an hypergeometric function. The representation of Proposition~\ref{GenP} of the generating function of $P$ will allow to get convergence results for the distribution of $P$ in a quite general case and without too much technicality.

\bigskip

\noindent
{\bf A Limiting Gaussian Distribution for the Number of Proteins}\\
The scaling considered here assumes that the average lifetime  of proteins goes to infinity and the other parameters are fixed. In particular these proteins should be numerous within the cell.  This setting is well suited for protein types having a large number of copies, of the order at least  of 10,000 for example.  The following result shows that a central limit theorem holds in this context.
\begin{theorem}[Central Limit Theorem]\label{CTL}
  Under the conditions
  \begin{enumerate}
  \item the gene is always active;
  \item the distribution $L_\iP$ of lifetimes of proteins is  exponential with parameter $\gamma_\iP$;
  \item  the parameters $k_\iM$, $k_\iP$ are fixed as well as the distribution $L_\iM$ of the lifetimes of mRNAs and the respective distributions $E_\iM$ and $E_\iP$ of  the elongation times of mRNAs and proteins.  The distribution  $L_\iM$ of lifetimes of mRNAs has a finite third moment, $\E(L_\iM^3){<}{+}\infty$,
  \end{enumerate}
if $P$ is the random variable defined by Relation~\eqref{EqEqP}, then  for the convergence in distribution,
  \[
  \lim_{\gamma_\iP{\to}0} \sqrt{\gamma_\iP}\left(P{-}\E(P)\right)={\cal N}\left(\sigma\right),
  \]
  where
  \[
  \sigma{=}  \sqrt{k_\iM\E(L_\iM)k_\iP}\left(1{+}k_\iP\int_{\R_+} \P\left(F_{{L_{\iM}}}{\ge}|E_{\iP,1}{-}E_{\iP,2}{+}x|\right)\,\diff x\right)^{1/2},
  \]
  and, for $a{>}0$, ${\cal N}(a)$ is a centered Gaussian random variable with standard deviation~$a$. 
\end{theorem}
Despite this framework is quite natural, curiously it does not seem to have  been investigated  in the biological literature, even for the classical Markovian three-step model. It gives in fact  a  Gaussian approximation for the number of proteins at equilibrium when the average lifetime $1/\gamma_\iP$ of proteins is large. It should be noted that, in this scaling regime,  the second order does depend on the distribution of the elongation times of proteins.
\begin{proof}
By using Proposition~\ref{GenP}, the relation $F_{L_\iP}{\steq{dist}}L_\iP$ due to the exponential assumption for $L_\iP$, and the calculation of the proof of Relation~\eqref{meanP} of Proposition~\ref{FTMP}, we get that, for ${\xi}{>}0$, 
\[
\log\left[\E\left(\exp\left(\rule{0mm}{4mm}-\sqrt{\gamma_\iP}\xi (P{-}\E(P))\right)\right)\right]={-}\int_{\R_+^2} I_{\gamma_\iP}(u,w) k_{\iM} \diff uL_\iM(\diff w),
\]
with, for $(u,w){\in}\R_+^2$,
\[
I_{\gamma_\iP}(u,w){\steq{def}}  1{-}e^{{-}\left(1{-}e^{-\sqrt{\gamma_\iP}\xi}\right)\frac{k_{\iP}}{\gamma_\iP}\P\left(L_{\iP}{+}E_{\iP}{\in}(u{-}w,u)\right)} {-}\xi\frac{k_{\iP}}{\sqrt{\gamma_\iP}}\P\left(L_{\iP}{+}E_{\iP}{\in}(u{-}w,u)\right).
\]
The elementary inequality
\begin{equation}\label{elemineq}
\left|1{-}e^{-x}{-}x{+}\frac{x^2}{2}\right|\leq \frac{x^3}{6},\qquad\text{ for }
x{=}\left(1{-}e^{-\xi \sqrt{\gamma_\iP}}\right) \frac{k_{\iP}}{\gamma_\iP}\P\left(L_{\iP}{+}E_{\iP}{\in}(u{-}w,u)\right)
\end{equation}

gives the relation
\begin{multline}\label{Jaug1}
\left|\rule{0mm}{6mm}\right.I_{\gamma_\iP}(u,v)
           {-}\left(1{-}e^{-\xi \sqrt{\gamma_\iP}}{-}\xi\sqrt{\gamma_\iP}{+}\frac{1}{2}\xi^2\gamma_\iP\right) \frac{k_{\iP}}{\gamma_\iP}\P\left(L_{\iP}{+}E_{\iP}{\in}(u{-}w,u)\right)\\
           {+}\frac{1}{2}\xi^2\gamma_\iP\frac{k_{\iP}}{\gamma_\iP}\P\left(L_{\iP}{+}E_{\iP}{\in}(u{-}w,u)\right)
           \left.\rule{0mm}{6mm} {+}\frac{1}{2}\left(1{-}e^{-\xi\sqrt{\gamma_\iP}}\right)^2 \left(\frac{k_{\iP}}{\gamma_\iP}\right)^2\P\left(L_{\iP}{+}E_{\iP}{\in}(u{-}w,u)\right)^2\right|\\
           \leq \frac{1}{6}\left(1{-}e^{-\xi\sqrt{\gamma_\iP} }\right)^3\left(\frac{k_\iP}{\gamma_\iP}\right)^3\P\left(L_{\iP}{+}E_{\iP}{\in}(u{-}w,u)\right)^3.
\end{multline}
After integration with respect to the measure $k_{\iM}\diff u L_{\iM}(\diff w)$ on $\R_+^2$ and by using the relation
\[
\left|1{-}e^{-\xi \sqrt{\gamma_\iP}}{-}\xi\sqrt{\gamma_\iP}{+}\frac{1}{2}\xi^2\gamma_\iP\right|\le \frac{1}{6} \xi^3\gamma_\iP^{3/2},
\]
we get the inequality
\begin{multline}\label{dimeq1}
\left|\log\left[\E\left(\exp\left(\rule{0mm}{4mm}-\sqrt{\gamma_\iP}\xi (P{-}\E(P))\right)\right)\right] {-}A_1{-}A_2(\gamma_\iP) \right|\\ \leq \frac{k_{\iM}}{6} \left(k_\iP^3 B_1(\gamma_\iP){+}k_\iP\xi^3 B_2(\gamma_\iP)\right),
\end{multline}
with
\begin{align*}
  A_1 &\steq{def} \frac{1}{2} k_{\iM}k_{\iP}\xi^2\int_{\R_+^2}  \P\left(L_{\iP}{+}E_{\iP}{\in}(u{-}w,u)\right) \diff uL_\iM(\diff w),\\
  A_2(\gamma_\iP)&\steq{def} \frac{1}{2} k_{\iM}\left(1{-}e^{-\xi\sqrt{\gamma_\iP}}\right)^2 \left(\frac{k_{\iP}}{\gamma_\iP}\right)^2\int_{\R_+^2} \P\left(L_{\iP}{+}E_{\iP}{\in}(u{-}w,u)\right)^2 \diff uL_\iM(\diff w),\\
B_1(\gamma_\iP) &\steq{def}  \left(1{-}e^{-\xi\sqrt{\gamma_\iP} }\right)^3  \frac{1}{\gamma_\iP^3}\int_{\R_+^2}\P\left(L_{\iP}{+}B_{\iP}{\in}(u{-}w,u)\right)^3\diff uL_\iM(\diff w),  \\
B_2(\gamma_\iP) &\steq{def} \sqrt{\gamma_\iP}\int_{\R_+^2} \P\left(L_{\iP}{+}E_{\iP}{\in}(u{-}w,u)\right) \diff uL_\iM(\diff w).
\end{align*}
We now study the behavior of each of these four terms when $\gamma_\iP$ goes to $0$.
For the first term $A_1$, this is is straightforward since
\begin{equation}\label{eqq1}
A_1  =  \frac{1}{2}k_{\iM} k_{\iP} \xi^2 \int_{\R_+^2}\P\left(L_{\iP}{+}E_{\iP}{\le}u{\le} L_\iM{+}L_{\iP}{+}E_{\iP}\right)\diff u =  \frac{1}{2}k_\iM k_{\iP} \xi^2 \E(L_\iM). 
\end{equation}

The second term $A_2(\gamma_\iP)$ is handled in the same way as in the derivation of Relation~\eqref{varP} of Proposition~\ref{FTMP}. First note that
\begin{align*}
  \int_{\R_+^2}\P\left(L_{\iP}{+}E_{\iP}{\in}(u{-}w,u)\right)^2\diff u L_{\iM}(\diff w)&=  \int_{\R_+}\E\left(\left(w{-}|L_{\iP,1}{-}L_{\iP,1}{+}E_{\iP,2}{-}L_{\iP,2}|\right)^+\right)  \diff u L_{\iM}(\diff w)
  \\    &=  \E(L_\iM)\P\left(F_{L_\iM}{\ge}|L_{\iP,1}{-}L_{\iP,1}{+}E_{\iP,2}{-}L_{\iP,2}|\right)
  \\    &=  \E(L_\iM)\int_{0}^{+\infty} \P\left(F_{L_\iM}{\ge}|x{+}E_{\iP,1}{-}E_{\iP,2}|\right)\gamma_\iP e^{-\gamma_\iP x}\diff x,
\end{align*}
with the help of the symmetry $E_{\iP,1}{-}E_{\iP,2}{\steq{dist}}E_{\iP,2}{-}E_{\iP,1}$. Note that
\begin{multline*}
\int_0^{+\infty} \P\left(F_{L_\iM}{\ge}|x{+}E_{\iP,1}{-}E_{\iP,2}|\right)\diff x\leq
\int_0^{+\infty} \P\left(x{\le}F_{L_\iM}{+}|E_{\iP,1}{-}E_{\iP,2}|\right)\diff x\\\le\E\left(F_{L_\iM}\right){+}2\E(E_\iP)=\frac{\E(L_\iM^2)}{2\E(L_\iM)}{+}2\E(E_\iP){<}{+}\infty,
\end{multline*}
Lebesgue's Theorem gives therefore the relation
\begin{multline}\label{eqq2}
  \lim_{\gamma_\iP\to 0} A_2(\gamma_\iP)=
\frac{1}{2}\xi^2 k_{\iM}\E(L_\iM)k_\iP^2 \;  \lim_{\gamma_\iP\to 0} \int_0^{+\infty} \P\left(F_{L_\iM}{\ge}|x{+}E_{\iP,1}{-}E_{\iP,2}|\right)e^{-\gamma_\iP x}\diff x\\
=\frac{1}{2}\xi^2 k_{\iM}\E(L_\iM)k_\iP^2\int_0^{+\infty} \P\left(F_{L_\iM}{\ge}|x{+}E_{\iP,1}{-}E_{\iP,2}|\right)\diff x.
\end{multline}
We now analyze the upper bound of Relation~\eqref{dimeq1} expressed with the terms $B_1(\gamma_\iP)$ and $B_2(\gamma_\iP)$. The last one is easy to handle
\[
B_2(\gamma_\iP)=\sqrt{\gamma_\iP} \int_{\R_+^2} \P\left(L_{\iP}{+}E_{\iP}{\in}(u{-}w,u)\right)\diff u L_\iM(\diff w)= \E(L_\iM)\sqrt{\gamma_\iP},
 \]
 it converges to $0$ as $\gamma_\iP$ goes to $0$.

 For the term $B_1(\gamma_\iP)$, Jensen's Inequality and Fubini's Theorem give
\begin{align*}
 \int_{\R_+^2} \P\left(L_{\iP}{+}E_{\iP}{\in}(u{-}w,u)\right)^3\diff u L_\iM(\diff w)&{=}\int_{\R_+^2}\hspace{-2mm} \left(\E\left(e^{{-}\gamma_\iP(u{-}w{-}E_{\iP})^+}{-}e^{-\gamma_\iP(u{-}E_{\iP})^+}\right)\right)^3\diff u L_\iM(\diff w)\\
 &{\leq}\int_{\R_+^2}\hspace{-2mm} \E\left(\hspace{-1mm} \left(e^{-\gamma_\iP(u{-}w{-}E_{\iP})^+}{-}e^{-\gamma_\iP(u{-}E_{\iP})^+}\right)^3\right)\diff u L_\iM(\diff w)\\
 &{=}\int_{\R_+^2} \left(e^{-\gamma_\iP(u{-}w)^+}{-}e^{-\gamma_\iP u}\right)^3\diff u L_\iM(\diff w).
\end{align*}
With simple calculations, we get that, for $w{>}0$, the quantity 
\[
\int_{\R_+}\left(e^{-\gamma_\iP(u{-}w)^+}{-}e^{-\gamma_\iP u}\right)^3\,\diff u 
\]
is the sum of two terms, $C_1$ and $C_2$, with 
\begin{align*}
C_1 &\displaystyle \steq{def}\int_{w}^{\infty }\left( {e^{-\gamma_\iP \left( u-w \right) }}{-}{e^{-\gamma_\iP u}} \right) ^{3}\,\diff u
=\frac { \left(1{-}{{\rm e}^{-w \gamma_\iP}}\right)^{3}}{3\gamma_\iP},\\
C_2&\steq{def}\displaystyle \int_{0}^{w}\left( 1{-}{e^{-\gamma_\iP u}}\right)^{3}\,\diff u=\frac {2{e^{-3w \gamma_\iP}}{-}9{e^{-2wa}}{+}6w\gamma_\iP{+}
18{e^{-w\gamma_\iP}}{-}11}{6\gamma_\iP}.
\end{align*}
The elementary inequality~\eqref{elemineq} gives the relations
\[
C_1 \le \gamma_\iP^2 \frac{w^3}{3} \text{ and } |C_2|\le 4\gamma_\iP^2 w^3.
\]
By using the fact that $L_\iM$ has a finite third moment, we deduce the relation
\[
\lim_{\gamma_\iP{\to}0}  \left(1{-}e^{-\xi\sqrt{\gamma_\iP} }\right)^3\left(\frac{k_\iP}{\gamma_\iP}\right)^3
\int_{\R_+^2} \P\left(L_{\iP}{+}E_{\iP}{\in}(u{-}w,u)\right)^3\diff u L_\iM(\diff w)=0.
 \]
Relations\eqref{dimeq1}, \eqref{eqq1} and~\eqref{eqq2}  give therefore the convergence
\begin{multline*}
\lim_{\gamma_\iP{\to}0} \log\left[\E\left(\rule{0mm}{4mm}\exp(-\sqrt{\gamma_\iP}\xi (P{-}\E(P)))\right)\right]\\=\frac{\xi^2}{2}  k_\iM k_{\iP}\E(L_\iM)\left(1{+}k_\iP\int_0^{+\infty} \P\left(F_{L_\iM}{\ge}|x{+}E_{\iP,1}{-}E_{\iP,2}|\right)\diff x\right).
\end{multline*}
The theorem is proved. 
\end{proof}

\bigskip

\noindent
{\bf Short-Lived mRNAs and  Long-Lived Proteins}\\
For the scaling considered in this section  the death rate  $\gamma_\iP$ of proteins goes to $0$ and with the constraint that the death rate  $\gamma_\iM$ of mRNAs  is such that the average number of proteins is  kept fixed and equal to $a{>}0$,
\[
k_\iM\E(L_\iM)k_\iP\E(L_\iP)=\frac{k_\iM}{\gamma_\iM}\frac{k_\iP}{\gamma_\iP}=a.
\]
In particular the average lifetime $1/\gamma_\iM$ of an mRNA goes to $0$. 
\begin{prop}\label{Long1Prop}
  Under the conditions
  \begin{enumerate}
  \item the gene is always active;
  \item the distribution of  lifetimes of  mRNAs  and proteins are  exponential with respective  parameters $\gamma_\iM$ and $\gamma_\iP$;
  \item the transcription and translation parameters $k_\iM$ and $k_{\iP}{>}0$  are fixed as well as the distributions of the elongation times,
  \end{enumerate}
if $P$ is the random variable defined by Relation~\eqref{EqEqP}, then  for $a{>}0$, the convergence in distribution,
  \[
  \lim_{\substack{\gamma_\iM\to {+}\infty\\\gamma_\iM\gamma_\iP=k_1k_2/a}} P={\rm Pois}(a),
  \]
where ${\rm Pois}(a)$ is the Poisson distribution with parameter $a$.
\end{prop}
\noindent
Note that Proposition~\ref{PaulProp} gives readily that
\[
\lim_{\substack{\gamma_\iM\to {+}\infty\\\gamma_\iM\gamma_\iP=k_1k_2/a}} \frac{\Var(P)}{\E(P)}=1.
\]
From the point of view of the Fano factor, the fluctuations of the number of proteins are minimal. A (rough) picture of that is that short-lived mRNAs and almost permanent proteins minimize fluctuations of gene expression. The proposition is more precise since it states that the number of proteins at equilibrium is in fact asymptotically Poisson.
\begin{proof}
In view of Proposition~\ref{GenP} and, since $L_\iP$ is exponentially distributed, $L_\iP{\steq{dist}}F_{L_{\iP}}$, all we have to do is that the quantity  
\[
\Delta{\steq{def}}\int_{\R_+^2}\left[\rule{0mm}{4mm}1{-}\exp\left(\rule{0mm}{4mm}{-}\left( 1{-}z\right)k_{\iP}\E(L_{\iP})\P\left(L_{\iP}{+}E_{\iP}{\in}((u{-}w)^+,u)\right)\right)\right] k_{\iM}\diff u  L_{\iM}(\diff w)
\]
converges to $(1{-}z)a$, as $\gamma_1$ goes to infinity, with the constraint $\gamma_\iM\gamma_\iP{=}k_1k_2/a$.

With the elementary inequality~\eqref{elemineq}, we get 
\[
|\Delta{-}( 1{-}z)k_\iM k_{\iP}\Delta_1|{\le}\frac{( 1{-}z)^2k_\iM k_{\iP}^2}{2}\Delta_2,
\]
with 
\[
\Delta_1{\steq{def}}  \int_{\R_+^2}
\frac{1}{\gamma_\iP}\P\left(L_{\iP}{+}E_{\iP}{\in}((u{-}w)^+,u)\right) \diff u \gamma_{\iM}e^{-\gamma_{\iM}w}\diff w
  \]
  and
\[
\Delta_2{\steq{def}}  \int_{\R_+^2}
\left(\frac{\P\left(L_{\iP}{+}E_{\iP}{\in}((u{-}w)^+,u)\right)}{\gamma_\iP}\right)^2 \diff u \gamma_{\iM}e^{-\gamma_{\iM}w}\diff w.
\]
We now take care of these two quantities. By invariance by translation, we have
\begin{align*}
  \Delta_1&{=} \frac{1}{\gamma_\iP}\E\left(\int\left(e^{-\gamma_\iP(u{-}E_\iP{-}w)^+}{-}e^{-\gamma_\iP(u{-}E_\iP)^+}\right)\diff u   \gamma_{\iM}e^{-\gamma_{\iM}w}\diff w\right)\\
&{=} \frac{1}{\gamma_\iP}\E\left(\int\left(e^{-\gamma_\iP(u{-}w)^+}-e^{-\gamma_\iP u}\right)\diff u   \gamma_{\iM}e^{-\gamma_{\iM}w}\diff w\right)=\frac{1}{\gamma_\iM\gamma_\iP}.
\end{align*}
The  Cauchy-Shwartz Inequality gives for the second term
\begin{align*}
  \Delta_2&{\steq{def}}  \int_{\R_+^2}\left[\frac{1}{\gamma_\iP}\E\left(\left(e^{-\gamma_\iP(u{-}E_\iP{-}w)^+}{-}e^{-\gamma_\iP(u{-}E_\iP)^+}\right)\right)\right]^2  \diff u \gamma_{\iM}e^{-\gamma_{\iM}w}\diff w\\
  &\le
  \E\left( \int_{\R_+^2}\frac{1}{\gamma_\iP^2}\left(e^{-\gamma_\iP(u{-}w)^+}{-}e^{-\gamma_\iP u}\right)^2 \diff u \gamma_{\iM}e^{-\gamma_{\iM}w}\diff w\right)=\frac{1}{\gamma_\iM\gamma_\iP(\gamma_\iM{+}\gamma_\iP)}
\end{align*}
hence it converges to $0$ as $\gamma_1$ goes to infinity with $\gamma_1\gamma_2$ being constant. The proposition is proved by gathering these results. 
\end{proof}

\medskip
\noindent
The above result can be explained quite simply as follows. We assume that the elongation times are null for simplicity.  With the ``coupling from the past'' approach of Section~\ref{ConvSec}, if $(u_n)$ denotes the non-decreasing sequence of instants of binding instants before time $0$, since the distribution of the sequence of lifetimes  $(L_{\iM,n})$ of mRNAs is converging to $0$, with high probability there should at most one protein created by the mRNA with index $n$, so that
\[
P \sim R \steq{def}\sum_{n<0} \ind{{\cal N}_{\iP}([u_n{+}E_{\iM,n},u_n{+}E_{\iM,n}{+}L_{\iM,n}]){\not=}0,u_n{+}E_{\iM,n}{+}L_{\iP,n}{>}0 },
\]
and a simple calculation gives
\[
\E\left(z^R\right)=\exp\left( -(1-z) \frac{k_{\iP} k_{\iM}}{\gamma_\iP(\gamma_{\iM}{+}k_\iP)}\right)\sim \exp\left( -(1-z) \frac{k_{\iP} k_{\iM}}{\gamma_\iP\gamma_{\iM}}\right),
\]
as $\gamma_\iP$ goes to $0$ with $\gamma_\iM\gamma_\iP$ constant. The distribution of the variable $R$ is asymptotically Poisson with parameter ${k_{\iP} k_{\iM}}/{(\gamma_\iP\gamma_{\iM})}$.  It is not difficult to see that, with convenient estimates, a rigorous proof could be obtained by using this observation.

\bigskip

\noindent
{\bf Short-Lived mRNAs and High Expression Rate of Proteins}\\
In this section  the death rate  $\gamma_\iM$ of proteins goes to infinity  but it is assumed in addition that  the quantity $k_\iP/\gamma_\iM$, the average number of proteins translated by an mRNA,  is constant. In particular $k_\iP$, the expression rate of proteins, is large and  the average number of proteins is fixed. 

\begin{prop}\label{Long2Prop}
  Under the conditions
  \begin{enumerate}
  \item the gene is always active;
  \item the distribution of  lifetimes of  mRNAs  and proteins are  exponential with respective  parameters $\gamma_\iM$ and $\gamma_\iP$;
  \item  the parameters $k_\iM$ and $\gamma_{\iP}{>}0$  are fixed as well as the distributions of the elongation times,
  \end{enumerate}
if $P$ is the random variable defined by Relation~\eqref{EqEqP} , then  for $a{>}0$, the relation 
  \[
  \lim_{\substack{\gamma_\iM\to {+}\infty\\k_\iP/\gamma_\iM=a}} P=N(a),
  \]
holds for the convergence in distribution, where $N(a)$ is a discrete random variable whose generating function is given by, for $z{\in}(0,1)$, 
  \[
\E\left(z^{N(a)}\right)=\exp\left({-}k_\iM\int_{0}^{+\infty}\frac{\left( 1{-}z\right)a\E\left(e^{{-}\gamma_\iP (u{-}E_\iP)^+}\ind{u{>}E_\iP}\right)}{1{+}\left( 1{-}z\right)a \E\left(e^{{-}\gamma_\iP (u{-}E_\iP)^+}\ind{u{>}E_\iP}\right)}\diff u\right).
  \]
\end{prop}

\noindent
The distribution of the  random variable $N(a)$ belong to the class of generalized {\em Neyman Type A distributions}. See Section~9 of Chapter~9 of Johnson et al.~\cite{Johnson} for example. In the Markovian case of Section~\ref{MarkSec}, this is a negative binomial distribution. See Proposition~\ref{NBD}.
\begin{proof}
For $z{\in}(0,1)$, the quantity ${-}\log(\E(z^{P}))/k_\iM$ is, since $F_{L_\iP}{\steq{dist}}L_\iP$, 
\begin{multline}\label{H-Long}
  H(z)\steq{def}\int_{\R_+^2}\left[1{-}e^{{-}\left( 1{-}z\right)k_{\iP}\E(L_{\iP})\P\left(L_{\iP}{+}E_{\iP}{\in}(u{-}w,u)\right)}\right] \diff u \gamma_{\iM}e^{{-}\gamma_{\iM}w}\diff w\\
  =\int_{\R_+^2}\left[1{-}e^{{-}\left( 1{-}z\right)k_{\iP}\E(L_{\iP})\E\left(e^{{-}\gamma_\iP(u{-}w/\gamma_{\iM}{-}E_{\iP})^+}{-}e^{{-}\gamma_\iP(u{-}E_{\iP})^+}\right)}\right] \diff u e^{{-}w}\diff w.
\end{multline}
With similar arguments as in the proof of Proposition~\ref{Long1Prop} and Lebesgue's Theorem, we obtain the relation
\[
\lim_{\substack{\gamma_\iM\to {+}\infty\\k_\iP/\gamma_\iM{=}a}}  H(z)
 =\int_{\R_+^2}\left[1{-}e^{{-}\left( 1{-}z\right)a\gamma_\iP\E(L_{\iP})w\E\left(e^{{-}\gamma_\iP(u{-}E_{\iP})^+}\ind{u{>}E_\iP} \right)}\right] \diff u e^{{-}w}\diff w.
 \]
The proposition is proved. 
\end{proof}

\bigskip

\subsection{The Impact of Elongation on  Variability}

In this section we discuss the impact of the elongation of proteins on the variability of the protein production process. This aspect has been rarely considered in the mathematical models of the literature. Note that the first moments of $M$ and $P$ do not depend on the distributions of the elongation times. The distribution of the random variable $M$ and  the second moment of $P$ do not depend on the distribution of the elongation phase of mRNAs when the gene is always active. 

To simplify the presentation, we assume in this section that the lifetimes of mRNAs and proteins are exponentially distributed with respective parameters $\gamma_\iM$ and $\gamma_\iP$ and that the gene is always active.

If  the protein is composed of $N$ amino-acids, it is natural to assume that the average of the elongation time is proportional to $N$, i.e.\ given by $N/\alpha$ for some $\alpha{>}0$. An index $N$ is added to the variables $E_2$ and  $P$ defined by Relation~\eqref{EqEqP}. Relation~\eqref{varP} of Proposition~\ref{FTMP} gives that the variance of the number of proteins at equilibrium is given by, 
\begin{equation}\label{eqvarP}
\Var(P^N)-\frac{k_\iM k_\iP}{\gamma_\iM\gamma_\iP} = \frac{k_\iM k_\iP^2}{\gamma_\iM\gamma_\iP^2} \E\left(e^{{-}\gamma_\iM \left|S_\iP^N{+}Y_\iP\right|}\right),
\end{equation}
by the equality in distribution $F_{L_{a}}{\steq{dist}}L_a$ for $a{\in}\{\iM,\iP\}$ due to the exponential distribution assumption, with $S_\iP^N{=}E_{\iP,1}^N{-}E_{\iP,2}^N$ and  $Y_\iP{=}L_{\iP,1}{-}L_{\iP,2}$.  

Two assumptions on the distribution of the elongation time of a protein are now considered.
\begin{enumerate}
\item {\sc Markov model.} The variable $E_\iP^N$ is an exponential random variable  given by $E_{\iP,M}^N$ with
  \[
  E_{\iP,M}^N\steq{def} N V_\alpha,
  \]
  where $V_\alpha$ is an exponential random variable with parameter $\alpha$. 

A similar assumption can be done for the elongation of mRNAs, i.e.\ that its distribution is also exponential. If $M^0(t)$, (resp., $P^0(t)$), denotes the number of mRNAs, (resp., proteins),  being built at time $t{\ge}0$, then the process $$(I(t),M^0(t),M(t),P^0(t),P(t))$$ has clearly the Markov property. This is a natural Markovian extension of the classical Markovian model including elongation.
\item[]  
\item {\sc Additive Model.} The variable $E_\iP^N$ is given by $E_{\iP,A}^N$ with
  \[
  E_{\iP,A}^N\steq{def}\sum_{k=1}^N V_{\alpha,i},
  \]
where $(V_{\alpha,i})$ is an i.i.d.\ sequence of exponential random variables with parameter $\alpha$.  This is also a natural assumption since it considers that it takes a random amount of time with an exponential distribution to find each  amino-acid of the protein within the cytoplasm. 
\end{enumerate}

\medskip

The variable $P$ defined by  Relation~\eqref{EqEqP} is denoted as $P^N_M$ for the Markov model and $P^N_A$ for the additive model. Similarly for the variable $S_\iP^N$ of Relation~\eqref{eqvarP}, we define the corresponding variables $S_{\iP,M}^N$ and $S_{\iP,A}^N$.

\bigskip

\noindent
{\bf A non-intuitive phenomenon}.
A simple calculation gives
\[
\Var(E_{\iP,M}^N)=\frac{N^2}{\alpha^2} \text{ and }\Var(E_{\iP,A}^N)=\frac{N}{\alpha^2}.
\]
As it can be seen the variability of the elongation time is larger for the Markov model. Intuitively, this could suggest that the same property holds for the number of proteins, that the variance of $P^N_M$  is larger than the variance of $P^N_A$. But a glance at Relation~\eqref{eqvarP} suggests in fact the contrary. Indeed, the variable $S_{\iP}^N{+}Y_\iP$ will be more variable for the Markov model and, therefore, due to this relation,  the variance should be smaller. The next proposition gives a more formal formulation of this observation.  See also Leoncini~\cite{Leoncini} for related  numerical results of this kind. 
\begin{prop}[Variance for a Large Number of Amino-Acids]
  Under the conditions
  \begin{enumerate}
  \item the gene is always active;
  \item the number of amino-acids of the protein is $N$, the mean elongation time of a protein is $N/\alpha$ for $\alpha{>}0$;
  \item the distribution of  lifetimes of  mRNAs  and proteins are  exponential with respective  parameters $\gamma_\iM$ and $\gamma_\iP$;
  \end{enumerate}
the variable $P$ defined by  Relation~\eqref{EqEqP} is denoted as $P^N_M$ for the Markov model and $P^N_A$ for the additive model, then
\[
\begin{cases} 
  \displaystyle \E\left(P^N_M\right){=}\E\left(P^N_A\right){=}\frac{k_\iM k_\iP}{\gamma_\iM \gamma_\iP}{\steq{def}}a,\\
  \displaystyle \lim_{N\to+\infty}  \Var\left(P^N_M\right){=}  \lim_{N\to+\infty}  \Var\left(P^N_A\right){=} a.
\end{cases}
  \]
 Furthermore,
 if $\gamma_\iM{<}\gamma_\iP$, there exists $\alpha_0{>}0$ such that if $\alpha{>}\alpha_0$ then  there exist constants  $C{>}c{>}0$ such that the relations
  \[
\begin{cases} 
  \displaystyle c \le N\left( \Var(P^N_M){-} a\right) \le  C,\\
  \displaystyle  c \le \sqrt{N}\left(\Var(P^N_A){-}a\right)  \le  C.
\end{cases}
  \]
  hold  for any $N{\ge}1$.
\end{prop}
This proposition shows that, if the second moment of $P$ depend on the distribution of the elongation times, its dependence is somewhat limited for both models, Markov and  additive. It converges to some constant as the average (and the variance) of the elongation time goes to infinity. Recall that the average number of proteins does not depend at all on the distribution of the elongation time.

The inequalities of the proposition show that the convergence rate of the variance to $a$,  as $N$ gets large, is of the order of $1/\sqrt{N}$ for the additive model, and $1/N$ for the Markov model. 

\begin{proof}
With the notations of Relation~\eqref{eqvarP},  for $B{\in}\{A,M\}$, the elementary relations
\[
\E\left(e^{{-}\gamma_\iM \left|Y_{\iP}\right|}\right) \E\left(e^{{-}\gamma_\iM \left|S_{\iP,B}^N\right|}\right)\le\E\left(e^{{-}\gamma_\iM \left|S_{\iP,B}^N{+}Y_\iP\right|}\right)\leq \E\left(e^{{-}\gamma_\iM\left|S_{\iP,B}^N\right|}\right)\E\left(e^{\gamma_\iM \left|Y_{\iP}\right|}\right)
\]
hold. Since $\gamma_\iM{<}\gamma_\iP$, one has
\[
\E\left(e^{\gamma_\iM|Y_\iP|}\right){<}{+}\infty,
\]
hence, one has only to take care of the limiting behavior of $\E\left(\exp\left({-}\gamma_\iM |S_\iP^N|\right)\right)$ as $N$ goes to infinity. 

For the Markov model, a simple calculation gives that $|S_{\iP,M}^N|$ is exponentially distributed with parameter $\alpha/N$,
\[
\E\left(e^{{-}\gamma_\iM \left|S_{\iP,M}^N\right|}\right)= \frac{\alpha}{\alpha{+}\gamma_\iM N},
\text{ so that }
\lim_{N{\to}{+}\infty} N\E\left(e^{{-}\gamma_\iM \left|S_{\iP,M}^N\right|}\right)=\frac{\alpha}{\gamma_\iM},
\]

For the additive model, $S_{\iP,A}^N$ is a sum of $N$ i.i.d. centered random variables
\[
S_{\iP,A}^N=\sum_{k=1}^N (V^1_{\alpha,i}{-}V^2_{\alpha,i}),
\]
where $(V^1_{\alpha,i})$ and $(V^2_{\alpha,i})$ are independent i.i.d.\ sequences of exponential random variables with parameter $\alpha$. Let
\[
\sigma\steq{def} \sqrt{\Var(V^1_{\alpha,1}{-}V^2_{\alpha,1})}=\frac{\sqrt{2}}{\alpha} \text{ and }\rho\steq{def} \E\left(\left|V^1_{\alpha,1}{-}V^2_{\alpha,1}\right|^3\right)=\frac{6}{\alpha^3}, 
\]
Berry-Esseen's theorem, See Feller~\cite{Feller} for example, gives that, for all $N{\in}\N$,
\begin{equation}\label{Berry}
  \sup_{x\in \R_+} \left|\P\left( \frac{|S_{\iP,A}^N|}{\sigma\sqrt{N}}{\leq} x\right){-}\P(|{\cal N}(1)|{\leq} x)\right|\le \frac{6\rho}{\sigma^3 \sqrt{N}}
  =\frac{18}{\sqrt{2N}},
\end{equation}
where ${\cal N}(1)$ is a centered Gaussian random variable with standard deviation $1$. 
From Fubini's Theorem,  we have
  \[
  \E\left(e^{-\gamma_\iM \left|S_{\iP,A}^N\right|}\right)= \gamma_\iM  \int_0^{+\infty} \P\left(\left|S_{\iP,A}^N\right|{\le} x\right)e^{{-}\gamma_\iM x}\,\diff x, 
  \]
and,  with  Relation~\eqref{Berry}, this  gives
\begin{equation}\label{Berry2}
  \left| \E\left(e^{-\left|S_\iP^N\right|} \right){-}  \gamma_\iM  \int_0^{+\infty} \P\left(|{\cal N}(1)|{\le} \frac{x}{\sigma \sqrt{N}}\right)e^{{-}\gamma_\iM x}\,\diff x\right|
\leq  \frac{18}{\sqrt{2N}}.
\end{equation}
Let $a_N{=}{\gamma_\iM}{\sigma\sqrt{N}}$, then 
\begin{multline*}
  \gamma_\iM \int_0^{+\infty} \P\left(|{\cal N}(1)|{\le} \frac{x}{\sigma \sqrt{N}}\right)e^{{-}\gamma_\iM x}\,\diff x=
  \E\left(\exp\left({-}{a_N}|{\cal N}(1)|\right)\right)\\
 = \frac{1}{a_N} \sqrt{\frac{2}{\pi}}\int_0^{+\infty} e^{-x^2/(2a_N^2)}e^{-x}\diff x.
\end{multline*}
With Lebesgue's Theorem, we obtain 
\[
\lim_{N\to+\infty}   \sqrt{N}\gamma_\iM \int_0^{+\infty} \P\left(|{\cal N}(1)|{\le} \frac{x}{\sigma \sqrt{N}}\right)e^{{-}\gamma_\iM x}\,\diff x=\frac{1}{\gamma_\iM\sigma} \sqrt{\frac{2}{\pi}}=\frac{\alpha}{\gamma_\iM\sqrt{\pi}}.
\]
If $\alpha$ is such that
\[
\alpha{>}18\gamma_\iM\sqrt{\frac{\pi}{2}},
\]
Inequality~\eqref{Berry2} gives the relations
\[
0{<}\liminf_{N\to+\infty} \sqrt{N}\E\left(e^{-\left|S_\iP^N\right|} \right)\le \limsup_{N\to+\infty} \sqrt{N}\E\left(e^{-\left|S_\iP^N\right|} \right)<{+}\infty.
\]
This concludes the proof of the proposition. 
\end{proof}

\newpage
\providecommand{\bysame}{\leavevmode\hbox to3em{\hrulefill}\thinspace}
\providecommand{\MR}{\relax\ifhmode\unskip\space\fi MR }
\providecommand{\MRhref}[2]{%
  \href{http://www.ams.org/mathscinet-getitem?mr=#1}{#2}
}
\providecommand{\href}[2]{#2}

\newpage
\appendix
\section{Some Numbers for {\em Escherichia coli}  Bacterial Cells}\label{NumbSec}

We give some numbers and some statistics for a single cell, to give an idea of the orders of magnitude of the different elements of this production process. It is also used to justify some of the assumptions in  the mathematical models introduced. The sources of these numbers can be found at the address {\tt www.bioNumbers.org}. The references Bremer and Dennis~\cite{Bremer} and Chen et al.~\cite{ChenShi} have also been used. 

\begin{figure}[ht]
        \centering
        \begin{subfigure}[b]{0.45\textwidth}
          \includegraphics[width=1\textwidth]{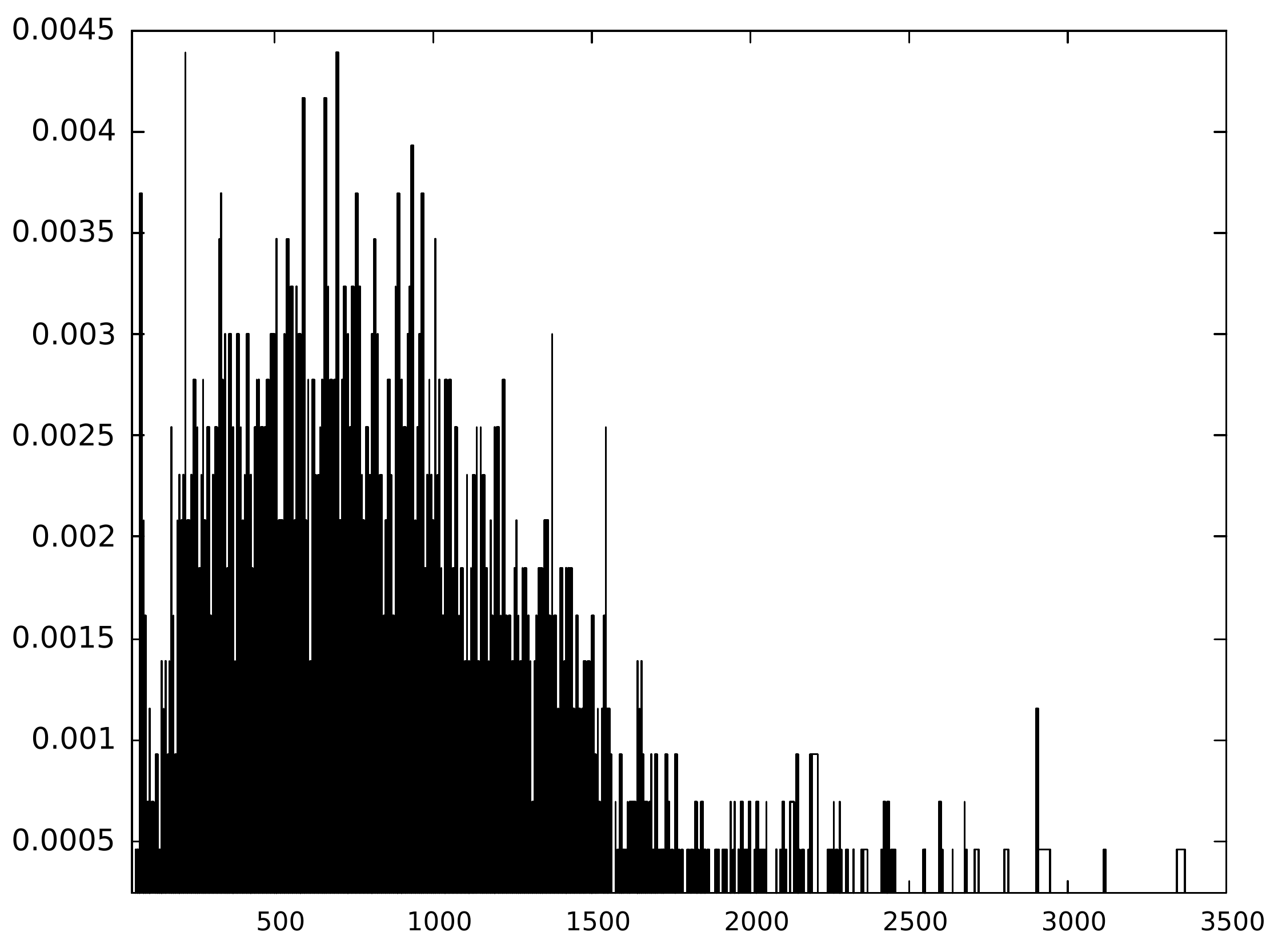}
          \put(-170,50){\rotatebox{90}{{\small Probability}}}
          \put(-80,-10){{\small Size}}
\caption{mRNAs: Dist. of Nb of nucleotides\\\phantom{aaa} 4325 types of mRNAs,\\\phantom{aaa} mean: 939.74, std. dev.: 643.97}
\end{subfigure}\quad
\begin{subfigure}[b]{0.45\textwidth}
\includegraphics[width=1\textwidth]{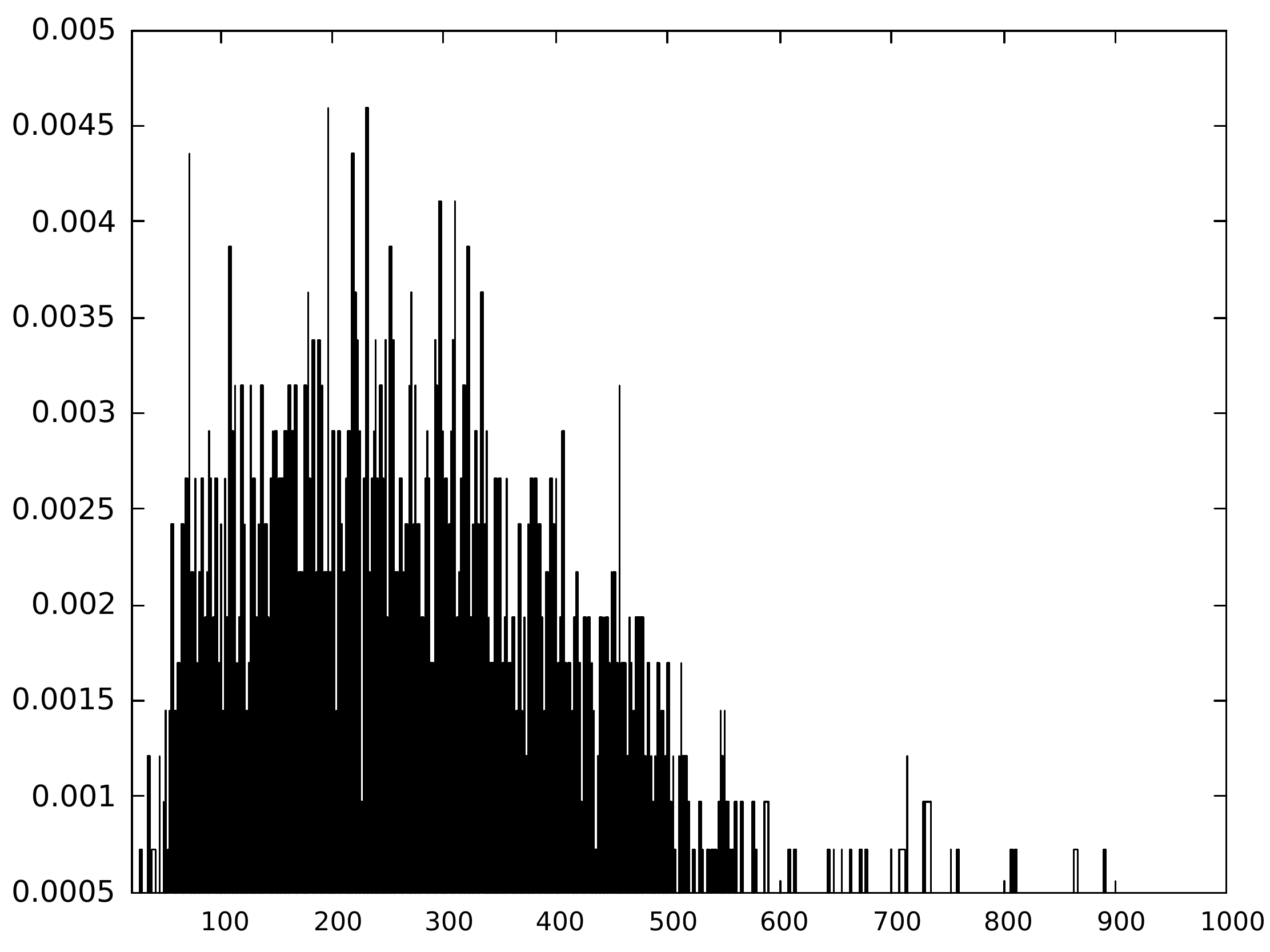}
          \put(-170,50){\rotatebox{90}{{\small Probability}}}
          \put(-80,-10){{\small Size}}
\caption{Proteins: Dist. of Nb of Amino-Acids\\\phantom{aaa} 4136 types of proteins,\\\phantom{aaa} mean: 318.21, std. dev.: 208.29}
\end{subfigure}
\caption{Distributions of Sizes for  {\em Escherichia coli}.\\  Sources of Data:  Karp et al.~\cite{BioCyc} and Valgepea et al.~\cite{Valgepea}.}\label{FigSize}
\end{figure}

 \begin{enumerate}
\item[]
   \item {\sc Global Parameters.}
 \begin{itemize}
 \item Diameter of bacteria: ${\sim} 2{\cdot}10^{-6}$m.
 \item Water content of a cell: $50{-}70$\%.
 \item Dry content of a cell: ${\sim}55$\% proteins, $20$\% mRNAs, $10$\% lipids.
 \end{itemize}
\item[]
\item {\sc Polymerases and Ribosomes.} A ribosome is a complex macro-molecule composed of ribosomal proteins, of the order of $50$ units, and ribosomal RNAs corresponding to a total of ${\sim}3000$ nucleotides. Polymerases and ribosomes can be seen, for the production process,  as resources of the cell. They are used by all types of genes of the cell,   in some way,  genes (resp., the mRNAs) of all types are competing to have access to polymerases (resp., ribosomes) to produce mRNAs (resp., proteins). 
 \begin{itemize}
\item[]
 \item Number of polymerases: $1500{-}11400$.
 \item Number of ribosomes: $6800{-}72000$.
 \item Rate of transcription by an RNA polymerase ${\sim}85$ nucleotides/sec.
 \item Rate of translation by a ribosome ${\sim}20$ amino-acids/sec.
\item[]
 \end{itemize}
 Note that the range of variations of these numbers depends on the growth rate of the cells, i.e.\ of the environment. 
\item[]
 \item {\sc  mRNAs and Proteins.}
\begin{itemize}
\item[]
\item Number of mRNAs in a cell: ${\sim}7800$.
\item Average lifetime of an mRNA: ${\sim}2$mn.
\item Size of protein: ${\sim}2{\cdot}10^{-9}$m.
\item Number of proteins in a cell: ${\sim}2.36{\cdot}10^6$.
\item Diffusion coefficient of a protein:  ${\sim}10^{-2}$sec to cross cell.
\end{itemize}
 \end{enumerate}

\section{Poisson Point  Processes}\
\label{PoisSec}
The main results concerning Poisson processes used in this chapter are briefly recalled, in the following  $H$ is a separable locally compact metric space. See Kingman~\cite{Kingman}, Neveu~\cite{Neveu}, or Robert~\cite{Robert}. 
\begin{defi} 
 If $\lambda{>}0$, $\mu$ is a probability distribution on $H$, a marked Poisson process on $\R{\times}H$ with intensity $\lambda \diff x{\otimes}\mu$  is a sequence ${\cal N}_\lambda{=}((t_n,X_n), n{\in}\Z)$  of elements of $\R{\times}H$ where
\begin{itemize}
\item $(t_n, n{\in}\Z)$ is a (classical) Poisson process on $\R$ with rate $\lambda$, with the convention that $t_n{\leq} t_{n+1}$ for any  $n{\in}\Z$ and $t_0{\leq} 0{<}t_1$.
\item $(X_n, n{\in}\Z)$ is an i.i.d. sequence of random variables with values in $H$ and whose distribution is $X(\diff x)$. 
\end{itemize}
\end{defi}
The sequence ${\cal N}_\lambda$ can also be seen as a marked point process on $\R{\times}H$, if $F$ is a non{-}negative Borelian function on $\R{\times}H$, then 
\[
\croc{{\cal N}_\lambda,F}=\int_{\R\times H} F(u,x){\cal N}_{\lambda}(\diff u,\diff x)=\sum_{n\in\Z} F(t_n,X_n). 
\]
 The following important proposition characterizes marked Poisson point processes. 
\begin{prop}[Laplace Transform of a Poisson Process]\label{Pois-lapois}
The marked point process ${\cal N}_\lambda=(t_n,X_n)$ is a Poisson point process with intensity $\lambda \diff u{\otimes}X(\diff x)$ if and only if the relation
\[
\E\left(\rule{0mm}{4mm}\exp\left(-\croc{{\cal N}_\lambda,F}\right)\right)=\exp\left(-\lambda \int_{\R{\times} H} \left(1-e^{-F(s,x)}\right)\,\diff s\,X(\diff x)\right)
\]
holds for any non{-}negative continuous function $F$ on $\R{\times}H$.
\end{prop}

\begin{corollary}\label{PoisMom}
  If $F$ and  $G$ are non{-}negative functions in $L_2{\cap}L_1(\diff u{\otimes}X)$,
  then
\begin{equation}\label{PoisMean}
\E\left(\int_{\R{\times}H} F(s,x)\,{\cal N}_\lambda(\diff s,\diff x)\right)=\int_{\R{\times}H}F(s,x)\,\lambda\diff sX(\diff x)
\end{equation}
and
\begin{multline}\label{PoisVar}
\Cov\left(\int_{\R{\times}H} F(s,x)\,{\cal N}_\lambda(\diff s, \diff x),\int_{\R{\times}H} G(s,x)\,{\cal N}_\lambda(\diff s, \diff x)\right)\\=\int_{\R{\times}H} F(s,x)G(s,x)\,\lambda\diff sX(\diff x)
\end{multline}
\end{corollary}

\bigskip
\noindent
{\bf Negative Binomial Distribution.}\\
For $a$, $b{>}0$, let $G$ be a random variable whose distribution has the density
\[
\frac{b^a}{\Gamma(a)}x^{a-1}e^{-b x}
\]
on $\R_+$, where $\Gamma$ is the classical Gamma function
\[
\Gamma(a)=\int_0^{+\infty} u^{a-1}e^{-u}\,\diff u.
\]
See Whittaker and Watson~\cite{Whittaker} for example. Let ${\cal N}_1$ be a Poisson process on $\R_+$ with rate $1$, then, for $z{\in}[0,1]$,
\begin{equation}\label{NBgen}
\E\left(z^{{\cal N}([0,G])}\right)=\E\left(e^{-G (1-z)}\right)=\int_0^{+\infty} \frac{b^a}{\Gamma(a)}x^{a-1}e^{-(b+1-z)x}\,\diff x
=\left(\frac{b}{1{+}b{-}z}\right)^a,
\end{equation}
additionally, for $n{\in}\N$,
\begin{equation}\label{NBdist}
\P({\cal N}([0,G]){=}n)=\frac{b^a}{n!\Gamma(a)}\int_0^{+\infty} x^{a+n-1}e^{-(b+1)x}\,\diff x
=\left(\frac{b}{b{+}1}\right)^a \frac{\Gamma(a{+}n)}{n!\Gamma(a)}\frac{1}{(1{+}b)^n}.
\end{equation}
\tocless\section{Proof of Proposition~\ref{FTMP}}{ProofFTMP} \ \\
As in section~\ref{mRNAsec}, ${\cal F}^I$ denotes the $\sigma$-field generated by the stochastic process $(I(t))$,  by using again Proposition~\ref{CondI} , we have
\begin{align*}
  \E\left(P{\mid} {\cal F}^I\right) &=  \int_{{\cal S}} G_{\iP}(u,v,w,m)   I(u) \nu_{\cal P}(\diff u\,\diff v,\diff w, \diff m),\\
  &=   {k_{\iM}}\int_{{\cal S}} I(u) G_{\iP}(u,v,w,m)    \diff u\,{L_{\iM}}(\diff v) Q_{\iP}(\diff m)\\
   &=   {k_{\iM}} \int_{\R} I(u) \E\left( G_{\iP}(u,{E_{\iM}},{L_{\iM}},{\cal N}_{2})\right) \diff u,
\end{align*}
by definition of $Q_2$, see Relation~\eqref{nup},  and $G_\iP$ is defined by Relation~\eqref{EqP2}.  The integration of this identity gives
\begin{align*}
  \E(P) &=\delta_+{k_{\iM}}\E\left( \int_{\R}\int_{\R{\times}\R_+^2}\ind{\substack{x{+}{E_{\iM}}{\le}u{<}x{+}{E_{\iM}}{+}{L_{\iM}}\\y{\le} x{<}y{+}z}}{\cal N}_2(\diff x,\diff y,\diff z)\diff u\right)\\
  &=\delta_+{k_{\iM}}\E\left( \int_{\R^2}\ind{\substack{x{+}{E_{\iM}}{\le}u{<}x{+}{E_{\iM}}{+}{L_{\iM}}\\{E_{\iP}}{\le} x{<}{E_{\iP}}{+}{L_{\iP}}}}{k_{\iP}}\diff x\diff u\right)\\
  &=\delta_+{k_{\iM}}\E({L_{\iM}}){k_{\iP}}\E({L_{\iP}}),
\end{align*}
where $\delta_+{=}k_+/\Lambda$. The first identity of the proposition is proved.

From Representation~\eqref{EqP2} of $P$, Proposition~\ref{CondI} and Relation~\eqref{PoisVar} of Corollary~\ref{PoisMom}, we obtain 
\begin{align}\label{vqaux1}
  \E\left(P^2{\mid} {\cal F}^I\right)&=\left(\E\left(P{\mid} {\cal F}^I\right)\right)^2 +\int_{{\cal S}} I(u) G_{\iP}(u,v,w,m)^2\nu_{\cal P}(\diff u,\diff v,\diff w,\diff m)\\
  &=\left(\E\left(P{\mid}{\cal F}^I\right)\right)^2 +{k_{\iM}} \int_{\R} I(u) \E\left[\rule{0mm}{4mm}G_\iP\left(u,{E_{\iM}},{L_{\iM}},{\cal N}_2\right)^2\right] \diff u.\notag
\end{align}
We now take care successively of the two terms of the right hand side of the last equality. 
First we study  $\E\left[G_\iP\left(u,{E_{\iM}},{L_{\iM}},{\cal N}_2\right)^2\right]$.

For $u$, $v$, $w{\in}\R_+$, the distribution of the  random variable 
\[
G_\iP\left(u,v,w,{\cal N}_2\right)=
\int_{\R{\times}\R_+^2}\ind{\substack{x{+}v{\le}u{<}x{+}v{+}w\\y{\le} x{<}y{+}z}}{\cal N}_2(\diff x,\diff y,\diff z)
\]
is Poisson with parameter
\begin{align*}
\kappa(u,v,w)&\steq{def}{k_{\iP}}\int_{\R{\times}\R_+^2}\ind{\substack{x{+}v{\le}u{<}x{+}v{+}w\\y{\le} x{<}y{+}z}}\diff x {E_{\iP}}(\diff y){L_{\iP}}(\diff z)\\
&={k_{\iP}}\E\left(\int_{\R} \ind{\substack{x{+}v{\le}u{<}x{+}v{+}w\\{E_{\iP}}{\le} x{<}{E_{\iP}}{+}{L_{\iP}}}}\diff x\right).
\end{align*}
By using again Relation~\eqref{PoisVar} of Corollary~\ref{PoisMom}, we obtain that
\begin{align*}
  {k_{\iM}}\E\left(\int_{\R_+}\right.&\left. I(u) \E\left[\rule{0mm}{4mm}G_\iP\left(u,{E_{\iM}},{L_{\iM}},{\cal N}_2\right)^2\right] \diff u\right)\\
  &= \delta_+{k_{\iM}} \int_{\R} \left[\E\left(\kappa(u,{E_{\iM}},{L_{\iM}})\right){+}\E\left(\kappa(u,{E_{\iM}},{L_{\iM}})^2\right)\right]\, \diff u\\
  & =\E(P)+\delta_+{k_{\iM}} \int_{\R{\times}\R_+^2} \kappa(u,v,w)^2 \diff u {E_{\iM}}(\diff v){L_{\iM}}(\diff w).
\end{align*}
For $i{=}1$, $2$, as before $(E_{\iP,i})$ and $(L_{\iP,i}$ are independent random variables with the same distribution as ${E_{\iP}}$ and ${L_{\iP}}$ respectively.
By rewriting $\kappa(u,v,w)^2$ as a double integral, we obtain
\begin{align*}
  &\int_{\R{\times}\R_+^2} \kappa(u,v,w)^2 \diff u {E_{\iM}}(\diff v){L_{\iM}}(\diff w)\\
  &={k_{\iP}}^2  \int_{\R{\times}\R_+^2}\E\left( \int_{\R^2} \ind{\substack{x{+}v{\le}u{<}x{+}v{+}w\\E_{\iP,1}{\le} x{<}E_{\iP,1}{+}L_{\iP,1}}, \substack{x'{+}v{\le}u{<}x'{+}v{+}w\\E_{\iP,2}{\le} x{<}E_{\iP,2}{+}L_{\iP,2}}}\diff x\diff x'\right) \diff u {E_{\iM}}(\diff v){L_{\iM}}(\diff w)\\
  &={k_{\iP}}^2  \int_{\R_+^2}\E\left( \int_{\R^2} \left(w{+}x{\wedge}x'{-}x{\vee}x'\right)^+\ind{\substack{E_{\iP,1}{\le} x{<}E_{\iP,1}{+}L_{\iP,1}\\E_{\iP,2}{\le} x'{<}E_{\iP,2}{+}L_{\iP,2}}}\diff x\diff x'\right) {E_{\iM}}(\diff v){L_{\iM}}(\diff w),
\end{align*}
after integration with respect to $u$,   the translation invariance of Lebesgue's measure on $\R$ and Relation~\eqref{eqFa} give the relation
\begin{align*}
  \int_{\R{\times}\R_+^2}& \kappa(u,v,w)^2 \diff u {E_{\iM}}(\diff v){L_{\iM}}(\diff w)\notag\\
  &={k_{\iP}}^2\E(L_\iM)  \E\left( \int_{\R_+^2} \P(F_{{L_{\iM}}}{\ge}|E_{\iP,1}{+}x{-}(E_{\iP,2}{+}x')|) \P(L_{\iP}{\ge}x)\P(L_{\iP}{\ge}x')\diff x\diff x'\right)\notag\\
   &=\E(L_\iM)({k_{\iP}}\E({L_{\iP}}))^2 \P(F_{{L_{\iM}}}{\ge}|E_{\iP,1}{+}F_{{L_{\iP}},1}{-}(E_{\iP,2}{+}F_{{L_{\iP}},2})|),
\end{align*}
hence
\begin{multline}\label{Term1}
  \frac{1}{\E(P)}{k_{\iM}} \int_{\R} I(u) \E\left[\rule{0mm}{4mm}G_\iP\left(u,{E_{\iM}},{L_{\iM}},{\cal N}_2\right)^2\right] \diff u \\
  =  1{+}{k_{\iP}}\E({L_{\iP}}) \P\left(\rule{0mm}{4mm}F_{{L_{\iM}}}{\ge}|E_{\iP,1}{+}F_{{L_{\iP}},1}{-}(E_{\iP,2}{+}F_{{L_{\iP}},2})|\right)
  \end{multline}

\medskip
We now investigate the term $\E\left(\E(P{\mid}{\cal F}^I)^2\right)$ of Relation~\eqref{vqaux1}. As for the variable $M$ of the last section, observe that
\begin{align*}
\Delta^I(P)&\steq{def}  \E\left(\E(P{\mid}{\cal F}^I)^2\right){-}\E(P)^2 =  \E\left(\left(\E(P{\mid}{\cal F}^I){-}\E(P)\right)^2\right) \\
&=\E\left(\left( {k_{\iM}} \int_{\R} (I(u){-}\delta_+) \E\left( G_{\iP}(u,{E_{\iM}},{L_{\iM}},{\cal N}_2)\right) \diff u\right)^2\right).
\end{align*}
We can rewrite the square of the integral as
\begin{align*}
 & \left(\int_{\R} (I(u){-}\delta_+) \E\left( G_{\iP}(u,{E_{\iM}},{L_{\iM}},{\cal N}_2)\right) \diff u\right)^2\\
  &\phantom{aaaa}=\int_{\R^2} (I(u){-}\delta_+)(I(u'){-}\delta_+) \E\left( G_{\iP}(u,{E_{\iM}},{L_{\iM}},{\cal N}_2)\right)\E\left( G_{\iP}(u',{E_{\iM}},{L_{\iM}},{\cal N}_2)\right) \diff u\diff u'\\
  &\phantom{aaaa}=\int_{\R^2} (I(u){-}\delta_+)(I(u'){-}\delta_+)\\
  &\hspace{2cm}{\times}\int_{\R^2}\P\left(\substack{x{+}E_{\iM,1}{\le}u{<}x{+}E_{\iM,1}{+}L_{\iM,1}\\ E_{\iP,1}{\le} x{<}E_{\iP,1}{+}L_{\iP,1}},\substack{x'{+}E_{\iM,2}{\le}u'{<}x'{+}E_{\iM,2}{+}L_{\iM,2}\\  E_{\iP,2}{\le} x'{<}E_{\iP,2}{+}L_{\iP,2}}\right){k_{\iP}}^2\diff x \diff x'  \diff u\diff u'
\end{align*}
By using Relation~\eqref{CorII} and  the translation invariance of Lebesgue's measure on $\R$, we get 
\begin{align*}
  \E\left(\Delta^I(P)\right) &=  \delta_+(1{-}\delta_+){k_{\iM}}^2 \int_{\R^4}e^{-\Lambda |u{-}u'|}\\
& \hspace{1cm}   {\times}\P\left(\substack{x{+}E_{\iM,1}{\le}u{<}x{+}E_{\iM,1}{+}L_{\iM,1}\\ E_{\iP,1}{\le} x{<}E_{\iP,1}{+}L_{\iP,1}},\substack{x'{+}E_{\iM,2}{\le}u'{<}x'{+}E_{\iM,2}{+}L_{\iM,2}\\  E_{\iP,2}{\le} x'{<}E_{\iP,2}{+}L_{\iP,2}}\right){k_{\iP}}^2\diff x \diff x'  \diff u\diff u'\\
&=  \delta_+(1{-}\delta_+) {k_{\iM}}^2{k_{\iP}}^2 \int_{\R^4}e^{-\Lambda |u{+}x{+}E_{\iM,1}{+}E_{\iP,1}{-}(u'{+}x'{+}E_{\iM,2}{+}E_{\iP,2})|}\\
& \hspace{2.5cm}   \P\left(\substack{0{\le}u{<}L_{\iM,1}\\ 0{\le} x{<}L_{\iP,1}},\substack{0{\le}u'{<}L_{\iM,2}\\  0{\le} x'{<}L_{\iP,2}}\right)\diff x \diff x'  \diff u\diff u',
\end{align*}
which gives the identity
\begin{multline}\label{vqaux2}
  \E\left(\Delta^I(P)\right){=}  \frac{1{-}\delta_+}{\delta_+}\E(P)^2\\{\times} \E\left(e^{-\Lambda |F_{{{L_{\iM}}},1}{+}F_{{{L_{\iP}}},1}{+}E_{\iM,1}{+}E_{\iP,1}{-}(F_{{{L_{\iM}}},2}{+}F_{{{L_{\iP}}},2}{+}E_{\iM,2}{+}E_{\iP,2})|}\right)
\end{multline}
%
By plugging the identities~\eqref{Term1} and~\eqref{vqaux2} into Relation~\eqref{vqaux1}, we obtain the desired result.


\begin{thebibliography}{10}

\bibitem{Anderson}
David~F. Anderson and Thomas~G. Kurtz, \emph{Stochastic analysis of biochemical
  systems}, Mathematical Biosciences Institute Lecture Series. Stochastics in
  Biological Systems, vol.~1, Springer, Cham; MBI Mathematical Biosciences
  Institute, Ohio State University, Columbus, OH, 2015. \MR{3363610}

\bibitem{Berg}
O.~G. Berg, \emph{{A model for the statistical fluctuations of protein numbers
  in a microbial population}}, Journal of theoretical biology \textbf{71}
  (1978), no.~4, 587--603.

\bibitem{Bokes}
Pavol Bokes, John~R. King, Andrew T.~A. Wood, and Matthew Loose, \emph{Exact
  and approximate distributions of protein and {mRNA} levels in the low-copy
  regime of gene expression}, Journal of Mathematical Biology \textbf{64}
  (2012), no.~5, 829--854.

\bibitem{Bremer}
H.~Bremer and P.~Dennis, \emph{Modulation of chemical composition and other
  parameters of the cell at different exponential growth rates}, EcoSal Plus
  (2008).

\bibitem{Brenner}
Sydney Brenner, Fran{\c{c}}ois Jacob, and Matthew Meselson, \emph{An unstable
  intermediate carrying information from genes to ribosomes for protein
  synthesis}, Nature \textbf{190} (1961), no.~4776, 576--581.

\bibitem{Bressloff}
Paul~C. Bressloff, \emph{Stochastic processes in cell biology},
  Interdisciplinary applied mathematics, Springer, Cham [u.a.], 2014.

\bibitem{Burrill}
Devin~R Burrill and Pamela~A Silver, \emph{Making cellular memories}, Cell
  \textbf{140} (2010), no.~1, 13--18.

\bibitem{ChenShi}
H.~Chen, K.~Shiroguchi, H.~Ge, and X.S. Xie, \emph{Genome-wide study of {mRNA}
  degradation and transcript elongation in escherichia coli}, Mol. Syst. Biol.
  \textbf{11} (2015), no.~1, 781, published correction appears in Mol Syst
  Biol. 2015 May;11(5):808.

\bibitem{Dawson}
Donald~A. Dawson, \emph{Measure-valued {M}arkov processes}, \'Ecole d'\'Et\'e
  de Probabilit\'es de Saint-Flour XXI---1991, Lecture Notes in Math., vol.
  1541, Springer, Berlin, 1993, pp.~1--260.

\bibitem{eldar}
Avigdor Eldar and Michael~B Elowitz, \emph{Functional roles for noise in
  genetic circuits}, Nature \textbf{467} (2010), no.~7312, 167.

\bibitem{Feller}
W.~Feller, \emph{An introduction to probability theory and its applications},
  2nd ed., vol.~{II}, John Wiley \& Sons Ltd, New York, 1971.

\bibitem{Friedman}
Nir Friedman, Long Cai, and X~Sunney Xie, \emph{Linking stochastic dynamics to
  population distribution: an analytical framework of gene expression},
  Physical review letters \textbf{97} (2006), no.~16, 168302.

\bibitem{Fromion}
Vincent Fromion, Emanuele Leoncini, and Philippe Robert, \emph{Stochastic gene
  expression in cells: A point process approach}, SIAM Journal on Applied
  Mathematics \textbf{73} (2013), no.~1, 195--211.

\bibitem{Gupta}
Ankit Gupta, Corentin Briat, and Mustafa Khammash, \emph{A scalable
  computational framework for establishing long-term behavior of stochastic
  reaction networks}, PLoS computational biology \textbf{10} (2014), no.~6,
  e1003669.

\bibitem{Jacob}
Fran\c{c}ois Jacob and Jacques Monod, \emph{Genetic regulatory mechanisms in
  the synthesis of proteins}, Journal of Molecular Biology \textbf{3} (1961),
  no.~3, 318 -- 356.

\bibitem{Johnson}
Norman~L. Johnson, Adrienne~W. Kemp, and Samuel Kotz, \emph{Univariate discrete
  distributions}, third ed., Wiley Series in Probability and Statistics,
  Wiley-Interscience [John Wiley \& Sons], Hoboken, NJ, 2005.

\bibitem{BioCyc}
Peter~D Karp, Anamika Kothari, Carol~A Fulcher, Ingrid~M Keseler, Mario
  Latendresse, Markus Krummenacker, Pallavi Subhraveti, Quang Ong, Richard
  Billington, Ron Caspi, Suzanne~M Paley, Wai~Kit Ong, and Peter~E Midford,
  \emph{{The BioCyc collection of microbial genomes and metabolic pathways}},
  Briefings in Bioinformatics (2017).

\bibitem{Kelly}
Frank~P. Kelly, \emph{Reversibility and stochastic networks}, John Wiley \&\
  Sons Ltd., Chichester, 1979, Wiley Series in Probability and Mathematical
  Statistics.

\bibitem{Kingman}
J.~F.~C. Kingman, \emph{Poisson processes}, Oxford studies in probability,
  1993.

\bibitem{Kurtz}
T.G. Kurtz, \emph{Averaging for martingale problems and stochastic
  approximation}, Applied Stochastic Analysis, US-French Workshop, Lecture
  notes in Control and Information sciences, vol. 177, Springer Verlag, 1992,
  pp.~186--209.

\bibitem{Leoncini}
Emanuele Leoncini, \emph{{Towards a global and systemic understanding of
  protein production in prokaryotes}}, {Ph.D.} document, {\'Ecole
  Polytechnique}, December 2013.

\bibitem{Lestas}
Ioannis Lestas, Glenn Vinnicombe, and Johan Paulsson, \emph{Fundamental limits
  on the suppression of molecular fluctuations}, Nature \textbf{467} (2010),
  no.~7312, 174.

\bibitem{Levin}
David~A. Levin, Yuval Peres, and Elizabeth~L. Wilmer, \emph{Markov chains and
  mixing times}, American Mathematical Society, Providence, RI, 2009.

\bibitem{Loynes}
R.M. Loynes, \emph{The stability of queues with non independent inter-arrival
  and service times}, Proc. Cambridge Ph. Soc. \textbf{58} (1962), 497--520.

\bibitem{Mackey}
Michael~C. Mackey, Mois\'{e}s Santill\'{a}n, Marta Tyran-Kami\'{n}ska, and
  Eduardo~S. Zeron, \emph{Simple mathematical models of gene regulatory
  dynamics}, Lecture Notes on Mathematical Modelling in the Life Sciences,
  Springer, Cham, 2016.

\bibitem{Murray}
J.~D. Murray, \emph{Mathematical biology. {I}}, third ed., Interdisciplinary
  Applied Mathematics, vol.~17, Springer-Verlag, New York, 2002, An
  introduction.

\bibitem{Neveu}
Jacques Neveu, \emph{Processus ponctuels}, \'{E}cole d'\'{e}t\'e de
  {P}robabilit\'es de {S}aint-{F}lour (P.-L. Hennequin, ed.), Lecture Notes in
  Mathematics, vol. 598, Springer-Verlag, Berlin, 1977, pp.~249--445.

\bibitem{Norman}
Thomas~M Norman, Nathan~D Lord, Johan Paulsson, and Richard Losick,
  \emph{Memory and modularity in cell-fate decision making}, Nature
  \textbf{503} (2013), no.~7477, 481.

\bibitem{Paulsson}
J.~Paulsson, \emph{{Models of stochastic gene expression}}, Physics of Life
  Reviews \textbf{2} (2005), no.~2, 157--175.

\bibitem{Propp}
James~Gary Propp and David~Bruce Wilson, \emph{Exact sampling with coupled
  {M}arkov chains and applications to statistical mechanics}, Random Structures
  \& Algorithms \textbf{9} (1996), no.~1-2, 223--252.

\bibitem{Raj}
Arjun Raj, Charles~S Peskin, Daniel Tranchina, Diana~Y Vargas, and Sanjay
  Tyagi, \emph{Stochastic {mRNA} synthesis in mammalian cells}, PLoS biology
  \textbf{4} (2006), no.~10, e309.

\bibitem{Rigney}
D.~Rigney and W.~Schieve, \emph{{Stochastic model of linear, continuous protein
  synthesis in bacterial populations}}, Journal of Theoretical Biology
  \textbf{69} (1977), no.~4, 761--766.

\bibitem{Robert}
Philippe Robert, \emph{Stochastic networks and queues}, Stochastic Modelling
  and Applied Probability Series, Springer-Verlag, New York, 2003.

\bibitem{Russell}
James~B Russell and Gregory~M Cook, \emph{Energetics of bacterial growth:
  balance of anabolic and catabolic reactions.}, Microbiol. Mol. Biol. Rev.
  \textbf{59} (1995), no.~1, 48--62.

\bibitem{Swain2}
Vahid Shahrezaei and Peter~S. Swain, \emph{Analytical distributions for
  stochastic gene expression}, Proceedings of the National Academy of Sciences
  \textbf{105} (2008), no.~45, 17256--17261.

\bibitem{Swain}
Peter~S. Swain, Michael~B. Elowitz, and Eric~D. Siggia, \emph{Intrinsic and
  extrinsic contributions to stochasticity in gene expression}, Proceedings of
  the National Academy of Sciences \textbf{99} (2002), no.~20, 12795--12800.

\bibitem{Taniguchi}
Yuichi Taniguchi, Paul~J. Choi, Gene-Wei Li, Huiyi Chen, Mohan Babu, Jeremy
  Hearn, Andrew Emili, and X.~Sunney Xie, \emph{Quantifying e. coli proteome
  and transcriptome with single-molecule sensitivity in single cells}, Science
  \textbf{329} (2010), no.~5991, 533--538.

\bibitem{Thattai}
Mukund Thattai and Alexander van Oudenaarden, \emph{Intrinsic noise in gene
  regulatory networks}, Proceedings of the National Academy of Sciences
  \textbf{98} (2001), no.~15, 8614--8619.

\bibitem{Valgepea}
Kaspar Valgepea, Kaarel Adamberg, Andrus Seiman, and Raivo Vilu,
  \emph{Escherichia coli achieves faster growth by increasing catalytic and
  translation rates of proteins}, Molecular BioSystems \textbf{9} (2013),
  no.~9, 2344--2358.

\bibitem{vanKampen}
Nicolaas~Godfried Van~Kampen, \emph{Stochastic processes in physics and
  chemistry}, vol.~1, Elsevier, 1992.

\bibitem{Watson}
J.~D. Watson, T.~A. Baker, S.~P. Bell, A.~Gann, M.~Levine, R.~Losick, and
  I.~CSHLP, \emph{Molecular biology of the gene}, 6th ed., Pearson/Benjamin
  Cummings; Cold Spring Harbor Laboratory Press, San Francisco; Cold Spring
  Harbor, N.Y., December 2007.

\bibitem{Crick}
James~D Watson, Francis~HC Crick, et~al., \emph{Molecular structure of nucleic
  acids}, Nature \textbf{171} (1953), no.~4356, 737--738.

\bibitem{Whittaker}
E.~T. Whittaker and G.~N. Watson, \emph{A course of modern analysis}, Cambridge
  Mathematical Library, Cambridge University Press, Cambridge, 1996, Reprint of
  the fourth (1927) edition.

\bibitem{Williams}
David Williams, \emph{Probability with martingales}, Cambridge University
  Press, Cambridge, 1991.

\end{thebibliography}
\end{document}